\documentclass[twoside,11pt,a4paper]{article}

\usepackage[utf8]{inputenc}

\usepackage[margin=23mm]{geometry}
\usepackage{natbib}
\bibpunct{(}{)}{;}{a}{,}{,}

\usepackage{theorem}
\usepackage{amsfonts}
\usepackage{amsmath}
\usepackage{amssymb}
\usepackage{bbm}

\usepackage{graphicx}
\usepackage{tikz}
\usetikzlibrary{arrows}
\usepackage{subfigure}
\usepackage{url}

\usepackage{enumitem}
\usepackage{units}
\usepackage[boxed, vlined, linesnumberedhidden]{algorithm2e}

\newcommand{\BlackBox}{\rule{1.5ex}{1.5ex}}  
\newenvironment{proof}[1][]{\par\noindent{\bf Proof\ifthenelse{\equal{#1}{}}{}{ of #1}\ }}{\hfill\BlackBox\\[2mm]}
 
\newtheorem{theorem}{Theorem}
\newtheorem{lemma}[theorem]{Lemma} 
\newtheorem{proposition}[theorem]{Proposition} 

\newtheorem{corollary}[theorem]{Corollary}
{\theorembodyfont{\rmfamily} \newtheorem{definition}[theorem]{Definition}}

\newenvironment{subprop}{\begin{enumerate}[label=(\roman*), ref=(\roman*), noitemsep]}{\end{enumerate}}
\newcommand{\proofitem}[1]{\vspace{3mm}\par\noindent\emph{#1:}}
\newcommand{\proofnegspace}{\vspace{-8mm}}
\SetCommentSty{textit}
\SetAlFnt{\small}
\SetAlTitleFnt{textnormal}

\newcommand{\spforall}{\; \forall \;}
\newcommand{\spexists}{\; \exists \;}
\newcommand{\spst}{\;|\;}
\newcommand{\transp}{\mathrm{T}}

\newcommand{\mcI}{\ensuremath{\mathcal{I}}}

\newcommand{\smallbullet}{\tikz[baseline]\draw[fill=black] (0, 0.4ex) circle (0.15ex);}
\newcommand{\sR}{\mathbb{R}}

\newcommand{\supscr}[1]{\ensuremath{^{\mathrm{#1}}}}
\newcommand{\subscr}[1]{\ensuremath{_{\mathrm{#1}}}}
\newcommand{\LexBFS}{\ensuremath{\mathsc{LexBFS}}}

\DeclareMathOperator{\Cov}{Cov}
\DeclareMathOperator{\pa}{pa}
\DeclareMathOperator{\nb}{ne}
\DeclareMathOperator{\ch}{ch}
\DeclareMathOperator{\ad}{ad}

\DeclareMathOperator{\diag}{diag}
\DeclareMathOperator{\doop}{do}

\DeclareMathOperator{\SHD}{SHD}
\DeclareMathOperator*{\argmax}{arg\,max}
\DeclareMathAlphabet{\mathsc}{OT1}{cmr}{m}{sc}
\DeclareMathAlphabet{\mathcal}{OMS}{cmsy}{m}{n}

\newlength{\edgelength}
\DeclareRobustCommand{\grarright}{\mathbin{\tikz[baseline] \draw[->] (0pt, 0.7ex) -- (\edgelength, 0.7ex);}}
\DeclareRobustCommand{\grarleft}{\mathbin{\tikz[baseline] \draw[<-] (0pt, 0.7ex) -- (\edgelength, 0.7ex);}}
\DeclareRobustCommand{\grline}{\mathbin{\tikz[baseline] \draw[-] (0pt, 0.7ex) -- (\edgelength, 0.7ex);}}
\DeclareRobustCommand{\grdots}{\mathbin{\tikz[baseline] \draw[dotted, thick] (0pt, 0.7ex) -- (\edgelength, 0.7ex);}}

\newcommand{\threegraph}[6]{%
    \begin{tikzpicture}[baseline=(one.base)]
        \node[anchor=base east] (one) at (0, 0) {#1};
        \node[anchor=base west] (two) at (1.6, 0) {#3};
        \node[anchor=base] (three) at (0.8, -0.7) {#5};
        \ifthenelse{\equal{#2}{}}{}{\draw[#2] (one.mid east) -- (two.mid west);}
        \ifthenelse{\equal{#4}{}}{}{\draw[#4] (two) -- (three);}
        \ifthenelse{\equal{#6}{}}{}{\draw[#6] (three) -- (one);}
    \end{tikzpicture}
}

\newlength{\exgredge}
\newenvironment{exgraphpicture}{%
    \begin{tikzpicture}[baseline=(v1.base)]
        \node[anchor=mid] (v1) at (0, 0) {$1$};
        \node[anchor=mid] (v2) at (\exgredge,  0) {$2$};
        \node[anchor=mid] (v3) at (2\exgredge, 0) {$3$};
        \node[anchor=mid] (v4) at (3\exgredge, 0) {$4$};
        \node[anchor=mid] (v5) at (\exgredge,  -\exgredge) {$5$};
        \node[anchor=mid] (v6) at (2\exgredge, -\exgredge) {$6$};
        \node[anchor=mid] (v7) at (3\exgredge, -\exgredge) {$7$};
}{\end{tikzpicture}}
\newenvironment{exgraphpicturesmall}{%
    \begin{tikzpicture}[baseline=(v1.base)]
        \node[anchor=mid] (v1) at (0, 0) {$1$};
        \node[anchor=mid] (v2) at (\exgredge,  0) {$2$};
        \node[anchor=mid] (v3) at (2\exgredge, 0) {$3$};
        \node[anchor=mid] (v4) at (0,  -\exgredge) {$4$};
        \node[anchor=mid] (v5) at (\exgredge, -\exgredge) {$5$};
}{\end{tikzpicture}}

\title{Characterization and Greedy Learning of Interventional Markov Equivalence Classes of Directed Acyclic Graphs}
\author{Alain Hauser and Peter Bühlmann\\
{\small\texttt{\{hauser, buhlmann\}@stat.math.ethz.ch}} \\
       Seminar für Statistik\\
       ETH Zürich\\
       8092 Zürich, Switzerland}
\date{}

\begin{document}

\maketitle

\begin{abstract}
The investigation of directed acyclic graphs (DAGs) encoding the same Markov property, that is the same conditional independence relations of multivariate observational distributions, has a long tradition; many algorithms exist for model selection and structure learning in Markov equivalence classes.  In this paper, we extend the notion of Markov equivalence of DAGs to the case of interventional distributions arising from \emph{multiple} intervention experiments.  We show that under reasonable assumptions on the intervention experiments, interventional Markov equivalence defines a finer partitioning of DAGs than observational Markov equivalence and hence improves the identifiability of causal models.  We give a graph theoretic criterion for two DAGs being Markov equivalent under interventions and show that each interventional Markov equivalence class can, analogously to the observational case, be uniquely represented by a chain graph called \emph{interventional essential graph} (also known as \emph{CPDAG} in the observational case).  These are key insights for deriving a generalization of the Greedy Equivalence Search algorithm aimed at structure learning from interventional data.  This new algorithm is evaluated in a simulation study.
\end{abstract}

\section{Introduction}
\label{sec:introduction}

Directed acyclic graphs (or DAGs for short) are commonly used to model causal relationships between random variables; in such models, parents of some vertex in the graph are understood as ``causes'', and edges have the meaning of ``causal influences''.  The causal influences between random variables imply conditional independence relations among them.  However, those independence relations, or the corresponding Markov properties, do \emph{not} identify the corresponding DAG completely, but only up to Markov equivalence.  To put it simple, the skeleton of an underlying DAG is completely determined by its Markov property, whereas the \emph{direction} of the arrows (which is crucial for causal interpretation) is in general not encoded in the Markov property for the observational distribution.

Interventions can help to overcome those limitations in identifiability.  An \emph{intervention} is realized by forcing the value of one or several random variables of the system to chosen values, destroying their original causal dependencies.  The ensemble of both the observational and interventional distributions can greatly improve the identifiability of the causal structure of the system, the underlying DAG.

This paper has two main contributions.  The first one is an algorithmically tractable graphical representation of Markov equivalence classes under a given set of interventions (possibly affecting several variables) from which the identifiability of causal models can be read off.  This is of general interest for computation and algorithms dealing with structure (DAG) learning from an ensemble of observational and interventional data such as MCMC.  The second contribution is a generalization of the Greedy Equivalence Search (GES) algorithm of \citet{Chickering2002Optimal}, yielding an algorithm called Greedy Interventional Equivalence Search (GIES) which can be used for regularized maximum likelihood estimation in such an interventional setting.

In Section \ref{sec:model}, we establish a criterion for two DAGs being Markov equivalent under a given intervention setting.  We then generalize the concept of essential graphs, a graph theoretic representation of Markov equivalence classes, to the interventional case and characterize the properties of those graphs in Section \ref{sec:essential-graphs}.  In Section \ref{sec:greedy-search}, we elaborate a set of algorithmic operations to efficiently traverse the search space of interventional essential graphs and finally present the GIES algorithm.  An experimental evaluation thereof is given in Section \ref{sec:evaluation}.  We postpone all proofs to Appendix \ref{sec:proofs}, while Appendix \ref{sec:graphs} contains a review on graph theoretic concepts and definitions.  An implementation of the GIES algorithm will be available in the next release of the R package \texttt{pcalg} \citep{Kalisch2012Causal}; meanwhile, a prerelease version is available upon request from the first author.

\subsection{Related Work}
\label{sec:related-work}

The investigation of Markov equivalence classes of directed graphical models has a long tradition, perhaps starting with the criterion for two DAGs being Markov equivalent by \citet{Verma1990Equivalence} and culminating in the graph theoretic characterization of essential graphs (also called \emph{CPDAGs}, ``completed partially directed acyclic graphs'') representing Markov equivalence classes by \citet{Andersson1997Characterization}.  Several algorithms for estimating essential graphs from observational data exist, such as the PC algorithm \citep{Spirtes2000Causation} or the Greedy Equivalence Search (GES) algorithm \citep{Meek1997Graphical, Chickering2002Optimal}; a more complete overview is given in \citet{Brown2005Comparison} and \citet{Murphy2001Bayes}.

Different approaches to incorporate interventional data for learning causal models have been developed in the past.  The Bayesian procedures of \citet{Cooper1999Causal} or 
\citet{Eaton2007Exact} address the problem of calculating a posterior (and also a likelihood) of an ensemble of observational and interventional data but do not address questions of identifiability or Markov equivalence: allowing different posteriors for Markov equivalent models can be intended in Bayesian methods (and realized by giving the corresponding models different priors).  Since the number of DAGs with $p$ variables grows super-exponentially with $p$ \citep{Robinson1973Counting}, the computation of a \emph{full} posterior is intractable.  For this reason, the mentioned Bayesian approaches are limited to computing posterior probabilities for certain features of a DAG; such a feature could be an edge from a vertex $a$ to another vertex $b$, or a directed path from $a$ to $b$ visiting additional vertices.  Approaches based on active learning \citep{He2008Active, Tong2001Active, Eberhardt2008Almost} propose an iterative line of action, estimating the essential graph with observational data in a first step and using interventional data in a second step to orient beforehand unorientable edges.  \citet{He2008Active} present a greedy procedure in which interventional data is uniquely used for deciding about edge orientations; this is not favorable from a statistical point of view since interventional data can also help to improve the estimation of the skeleton (or, more generally, the observational essential graph).  \citet{Tong2001Active} avoid this problem by using a Bayesian framework, but do not address the issue of Markov equivalence therewith.  \citet{Eberhardt2005Number} and \citet{Eberhardt2008Almost} provide algorithms for choosing intervention targets that \emph{completely} identify \emph{all} causal models of $p$ variables uniformly, but neither address the question of partial identifiability under a limited number of interventions nor provide an algorithm for learning the causal structure from data.  \citet{Eberhardt2010Combining} present an algorithm for learning \emph{cyclic} linear causal models, but focus on complete identifiability; identifiability results for cyclic models only imply \emph{sufficient}, but not \emph{necessary}, conditions for the identifiability of acyclic models.

Probably the most advanced result concerning identifiability of causal models under single-variable interventions so far is given in the work of \citet{Tian2001CausalA}.  Although they do not provide a characterization of equivalence classes as a whole (as this paper does), they present a necessary and sufficient graph theoretic criterion for two models being indistinguishable under a set of single-variable interventions as well as a learning algorithm based on the detection of changes in marginal distributions.

\section{Model}
\label{sec:model}

We consider $p$ random variables $(X_1, \ldots, X_p) =: X$ which take values in some product measure space $(\mathcal{X}, \mathcal{A}, \mu) = (\prod_{i=1}^p \mathcal{X}_i, \bigotimes_{i=1}^p \mathcal{A}_i, \bigotimes_{i=1}^p \mu_i)$ with $\mathcal{X}_i \subset \sR \spforall i$.  Each $\sigma$-algebra $\mathcal{A}_i$ is assumed to contain at least two disjoint sets of positive measure to avoid pathologies, and $X$ is assumed to have a strictly positive joint density w.r.t.\ the measure $\mu$ on $\mathcal{X}$.  We denote the set of all positive densities on $\mathcal{X}$ by $\mathcal{M}$.  For any subset of component indices $A \subset [p] := \{1, \ldots, p\}$, we use the notation $\mathcal{X}_A := \prod_{a \in A} \mathcal{X}_a$, $X_A := (X_a)_{a \in A}$ and the convention $X_\emptyset \equiv 0$.  Lowercase symbols like $x_A$ represent a value in $\mathcal{X}_A$.

The model we are considering is built upon Markov properties with respect to DAGs.  By convention, all graphs appearing in the paper shall have the vertex set $[p]$, representing the $p$ random variables $X_1, \ldots, X_p$.  Our notation and definitions related to graphs are summarized in Appendix \ref{sec:graphs-notation}.

\subsection{Causal Calculus: A Short Review}
\label{sec:causal-calculus}

We start by summarizing important facts and fixing our notation concerning Markov properties and intervention calculus.

\begin{definition}[Markov property; \citealp{Lauritzen1996Graphical}]
    \label{def:directed-markov-property}
    Let $D$ be a DAG.  Then we say that a probability density $f \in \mathcal{M}$ \textbf{obeys the Markov property of $D$} if $f(x) = \prod_{i = 1}^p f(x_i | x_{\pa_D(i)})$.  The set of all positive densities obeying the Markov property of $D$ is denoted by $\mathcal{M}(D)$.
\end{definition}

Definition \ref{def:directed-markov-property} is the most straightforward translation of independence relations induced from structural equations, the historical origin of directed graphical models \citep{Wright1921Correlation}.  Related notions like local and global Markov properties exist and are equivalent to the factorization property of Definition \ref{def:directed-markov-property} for positive densities \citep{Lauritzen1996Graphical}.

\begin{definition}[Markov equivalence; \citealp{Andersson1997Characterization}]
    \label{def:markov-equivalence}
    Let $D_1$ and $D_2$ be two DAGs.  $D_1$ and $D_2$ are called \textbf{Markov equivalent} (notation: $D_1 \sim D_2$) if $\mathcal{M}(D_1) = \mathcal{M}(D_2)$.
\end{definition}

\begin{theorem}[\citealp{Verma1990Equivalence}]
    \label{thm:markov-equivalence}
    Two DAGs $D_1$ and $D_2$ are Markov-equivalent if and only if they have the same skeleton and the same v-structures.
\end{theorem}

Directed graphical models allow for an obvious causal interpretation.  For a density $f$ that obeys the Markov properties of some DAG $D$, we can think of a random variable $X_a$ being the \emph{direct cause} of another variable $X_b$ if $a$ is a parent of $b$ in $D$.

\begin{definition}[Causal model]
    \label{def:causal-model}
    A \textbf{causal model} is a pair $(D, f)$, where $D$ is a DAG on the vertex set $[p]$ and $f \in \mathcal{M}(D)$ is a density obeying the Markov property of $D$:  $D$ is called the \textbf{causal structure} of the model, and $f$ the \textbf{observational density}.
\end{definition}

Causality is strongly linked to interventions.  We consider \textbf{stochastic interventions} \citep{Korb2004Varieties} modeling the effect of setting or forcing one or several random variables $X_I$, where $I \subset [p]$ is called the \textbf{intervention target}, to the value of \emph{independent} random variables $U_I$, called \textbf{intervention variables}.  The joint product density of $U_I$ on $\mathcal{X}_I$, called \textbf{level density}, is denoted by $\tilde{f}$.  Extending the $\doop()$ operator \citep{Pearl1995Causal} to stochastic interventions, we denote the density of $X$ under such an intervention by $f(x | \doop_D(X_I = U_I))$.  Using truncated factorization and the assumption of independent intervention variables, this \textbf{interventional density} can be written as
\begin{equation}
    f(x \spst \doop_D(X_I = U_I)) = \prod_{i \notin I} f(x_i | x_{\pa_D(i)}) \prod_{i \in I} \tilde{f}(x_i)\ . \label{eqn:interventional-density}
\end{equation}
By denoting with $I = \emptyset$ and using the convention $f(x | \doop(X_\emptyset = U_\emptyset)) = f(x)$, we also encompass the observational case as an intervention target.

\begin{definition}[Intervention graph]
    \label{def:intervention-graph}
    Let $D = ([p], E)$ be a DAG with vertex set $[p]$ and edge set $E$ (see Appendix \ref{sec:graphs-notation}), and $I \subset [p]$ an intervention target.  The \textbf{intervention graph} of $D$ is the DAG $D^{(I)} = ([p], E^{(I)})$, where $E^{(I)} := \{(a, b) \; | \; (a, b) \in E, b \notin I\}$.
\end{definition}
For a causal model $(D, f)$, an interventional density $f(\cdot | \doop_D(X_I = U_I))$ obeys the Markov property of $D^{(I)}$: the Markov property of the observational density is inherited.  Figure \ref{fig:ex-intervention-dags} shows an example of a DAG and two corresponding intervention graphs.

\begin{figure}[b]
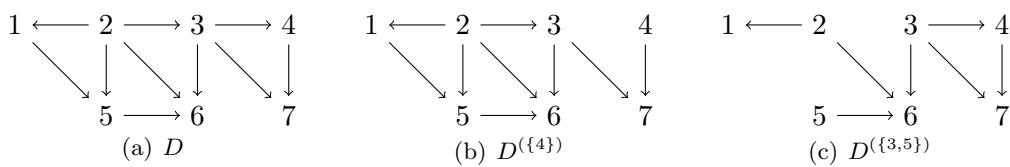

    \centering
    \subfigure[$D$ \label{fig:ex-observational-dag}]{%
        \begin{exgraphpicture}
            \draw[->] (v2) -- (v1);
            \draw[->] (v2) -- (v3);
            \draw[->] (v3) -- (v4);
            \draw[->] (v1) -- (v5);
            \draw[->] (v2) -- (v5);
            \draw[->] (v2) -- (v6);
            \draw[->] (v3) -- (v6);
            \draw[->] (v5) -- (v6);
            \draw[->] (v3) -- (v7);
            \draw[->] (v4) -- (v7);
        \end{exgraphpicture}
    } \quad
    \subfigure[$D^{(\{4\})}$ \label{fig:ex-intervention-4}]{%
        \begin{exgraphpicture}
            \draw[->] (v2) -- (v1);
            \draw[->] (v2) -- (v3);
            \draw[->] (v1) -- (v5);
            \draw[->] (v2) -- (v5);
            \draw[->] (v2) -- (v6);
            \draw[->] (v3) -- (v6);
            \draw[->] (v5) -- (v6);
            \draw[->] (v3) -- (v7);
            \draw[->] (v4) -- (v7);
        \end{exgraphpicture}
    } \quad
    \subfigure[$D^{(\{3, 5\})}$]{%
        \begin{exgraphpicture}
            \draw[->] (v2) -- (v1);
            \draw[->] (v3) -- (v4);
            \draw[->] (v2) -- (v6);
            \draw[->] (v3) -- (v6);
            \draw[->] (v5) -- (v6);
            \draw[->] (v3) -- (v7);
            \draw[->] (v4) -- (v7);
        \end{exgraphpicture}
    }
    \caption{A DAG $D$ and the corresponding intervention graphs $D^{(\{4\})}$ and $D^{(\{3, 5\})}$.}
    \label{fig:ex-intervention-dags}
\end{figure}

As foreshadowed in the introduction, we are interested in causal inference based on data sets originating from \emph{multiple} interventions, that means from a set of the form $\mathcal{S} = \{(I_j, \tilde{f}_j)\}_{j=1}^J$, where $I_j \subset [p]$ is an intervention target and $\tilde{f}_j$ a level density on $\mathcal{X}_{I_j}$ for $1 \leq j \leq J$.  We call such a set an \textbf{intervention setting}, and the corresponding (multi)set of intervention targets $\mathcal{I} = \{I_j\}_{j=1}^J$ a \textbf{family of targets}.  We often use the family of targets as an index set, for example to write a corresponding intervention setting as $\mathcal{S} = \{(I, \tilde{f}_I)\}_{I \in \mcI}$.

We consider \textbf{interventional data} of sample size $n$ produced by a causal model $(D, f)$ under an intervention setting $\mathcal{S} = \{(I, \tilde{f}_I)\}_{I \in \mcI}$.  We assume that the $n$ samples $X^{(1)}, \ldots, X^{(n)}$ are independent, and write them as usual as rows of a \textbf{data matrix} $\mathbf{X}$.  However, they are \emph{not} identically distributed as they arise from \emph{different} interventions.  The interventional data set is fully specified by the pair $(\mathcal{T}, \mathbf{X})$,
\begin{equation}
    \mathcal{T} = \left( \begin{array}{c} T^{(1)} \\ \vdots \\ T^{(n)} \end{array} \right) \in \mcI^n, \quad \mathbf{X} = \left( \begin{array}{c} \text{\textemdash } \, X^{(1)} \text{ \textemdash} \\ \vdots \\ \text{\textemdash } \, X^{(n)} \text{ \textemdash} \end{array} \right) \ , \label{eqn:dataset}
\end{equation}
where for each $i \in [n]$, $T^{(i)}$ denotes the intervention target under which the sample $X^{(i)}$ was produced.  This data set can potentially contain observational data as well, namely if $\emptyset \in \mcI$.  To summarize, we consider the statistical model
\begin{align}
    & X^{(1)}, X^{(2)}, \ldots, X^{(n)} \text{ independent,} \nonumber \\
    & X^{(i)} \sim f\big(\cdot \spst \doop_D(X^{(i)}_{T^{(i)}} = U_{T^{(i)}}) \big), \ U_{T^{(i)}} \sim \tilde{f}_{T^{(i)}}, \quad i = 1, \ldots, n \ , \label{eqn:sample-density}
\end{align}
and we assume that each target $I \in \mcI$ appears at least once in the sequence $\mathcal{T}$.

\subsection{Interventional Markov Equivalence: New Concepts and Results}
\label{sec:interventional-markov-properties}

An intervention at some target $a \in [p]$ destroys the original causal influence of other variables of the system on $X_a$.  Interventional data thereof can hence not be used to determine the causal parents of $X_a$ in the (undisturbed) system.  To be able to estimate at least the complete skeleton of a causal structure (as in the observational case), an intervention experiment has to be performed based on a \emph{conservative} family of targets:
\begin{definition}[Conservative family of targets]
    \label{def:conservative-target-family}
    A family of targets $\mathcal{I}$ is called \textbf{conservative} if for all $a \in [p]$, there is some $I \in \mathcal{I}$ such that $a \notin I$.
\end{definition}
In this paper, we restrict our considerations to \emph{conservative} families of targets; see Section \ref{sec:model-discussion} for a more detailed discussion.  Note that every experiment in which we also measure observational data corresponds to a conservative family of targets.

If a family of targets \mcI{} contains more than one target, interventional data as in Equation (\ref{eqn:sample-density}) are \emph{not} identically distributed.  Whereas the distribution of observational data is determined by a \emph{single} density, we need \emph{tuples} of densities as in the following definition to specify the distribution of interventional data.
\begin{definition}
    \label{def:intervention-densities}
    Let $D$ be a DAG on $[p]$, and let $\mathcal{I}$ be a family of targets.  Then we define
    \begin{align*}
        \mathcal{M_I}(D) := & \big\{(f^{(I)})_{I \in \mcI} \in \mathcal{M}^{|\mcI|} \,\big| \spforall I \in \mcI: f^{(I)} \in \mathcal{M}(D^{(I)}), \text{ and} \\
        & \spforall I, J \in \mcI, \spforall a \notin I \cup J: f^{(I)}(x_a | x_{\pa_D(a)}) = f^{(J)}(x_a | x_{\pa_D(a)}) \big\}\ .
    \end{align*}
\end{definition}
Although the $\doop()$ operator does not appear in Definition \ref{def:intervention-densities}, the elements in $\mathcal{M_I}(D)$ are exactly the tuples $(f(\cdot | \doop_D(X_I = U_I)))_{I \in \mcI}$ that can be realized as interventional densities of some causal model $(D, f)$.  The first condition in the definition reflects the fact that an intervention at a target $I$ generates a density obeying the Markov property of $D^{(I)}$; the second condition is a consequence of the truncated factorization in Equation (\ref{eqn:interventional-density}).  These considerations are formalized in the following lemma and motivate Definition \ref{def:interventional-markov-equivalence} of interventional Markov equivalence in analogy to the observational case.  Note that for $\mcI = \{\emptyset\}$, Definition \ref{def:intervention-densities} equals its observational counterpart: $\mathcal{M}_{\{\emptyset\}}(D) = \mathcal{M}(D)$ (see Definition \ref{def:directed-markov-property}).

\begin{lemma}
    \label{lem:intervention-densities-motivation}
    Let $D$ be a DAG on $[p]$, and $\mcI$ a conservative family of targets.
    \begin{subprop}
        \item \label{itm:interventions-meet-definition} Let $(D, f)$ be a causal model (that is, $f \in \mathcal{M}(D)$), $\mathcal{S} = \{(I, \tilde{f}_I)\}_{I \in \mcI}$ an intervention setting and $U_I \sim \tilde{f}_I$ intervention variables for $I \in \mcI$.  Then, we have
        $$
            \big(f(\cdot \spst \doop(X_I = U_I))\big)_{I \in \mcI} \in \mathcal{M_I}(D)\ .
        $$
        
        \item \label{itm:interventions-realize-definition} Let $(f^{(I)})_{I \in \mcI} \in \mathcal{M_I}(D)$.  Then there is some positive density $f \in \mathcal{M}(D)$ and an intervention setting $\mathcal{S} = \{(I, \tilde{f}_I)\}_{I \in \mcI}$ such that $f(\cdot|\doop(X_I = U_I)) = f^{(I)}(\cdot)$ for random variables $U_I$ with density $\tilde{f}_I$, for all $I \in \mcI$.
    \end{subprop}
\end{lemma}

\begin{definition}[Interventional Markov equivalence]
    \label{def:interventional-markov-equivalence}
    Let $D_1$ and $D_2$ be DAGs, and \mcI{} a family of targets.  $D_1$ and $D_2$ are called \textbf{\mcI-Markov equivalent} (notation: $D_1 \sim_\mcI D_2$) if $\mathcal{M_I}(D_1) = \mathcal{M_I}(D_2)$.  The \mcI-Markov equivalence class of a DAG $D$ is denoted by $[D]_\mcI$.
\end{definition}
Alternatively, we will also use the term ``interventionally Markov equivalent'' when it is clear which family of targets is meant.  For the simplest conservative family of targets, $\mcI = \{\emptyset\}$, we get back Definition \ref{def:markov-equivalence} for the observational case.  We now generalize Theorem \ref{thm:markov-equivalence} for the interventional case in order to get a purely graph theoretic criterion for interventional Markov equivalence of two given DAGs, the main result of this section.

\begin{theorem}
    \label{thm:interventional-markov-equivalence}
    Let $D_1$ and $D_2$ be two DAGs on $[p]$, and \mcI{} a conservative family of targets.  Then, the following statements are equivalent:
    \begin{subprop}
        \item \label{itm:interventional-markov-equivalence} $D_1 \sim_\mcI D_2$;
        \item \label{itm:interventional-densities-equal} for all $I \in \mcI$, $D_1^{(I)} \sim D_2^{(I)}$ (in the observational sense);
        \item \label{itm:interventional-dags-equivalent} for all $I \in \mcI$, $D_1^{(I)}$ and $D_2^{(I)}$ have the same skeleton and the same v-structures;
        \item \label{itm:skeleton-v-structures} $D_1$ and $D_2$ have the same skeleton and the same v-structures, and $D_1^{(I)}$ and $D_2^{(I)}$ have the same skeleton for all $I \in \mcI$.
    \end{subprop}
\end{theorem}

\subsection{Discussion}
\label{sec:model-discussion}

Throughout this paper, we always assume the observational density $f$ of a causal model to be strictly positive.  This assumption makes sure that the conditional densities in Equation (\ref{eqn:interventional-density}) are well-defined.  The requirement of a strictly positive density can, however, be a restriction for example for discrete models (where the density is with respect to the counting measure).  In the observational case, the notion of Markov equivalence remains the same when we also allow densities that are not strictly positive \citep{Lauritzen1996Graphical}.  We conjecture that the notion of interventional Markov equivalence (Definition \ref{def:interventional-markov-equivalence} and Theorem \ref{thm:interventional-markov-equivalence}) also remains valid for such densities; corresponding proofs would, however, require more caution to avoid the aforementioned problems with (truncated) factorization.

To illustrate the importance of a conservative family of targets for structure identification, let us consider the simplest non-trivial example of a causal model with 2 variables $X_1$ and $X_2$.  Under observational data, we can distinguish two Markov equivalence classes: one in which the variables are independent (represented by the empty DAG $D_0$), and one in which they are not independent (represented by the DAGs $D_1 := 1 \grarright 2$ and $D_2 := 1 \grarleft 2$).  $D_1$ and $D_2$ can be distinguished if we can measure data from an intervention at one of the vertices in addition to observational data; this experimental setting corresponds to the (conservative) family of targets $\mcI = \{\emptyset, \{1\}\}$.  However, an intervention at, say, $X_1$ \emph{alone} (that is, in the \emph{absence} of observational data), corresponding to the non-conservative family $\mcI = \{\{1\}\}$, only allows a distinction between the models $D_2$ and $D_0$ on the one hand (which do not show dependence between $X_1$ and $X_2$ under the intervention) and $D_1$ on the other hand (which does show dependence between $X_1$ and $X_2$ under the intervention).  Note that the two indistinguishable models $D_0$ and $D_2$ do not even have the same skeleton, and that it is impossible to determine the influence of $X_2$ on $X_1$ in the undisturbed system.  In this setting, it would be more natural to consider the intervened variable $X_1$ as an external parameter rather than a random variable of the system, and to perform regression to detect or determine the influence of $X_1$ on $X_2$.  Note, however, that full identifiability of the models does \emph{not} require observational data; interventions at $X_1$ and $X_2$ (corresponding to the conservative family $\mcI = \{\{1\}, \{2\}\}$ in our notation) are also sufficient.

Theorem \ref{thm:interventional-markov-equivalence} is of great importance for the description of Markov equivalence classes under interventions.  It shows that two DAGs which are interventionally Markov equivalent under some conservative family of targets are also observationally Markov equivalent:
\begin{equation}
    \label{eqn:equivalence-implication}
    D_1 \sim_\mathcal{I} D_2 \Rightarrow D_1 \sim D_2.
\end{equation}
This implication is \emph{not} true anymore for non-conservative families of targets.  This is an explanation for the term ``conservative'': a conservative family of targets yields a finer partitioning of DAGs into equivalence classes compared to observational Markov equivalence, but it preserves the ``borders'' of observational Markov equivalence classes.  Figure \ref{fig:ex-equivalent-dags} shows three DAGs that are observationally Markov equivalent, but which fall into two different interventional Markov equivalence classes under the family of targets $\mcI = \{\emptyset, \{4\}\}$.

Theorem \ref{thm:interventional-markov-equivalence} agrees with Theorem 3 of \citet{Tian2001CausalA} for single-variable interventions.  While we also make a statement about interventions at \emph{several} variables, they prove their theorem for perturbations of the system at single variables only, but for a wider class of perturbations called \textbf{mechanism changes} that go beyond our notion of interventions.  While an \emph{intervention} destroys the causal dependence of a variable from its parents (and hence replaces a conditional density by a marginal one in the Markov factorization, see Equation (\ref{eqn:interventional-density})), a \emph{mechanism change} (also known as ``imperfect'' or ``soft'' interventions; see \citealp{Eaton2007Exact}) alters the functional form of this dependence (and hence replaces a Markov factor by a different one which is still a conditional distribution).  The fact that Theorem \ref{thm:interventional-markov-equivalence} is true for mechanism changes on single variables motivates the conjecture that it also holds for mechanism changes on \emph{several} variables.

\begin{figure}[t]
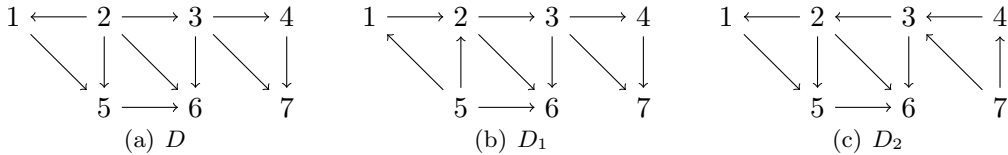

    \centering
    \subfigure[$D$]{%
        \begin{exgraphpicture}
            \draw[->] (v2) -- (v1);
            \draw[->] (v2) -- (v3);
            \draw[->] (v3) -- (v4);
            \draw[->] (v1) -- (v5);
            \draw[->] (v2) -- (v5);
            \draw[->] (v2) -- (v6);
            \draw[->] (v3) -- (v6);
            \draw[->] (v5) -- (v6);
            \draw[->] (v3) -- (v7);
            \draw[->] (v4) -- (v7);
        \end{exgraphpicture}
    } \quad
    \subfigure[$D_1$]{%
        \begin{exgraphpicture}
            \draw[->] (v5) -- (v1);
            \draw[->] (v1) -- (v2);
            \draw[->] (v5) -- (v2);
            \draw[->] (v2) -- (v3);
            \draw[->] (v3) -- (v4);
            \draw[->] (v2) -- (v6);
            \draw[->] (v3) -- (v6);
            \draw[->] (v5) -- (v6);
            \draw[->] (v3) -- (v7);
            \draw[->] (v4) -- (v7);
        \end{exgraphpicture}
    } \quad
    \subfigure[$D_2$]{%
        \begin{exgraphpicture}
            \draw[->] (v2) -- (v1);
            \draw[->] (v3) -- (v2);
            \draw[->] (v4) -- (v3);
            \draw[->] (v7) -- (v3);
            \draw[->] (v7) -- (v4);
            \draw[->] (v1) -- (v5);
            \draw[->] (v2) -- (v5);
            \draw[->] (v2) -- (v6);
            \draw[->] (v3) -- (v6);
            \draw[->] (v5) -- (v6);
        \end{exgraphpicture}
    }
    \caption{Three DAGs having equal skeletons and a single v-structure, $3 \grarright 6 \grarleft 5$, hence being observationally Markov equivalent.  For $\mcI = \{\emptyset, \{4\}\}$, we have $D \sim_\mcI D_1$, but $D \not\sim_\mcI D_2$ since the skeletons of $D^{(\{4\})}$ (Figure \ref{fig:ex-intervention-4}) and $D_2^{(\{4\})}$ do not coincide.}
    \label{fig:ex-equivalent-dags}
\end{figure}

\section{Essential Graphs}
\label{sec:essential-graphs}

Theorem \ref{thm:interventional-markov-equivalence} represents a computationally fast criterion for deciding whether two DAGs are interventionally Markov equivalent or not.  However, given some DAG $D$, it does not provide a possibility for quickly finding \emph{all} equivalent ones, and hence does not specify the equivalence class as a whole.  In this section, we give a characterization of graphs that uniquely represent an interventional Markov equivalence class (Theorem \ref{thm:essential-graph-characterization}).  Our characterization of these \emph{interventional essential graphs} is inspired by and similar to the one developed by \citet{Andersson1997Characterization} for the observational case and allows for handling equivalence classes algorithmically.  Furthermore, we present a linear time algorithm for constructing a representative of the equivalence class corresponding to an interventional essential graph (Proposition \ref{prop:construction-representative} and discussion thereafter), as well as a polynomial time algorithm for constructing the interventional essential graph of a given DAG (Algorithm \ref{alg:iterative-construction-essential-graph}).  Throughout this section, \mcI{} always stands for a conservative family of targets.

\subsection{Definitions and Motivation}
\label{sec:essential-graphs-motivation}

All DAGs in an \mcI-Markov equivalence class share the same skeleton; however, arrow orientations may vary between different representatives (Theorem \ref{thm:interventional-markov-equivalence}).  Varying and common arrow orientations are represented by undirected and directed edges, respectively, in \mcI-essential graphs.
\begin{definition}[\mcI-essential graph]
    \label{def:I-essential-graph}
    Let $D$ be a DAG.  The \textbf{\mcI-essential graph} of $D$ is defined as $\mathcal{E_I}(D) := \bigcup_{D' \in [D]_\mcI} D'$.  (The union is meant in the graph theoretic sense, see Appendix \ref{sec:graphs-notation}).
\end{definition}
When the family of targets \mcI{} in question is clear from the context, we will also use the term \textbf{interventional essential graph}, while ``observational essential graph'' shall refer to the concept of essential graphs as introduced by \citet{Andersson1997Characterization} in the observational case.  Simply speaking of ``essential graphs'', we mean interventional or observational essential graphs in the following.
\begin{definition}[\mcI-essential arrow]
    \label{def:essential-arrow}
    Let $D$ be a DAG.  An edge $a \grarright b \in D$ is \textbf{\mcI-essential} in $D$ if $a \grarright b \in D' \spforall D' \in [D]_\mcI$.
\end{definition}
An \mcI-essential graph typically contains directed as well as undirected edges.  Directed ones correspond to arrows that are \mcI-essential in every representative of the equivalence class; in other words, \mcI-essential arrows are those whose direction is identifiable.  A first sufficient criterion for an edge to be \mcI-essential follows immediately from Lemma \ref{lem:conserved-arrow} (Appendix \ref{sec:proofs-model}).

\begin{corollary}
    \label{cor:intervention-essential}
    Let $D$ be a DAG with $a \grarright b \in D$.  If there is an intervention target $I \in \mcI$ such that $|\{a, b\} \cap I| = 1$, then $a \grarright b$ is \mcI-essential.
\end{corollary}

The investigation of essential graphs has a long tradition in the observational case \citep{Andersson1997Characterization, Chickering2002Learning}.  Due to increased identifiability of causal structures, Markov equivalence classes shrink in the interventional case; Equation (\ref{eqn:equivalence-implication}) implies $\mathcal{E_I}(D) \subset \mathcal{E}_{\{\emptyset\}}(D)$ for any conservative family of targets \mcI{} (see also Figure \ref{fig:non-orientable-edges} in Section \ref{sec:evaluation}).  Essential graphs, interventional as well as observational ones, are mainly interesting because of two reasons:
\begin{itemize}
    \item It is important to know which arrow directions of a causal model are identifiable and which are not since arrow directions are relevant for the causal interpretation.
    
    \item Markov equivalent DAGs encode the same statistical model.  Hence the space of DAGs is no suitable ``parameter'' or search space for statistical inference and computation.  The natural search space is given by the set of the equivalence classes, the objects that can be distinguished from data.  Essential graphs uniquely represent these equivalence classes and are efficiently manageable in algorithms.
\end{itemize}

The characterization of \mcI-essential graphs (Theorem \ref{thm:essential-graph-characterization}) relies on the notion of strongly \mcI-protected arrows (Definition \ref{def:strongly-protected-arrow}) which reproduces the corresponding definition of \cite{Andersson1997Characterization} for $\mcI = \{\emptyset\}$; an illustration is given in Figure \ref{fig:ex-strong-protection}.

\begin{definition}[Strong protection]
    \label{def:strongly-protected-arrow} Let $G$ be a graph.  An arrow $a \grarright b \in G$ is \textbf{strongly \mcI-protected} in $G$ if there is some $I \in \mcI$ such that $|I \cap \{a, b\}| = 1$, or the arrow $a \grarright b$ occurs in at least one of the following four configurations as an induced subgraph of $G$:
        
    (a): \threegraph{$a$}{->}{$b$}{}{$c$}{->} \
    (b): \threegraph{$a$}{->}{$b$}{<-}{$c$}{} \
    (c): \threegraph{$a$}{->}{$b$}{<-}{$c$}{<-} \
    (d):
    \begin{tikzpicture}[baseline=(a.base)]
        \node[anchor=base east] (a) at (0, 0) {$a$};
        \node[anchor=base west] (b) at (1.6, 0) {$b$};
        \node[anchor=base] (c1) at (0.8,  0.7) {$c_1$};
        \node[anchor=base] (c2) at (0.8, -0.7) {$c_2$};
        \draw[->] (a)  -- (b);
        \draw[->] (c1) -- (b);
        \draw[->] (c2) -- (b);
        \draw[-]  (a)  -- (c1);
        \draw[-]  (a)  -- (c2);
    \end{tikzpicture}
\end{definition}

\begin{figure}[t]
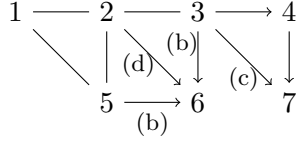

    \centering
    \begin{exgraphpicture}
        \draw[-]  (v2) -- (v1);
        \draw[-]  (v2) -- (v3);
        \draw[->] (v3) -- (v4);
        \draw[-]  (v1) -- (v5);
        \draw[-]  (v2) -- (v5);
        \draw[->] (v2) -- node[near start,below]{\footnotesize (d)} (v6);
        \draw[->] (v3) -- node[near start,left=-4pt]{\footnotesize (b)} (v6);
        \draw[->] (v5) -- node[below]{\footnotesize (b)} (v6);
        \draw[->] (v3) -- node[below]{\footnotesize (c)} (v7);
        \draw[->] (v4) -- (v7);
    \end{exgraphpicture}
    \caption{A graph with six arrows.  Four of them are strongly \mcI-protected for any conservative family of targets \mcI{} (in parentheses: arrow configurations according to Definition \ref{def:strongly-protected-arrow}).  Arrows $3 \grarright 4$ and $4 \grarright 7$ are strongly \mcI-protected for  $\mcI = \{\emptyset, \{4\}\}$, but not for $\mcI = \{\emptyset\}$.}
    \label{fig:ex-strong-protection}
\end{figure}

We will see in Theorem \ref{thm:essential-graph-characterization} that every arrow of an \mcI-essential graph (that is, every edge corresponding to an \mcI-essential arrow in the representative DAGs) is strongly \mcI-protected.  The configurations in Definition \ref{def:strongly-protected-arrow} guarantee the identifiability of the edge orientation between $a$ and $b$: if there is a target $I \in \mcI$ such that $|I \cap \{a, b\}| = 1$, turning the arrow would change the skeleton of the intervention graph $D^{(I)}$ (see also Corollary \ref{cor:intervention-essential}); in configuration (a), reversal would create a new v-structure; in (b), reversal would destroy a v-structure; in (c), reversal would create a cycle; an in (d) finally, at least one of the arrows between $a$ and $c_1$ or $c_2$ must point away from $a$ in each representative, hence turning the arrow $a \grarright b$ would create a cycle.  We refer to \citet{Andersson1997Characterization} for a more detailed discussion of the configurations (a) to (d).

\subsection{Characterization of Interventional Essential Graphs}
\label{sec:essential-graphs-characterization}

As in the observational setting, we can show that interventional essential graphs are chain graphs with chordal chain components (see Appendix \ref{sec:graphs-notation}).  For the observational case $\mcI = \{\emptyset\}$, Propositions \ref{prop:I-essential-chain-graph} and \ref{prop:construction-representative} below correspond to Propositions 4.1 and 4.2 of \citet{Andersson1997Characterization}.

\begin{proposition}
    \label{prop:I-essential-chain-graph}
    Let $D$ be a DAG on $[p]$. Then:
    \begin{subprop}
        \item \label{itm:essential-graph-chain-graph} $\mathcal{E_I}(D)$ is a chain graph.
        
        \item For each chain component $T \in \mathbf{T}(\mathcal{E_I}(D))$, the induced subgraph $\mathcal{E_I}(D)[T]$ is chordal.
    \end{subprop}
\end{proposition}

\begin{proposition}
    \label{prop:construction-representative}
    Let $D$ be a DAG.  A digraph $D'$ is acyclic and \mcI-equivalent to $D$ if and only if  $D'$ can be constructed by orienting the edges of every chain component of $\mathcal{E_I}(D)$ according to a perfect elimination ordering.
\end{proposition}
This proposition is not only of theoretic, but also of algorithmic interest.  According to the explanation in Appendix \ref{sec:perfect-elimination-orderings}, perfect elimination orderings on the (chordal) chain components of $\mathcal{E_I}(D)$ can be generated with \LexBFS{} (Algorithm \ref{alg:lex-bfs}); doing this for all chain components yields computational complexity $O(|E| + p)$, where $E$ denotes the edge set of $\mathcal{E_I}(D)$ (see Appendix \ref{sec:perfect-elimination-orderings}).

As an immediate consequence of Proposition \ref{prop:construction-representative}, interventional essential graphs are in one-to-one correspondence with interventional Markov equivalence classes.  We will therefore also speak about ``representatives of \mcI-essential graphs'', where we mean representatives (that is, DAGs) of the corresponding equivalence class.  Propositions \ref{prop:I-essential-chain-graph} and \ref{prop:construction-representative} give the justification for the following definition; note that in order to generate a representative of some \mcI-essential graph, the family of targets \mcI{} need not be known.

\begin{definition}
    \label{def:representatives}
    Let $G$ be the \mcI-essential graph of some DAG.  The set of representatives of $G$ is denoted by $\mathbf{D}(G)$:
    \begin{align*}
        \mathbf{D}(G) := & \{D \text{ a DAG} \spst D \subset G, D^u = G^u, D[T] \text{ oriented according to some}\\
        & \text{perfect elimination ordering for each chain component } T \in \mathbf{T}(G)\}.
    \end{align*}
\end{definition}
Here, $D^u$ denotes the skeleton of $D$ (Appendix \ref{sec:graphs-notation}). We can now state the main result of this section, a graph theoretic characterization of \mcI-essential graphs.  For the observational case $\mcI = \{\emptyset\}$, this theorem corresponds to Theorem 4.1 of \citet{Andersson1997Characterization}.

\begin{theorem}
    \label{thm:essential-graph-characterization}
    A graph $G$ is the \mcI-essential graph of a DAG $D$ if and only if
    \begin{subprop}
        \item \label{itm:chain-graph} $G$ is a chain graph;
        
        \item \label{itm:chordal-chain-components} for each chain component $T \in \mathbf{T}(G)$, $G[T]$ is chordal;
        
        \item \label{itm:forbidden-subgraph} $G$ has no induced subgraph of the form $a \grarright b \grline c$;
        
        \item \label{itm:forbidden-edge} $G$ has no line $a \grline b$ for which there exists some $I \in \mcI$ such that $|I \cap \{a, b\}| = 1$;
        
        \item \label{itm:strongly-protected-arrows} every arrow $a \grarright b \in G$ is strongly \mcI-protected.
    \end{subprop}
\end{theorem}

The graph $G$ of Figure \ref{fig:ex-strong-protection} satisfies points \ref{itm:chain-graph} to \ref{itm:forbidden-subgraph} of Theorem \ref{thm:essential-graph-characterization}.  For $\mcI = \{\emptyset, \{4\}\}$, it also fulfills points \ref{itm:forbidden-edge} and \ref{itm:strongly-protected-arrows}; in this case, it is the \mcI-essential graph $\mathcal{E_I}(D)$ of the DAG $D$ of Figure \ref{fig:ex-observational-dag} by Proposition \ref{prop:construction-representative}.

\subsection{Construction of Interventional Essential Graphs}
\label{sec:essential-graphs-construction}

In this section, we show that there is a simple way to construct the \mcI-essential graph $\mathcal{E_I}(D)$ of a DAG $D$: we need to successively convert arrows that are not strongly \mcI-protected into lines (Algorithm \ref{alg:iterative-construction-essential-graph}).  By doing this, we get a sequence of partial \mcI-essential graphs.
\begin{definition}[Partial \mcI-essential graph]
    \label{def:partial-essential-graph}
    Let $D$ be a DAG.  A graph $G$ with $D \subset G \subset \mathcal{E_I}(D)$ is called a \textbf{partial \mcI-essential graph} of $D$ if $a \grarright b \grline c$ does not occur as an induced subgraph of $G$.
\end{definition}
The following lemma can be understood as a motivation for looking at such graphs.  Note that due to the condition $G \subset \mathcal{E_I}(D)$, and because $G$ and $\mathcal{E_I}(D)$ have the same skeleton, every arrow of $\mathcal{E_I}(D)$ is also present in $G$, hence statement \ref{itm:partial-essential-protected} below makes sense.
\begin{lemma}
    \label{lem:partial-essential-graph-properties}
    Let $D$ be a DAG. Then:
    \begin{subprop}
        \item $D$ and $\mathcal{E_I}(D)$ are partial \mcI-essential graphs of $D$.
        
        \item \label{itm:partial-essential-protected} Let $G$ be a partial \mcI-essential graph of $D$.  Every arrow $a \grarright b \in \mathcal{E_I}(D)$ is strongly \mcI-protected in $G$.
        
        \item \label{itm:partial-essential-equivalent} Let $G$ be a partial \mcI-essential graph of two DAGs $D_1$ and $D_2$.  Then, $D_1 \sim_\mcI D_2$.
    \end{subprop}
\end{lemma}

Algorithm \ref{alg:iterative-construction-essential-graph} constructs the \mcI-essential graph $G$ from a partial \mcI-essential graph of any DAG $D \in \mathbf{D}(G)$.  The algorithm is indeed valid and calculates $\mathcal{E_I}(D)$, since the graph produced in each iteration is a partial \mcI-essential graph of $D$ (Lemma \ref{lem:iterative-arrow-replacement}), and the only partial \mcI-essential graph that has only strongly \mcI-protected arrows is $\mathcal{E_I}(D)$ (Lemma \ref{lem:maximal-partial-essential-graph}).
\begin{lemma}
    \label{lem:iterative-arrow-replacement}
    Let $D$ be a DAG and $G$ a partial \mcI-essential graph of $D$.  Assume that $a \grarright b \in G$ is not strongly \mcI-protected in $G$, and let $G' := G + (b, a)$ (that is, the graph we get by replacing the arrow $a \grarright b$ by a line $a \grline b$; see Appendix \ref{sec:graphs-notation}).  Then $G'$ is also a partial \mcI-essential graph of $D$.
\end{lemma}
\begin{lemma}
    \label{lem:maximal-partial-essential-graph}
    Let $D$ be a DAG.  There is exactly one partial \mcI-essential graph of $D$ in which every arrow is strongly \mcI-protected, namely $\mathcal{E_I}(D)$.
\end{lemma}

\begin{algorithm}[b]
    \caption{$\mathsc{ReplaceUnprotected}(\mcI, G)$.  Iterative construction of an \mcI-essential graph}
    \label{alg:iterative-construction-essential-graph}
    \small
    \SetKwInOut{Input}{Input}
    \SetKwInOut{Output}{Output}
    \Input{$G$: partial \mcI-essential graph of some DAG $D$ (not known)}
    \Output{$\mathcal{E_I}(D)$}
    \While{$\exists \ a \grarright b \in G$ s.t. $a \grarright b$ not strongly \mcI-protected in $G$}{%
        $G \leftarrow G + (b, a)$\;
    }
    \Return $G$\;
\end{algorithm}

To construct $\mathcal{E_I}(D)$ from some DAG $D = ([p], E)$, we must, in the worst case, execute the iteration of Algorithm \ref{alg:iterative-construction-essential-graph} for every arrow in the DAG; at each step, we must check every 4-tuple of vertices to see whether some arrow occurs in configuration (d) of Definition \ref{def:strongly-protected-arrow}.  Therefore Algorithm \ref{alg:iterative-construction-essential-graph} has at most complexity $O(|E| \cdot p^4)$; by exploiting the partial order $\preceq_G$ on $\mathbf{T}(G)$ (see Appendix \ref{sec:graphs-notation}), more efficient implementations are possible.  Note that some checks only need to be done once.  If, for example, an edge $a \grarright b$ is part of a v-structure (configuration (b) of Definition \ref{def:strongly-protected-arrow}), or if there is some $I \in \mcI$ such that $|I \cap \{a, b\}| = 1$ in the first iteration of Algorithm \ref{alg:iterative-construction-essential-graph}, this will also be the case in every later iteration.

\subsection{Example: Identifiability under Interventions}
\label{sec:essential-graphs-examples}

A simple example illustrates how much identifiability can be gained with a single intervention.  We consider a linear chain as observational essential graph:
$$
    G = \mathcal{E}_{\{\emptyset\}}(D): 1 \grline 2 \grline 3 \grline \cdots \grline p \ .
$$
We can easily count the number of representatives of $G$ using the following lemma.

\begin{lemma}[Source lemma]
    \label{lem:one-source}
    Let $G$ be a connected, chordal, undirected graph, and let $D \subset G$ be a DAG without v-structures and with $D^u = G$.  Then $D$ has exactly one source.
\end{lemma}

\begin{proof}
    Let $\sigma$ be a topological ordering of $D$; then, $\sigma(1)$ is a source, see Appendix \ref{sec:graphs-notation}.  It remains to show that there is at most one such source.  Assume, for the sake of contradiction, that there are two different sources $u$ and $v$.  Since $G$ is connected, there is a shortest $u$-$v$-path $\gamma = (a_0 \equiv u, a_1, \ldots, a_k \equiv v)$.  Let $a_i \grarleft a_{i+1} \in D$ be the first arrow that points away from $v$ in the chain $\gamma$ in $D$ (note $i \geq 1$ since $u \grarright a_1 \in D$ by assumption).  The v-structure $a_{i-1} \grarright a_i \grarleft a_{i+1}$ is not allowed as an induced subgraph of $D$, hence $a_{i-1}$ and $a_{i+1}$ must be adjacent in $D$ and in $G$; however, $\gamma$ is then no \emph{shortest} $u$-$v$-path, a contradiction.
\end{proof}

For our linear chain $G$ and any $s \in [p]$, there is exactly one DAG $D \in \mathbf{D}(G)$ that has the (unique) source $s$, namely the DAG we get by orienting \emph{all} edges of $G$ away from $s$; other edge orientations would produce a v-structure.  We conclude $G$ has $p$ representatives.

Assume that the true causal model producing the data is $(D, f)$, and denote the source of $D$ by $s \in [p]$.  Consider the conservative family of targets $\mcI = \{\emptyset, \{v\}\}$ with $v \in [p]$.  If $v < s$, the interventional essential graph $\mathcal{E_I}(D)$ is
$$
    1 \grarleft 2 \grarleft \ldots \grarleft v+1 \grline \ldots \grline p \ ,
$$
and $|\mathbf{D}(\mathcal{E_I}(D))| = p - v$ by the same arguments as above; analogously, if $v > s$, we find $|\mathbf{D}(\mathcal{E_I}(D))| = v - 1$.  On the other hand, if $v = s$, all edges of $D$ are strongly \mcI-protected: those incident to $s$ because of the intervention target, all others because they are in configuration (a) of Definition \ref{def:strongly-protected-arrow}; therefore, we have $\mathcal{E_I}(D) = D$.

In the best case, all edge orientations in the chain can be identified by a single intervention, while the observational essential graph $\mathcal{E}_{\{\emptyset\}}(D)$ that is identifiable from observational data alone contains $p$ representatives.  However, this needs an intervention at the a priori unknown source $s$.  Choosing the central vertex $\lceil \frac{p}{2} \rceil$ as intervention target ensures that at least half of the edges become directed in $\mathcal{E_I}(D)$, independent of the position $s$ of the source.

\section{Greedy Interventional Equivalence Search}
\label{sec:greedy-search}

Different algorithms have been proposed to estimate essential graphs from observational data.  One of them, the Greedy Equivalence Search (GES) \citep{Meek1997Graphical, Chickering2002Optimal}, is particularly interesting because of two properties:
\begin{itemize}
    \item It is score-based; it greedily maximizes some score function for given data over essential graphs.  It uses no tuning-parameter; the score function alone measures the quality of the estimate.  \citet{Chickering2002Optimal} chose the BIC score because of consistency; technically, any score equivalent and decomposable function (see Definition \ref{def:score-function-decomposable}) is adequate.
    
    \item It traverses the space of essential graphs which is the natural search space for model inference (see Section \ref{sec:essential-graphs}).  We will see in Section \ref{sec:evaluation} that a greedy search over \emph{equivalence classes} yields much better estimation results than a naïve greedy search over \emph{DAGs}.
\end{itemize}
GES greedily optimizes the score function in two phases \citep{Chickering2002Optimal}:
\begin{itemize}
    \item In the \textbf{forward phase}, the algorithm starts with the empty essential graph, $G_0 := ([p], \emptyset)$.  It then sequentially steps from one essential graph $G_i$ to a \emph{larger} one, $G_{i+1}$, for which there are representatives $D_i \in \mathbf{D}(G_i)$ and $D_{i+1} \in \mathbf{D}(G_{i+1})$ such that $D_{i+1}$ has exactly one arrow more than $D_i$.
    
    \item In the \textbf{backward phase}, the sequence $(G_i)_i$ is continued by gradually stepping from one essential graph $G_i$ to a \emph{smaller} one, $G_{i+1}$, for which there are representatives $D_i \in \mathbf{D}(G_i)$ and $D_{i+1} \in \mathbf{D}(G_{i+1})$ such that $D_{i+1}$ has exactly one arrow less than $D_i$.
\end{itemize}
In both phases, the respective candidate with maximal score is chosen, or the phase is aborted if no candidate scores higher than the current essential graph $G_i$.

We introduce in addition a new turning phase which proved to enhance estimation (see Section \ref{sec:evaluation}).  Here, the sequence $(G_i)_i$ is elongated by gradually stepping from one essential graph $G_i$ to a new one with the same number of edges, denoted by $G_{i+1}$, for which there are representatives  $D_i \in \mathbf{D}(G_i)$ and $D_{i+1} \in \mathbf{D}(G_{i+1})$ such that $D_{i+1}$ can be constructed from $D_i$ by turning exactly one arrow.  As before, we choose the highest scoring candidate.  Such a turning phase had already been proposed, but not characterized or implemented, by \citet{Chickering2002Optimal}.

Because GES is an optimization algorithm working on the space of observational essential graphs, and because the characterization of \emph{interventional} essential graphs is similar to that of observational ones (Theorem \ref{thm:essential-graph-characterization}), GES can indeed be generalized to handle interventional data as well by operating on interventional instead of observational essential graphs.  We call this generalized algorithm \emph{Greedy Interventional Equivalence Search} or GIES.  An overview is shown in Algorithm \ref{alg:greedy-search}: the forward, backward and turning phase are repeatedly executed in this order until none of them can augment the score function any more.

A naïve search strategy would perhaps traverse the space of DAGs instead of essential graphs, greedily adding, removing or turning single arrows from DAGs.  It is well-known in the observational case that such an approach performs markedly worse than one accounting for Markov equivalence \citep{Chickering2002Optimal, Castelo2003Inclusion}, and we will see in our simulations (Section \ref{sec:simulations}) that the same is true in the interventional case as long as few interventions are made.  Ignoring Markov equivalence cuts down the search space of successors at haphazard; since all DAGs in a Markov equivalence class represent the same statistical model, there is no justification for considering neighbors (that is, DAGs that can be reached by adding, removing or turning an arrow) of one of the representatives but not of the other ones.

GIES can be used with general score functions.  It goes without saying that the chosen score function should be a ``reasonable'' one which has favorable statistical properties such as consistency.  We denote the score of a DAG $D$ given interventional data $(\mathcal{T}, \mathbf{X})$ by $S(D; \mathcal{T}, \mathbf{X})$, and we assume that $S$ is \textbf{score equivalent}, that is, it assigns the same score to \mcI-equivalent DAGs; \mcI{} always stands for a conservative family of targets in this section.  Furthermore, we require $S$ to be decomposable.
\begin{definition}
    \label{def:score-function-decomposable}
    A score function $S$ is called \textbf{decomposable} if for each DAG $D$, $S$ can be written as a sum
    $$
        S(D; \mathcal{T}, \mathbf{X}) = \sum_{i=1}^p s(i, \pa_D(i); \mathcal{T}, \mathbf{X}),
    $$
    where the \textbf{local score} $s$ depends on $\mathbf{X}$ only via $\mathbf{X}_{\, \smallbullet \, i}$ and $\mathbf{X}_{\, \smallbullet \, \pa_D(i)}$, with $\mathbf{X}_{\, \smallbullet \, i}$ denoting the $i\supscr{th}$ column of $\mathbf{X}$ and $\mathbf{X}_{\, \smallbullet \, \pa_D(i)}$ the submatrix of $\mathbf{X}$ corresponding to the columns with index in $\pa_D(i)$.
\end{definition}

\begin{algorithm}[b!]
    \caption{$\mathsc{GIES}(\mathcal{T}, \mathbf{X})$. Greedy Interventional Equivalence Search.  The steps of the different phases of the algorithms are described in Algorithms \ref{alg:forward}--\ref{alg:turning}.}
    \label{alg:greedy-search}
    \small
    \SetKwInOut{Input}{Input}
    \SetKwInOut{Output}{Output}
    \Input{$(\mathcal{T}, \mathbf{X})$: interventional data for family of targets \mcI}
    \Output{\mcI-essential graph}
    $G \leftarrow ([p], \emptyset)$\;
    \Repeat{$\neg\mathrm{DoContinue}$}{%
        DoContinue $\leftarrow$ FALSE\;
        \Repeat{$G\subscr{old} = G$}{%
            $G\subscr{old} \leftarrow G$\;
            $G \leftarrow \mathsc{ForwardStep}(G; \mathcal{T}, \mathbf{X})$ \tcp*{See Algorithm \ref{alg:forward}}
        }
        \Repeat{$G\subscr{old} = G$}{%
            $G\subscr{old} \leftarrow G$\;
            $G \leftarrow \mathsc{BackwardStep}(G; \mathcal{T}, \mathbf{X})$ \tcp*{See Algorithm \ref{alg:backward}}
            \lIf{$G\subscr{old} \neq G$}{DoContinue $\leftarrow$ TRUE\;}
        }
        \Repeat{$G\subscr{old} = G$}{%
            $G\subscr{old} \leftarrow G$\;
            $G \leftarrow \mathsc{TurningStep}(G; \mathcal{T}, \mathbf{X})$ \tcp*{See Algorithm \ref{alg:turning}}
            \lIf{$G\subscr{old} \neq G$}{DoContinue $\leftarrow$ TRUE\;}
        }
    }
\end{algorithm}

Throughout the rest of this section, $S$ always denotes a score equivalent and decomposable score function.  Such a score function needs only be evaluated at one single representative of some interventional Markov equivalence class. Indeed, a key ingredient for the efficiency of the observational GES as well as our interventional GIES is an implementation that computes the greedy steps to the next equivalence class in a local fashion without enumerating all corresponding DAG members.  \citet{Chickering2002Optimal} found a clever way to do that in the forward and backward phase of the observational GES.  In Sections \ref{sec:gies-forward} and \ref{sec:gies-backward}, we generalize his methods to the interventional case, and in Section \ref{sec:gies-turning}, we propose an efficient implementation of the new turning phase.

\subsection{Forward Phase}
\label{sec:gies-forward}

A step in the forward phase of GIES can be formalized as follows: for an \mcI-essential graph $G_i$, find the next one $G_{i+1} := \mathcal{E_I}(D_{i+1})$, where
\begin{align*}
    D_{i+1} & := \argmax_{D' \in \mathbf{D}^+(G_i)} S(D'; \mathcal{T}, \mathbf{X}), \text{ and} \\
    \mathbf{D}^+(G_i) & := \{D' \text{ a DAG} \spst \exists \text{ an arrow } u \grarright v \in D': D' - (u, v) \in \mathbf{D}(G_i)\} \ .
\end{align*}
If no candidate DAG $D' \in \mathbf{D}^+(G_i)$ scores higher than $G_i$, abort the forward phase.

We denote the set of candidate \emph{\mcI-essential graphs} by $\boldsymbol{\mathcal{E}}_\mcI^+(G_i) := \{\mathcal{E_I}(D') \spst D' \in \mathbf{D}^+(G_i)\}$.  In the next proposition, we show that each graph $G' \in \boldsymbol{\mathcal{E}}_\mcI^+(G_i)$ can be characterized by a triple $(u, v, C)$, where $u \grarright v$ is the arrow that has to be added to a representative $D$ of $G_i$ in order to get a representative $D'$ of $G'$, and $C$ specifies the edge orientations of $D$ within the chain component of $v$ in $G$.

\begin{algorithm}[b!]
    \caption{$\mathsc{ForwardStep}(G; \mathcal{T}, \mathbf{X})$. One step of the forward phase of GIES.}
    \label{alg:forward}
    \SetKwInOut{Input}{Input}
    \SetKwInOut{Output}{Output}
    \Input{$G = ([p], E)$: \mcI-essential graph; $(\mathcal{T}, \mathbf{X})$: interventional data for \mcI}
    \Output{$G' \in \boldsymbol{\mathcal{E}}^+_\mcI(G)$, or $G$}
    $\Delta S\subscr{max} \leftarrow 0$\;
    \ShowLnLabel{ln:foreach-start}\ForEach{$v \in [p]$}{%
        \ForEach{$u \in [p] \setminus \ad_G(v)$}{%
            $N \leftarrow \nb_G(v) \cap \ad_G(u)$\;
            \ForEach(\tcp*[f]{Proposition \ref{prop:forward-characterization}\ref{itm:forward-C-clique} and \ref{itm:forward-N-subset-C}}){clique $C \subset \nb_G(v)$ with $N \subset C$} {
                \If(\tcp*[f]{Proposition \ref{prop:forward-characterization}\ref{itm:forward-paths}}){$\not\exists$ path from $v$ to $u$ in $G[[p] \setminus C]$}{%
                    $\Delta S \leftarrow s(v, \pa_G(v) \cup C \cup \{u\}; \mathcal{T}, \mathbf{X}) - s(v, \pa_G(v) \cup C; \mathcal{T}, \mathbf{X})$\;
                    \If{$\Delta S > \Delta S\subscr{max}$}{%
                        $\Delta S\subscr{max} \leftarrow \Delta S$\;
                        \ShowLnLabel{ln:foreach-end}$(u\subscr{max}, v\subscr{max}, C\subscr{max}) \leftarrow (u, v, C)$\;
                    }
                }
            }
        }
    }
    \If{$\Delta S\subscr{max} > 0$}{%
        $\sigma \leftarrow \LexBFS((C\subscr{max}, v\subscr{max}, \ldots), E[T_G(v\subscr{max})])$\;
        Orient edges of $G[T_G(v\subscr{max})]$ according to $\sigma$\;
        Insert edge $(u\subscr{max}, v\subscr{max})$ into $G$\;
        \Return{$\mathsc{ReplaceUnprotected}(\mcI, G)$} \tcp*{See Algorithm \ref{alg:iterative-construction-essential-graph}}
    }
    \lElse{\Return{$G$}\;}
\end{algorithm}

\begin{proposition}
    \label{prop:forward-characterization}
    Let $G$ be an \mcI-essential graph, let $u$ and $v$ be two non-adjacent vertices of $G$, and let $C \subset \nb_G(v)$.  Then there is a DAG $D \in \mathbf{D}(G)$ with $\{a \in \nb_G(v) \spst a \grarright v \in D\} = C$ such that $D' := D + (u, v) \in \mathbf{D}^+(G)$ if and only if
    \begin{subprop}
        \item \label{itm:forward-C-clique} $C$ is a clique in $G[T_G(v)]$;
        
        \item \label{itm:forward-N-subset-C} $N := \nb_G(v) \cap \ad_G(u) \subset C$;
        
        \item \label{itm:forward-paths} and every path from $v$ to $u$ in $G$ has a vertex in $C$.
    \end{subprop}
    For given $G$, $u$, $v$ and $C$ determine $D'$ uniquely up to \mcI-equivalence.
\end{proposition}
Note that points \ref{itm:forward-C-clique} and \ref{itm:forward-N-subset-C} imply in particular that $N$ is a clique in $G[T_G(v)]$.  Proposition \ref{prop:forward-characterization} has already been proven for the case of observational data \citep[Theorem 15]{Chickering2002Optimal}; it is not obvious, however, to see that this characterization of a forward step is also valid for interventional essential graphs, so we give a new proof in Appendix \ref{sec:proofs-greedy-search} using the results developed in Sections \ref{sec:model} and \ref{sec:essential-graphs}.

The DAGs $D$ and $D'$ in Proposition \ref{prop:forward-characterization} only differ in the edge $(u, v)$; $v$ is the only vertex whose parents are different in $D$ and $D'$.  Since the score function $S$ is assumed to be decomposable, the score difference between $D$ and $D'$ can be expressed by the local score change at vertex $v$, as stated in the following corollary.

\begin{corollary}
    \label{cor:forward-score-change}
    Let $G$, $u$, $v$, $C$, $D$ and $D'$ be as in Proposition \ref{prop:forward-characterization}.  The score difference $\Delta S := S(D'; \mathcal{T}, \mathbf{X}) - S(D; \mathcal{T}, \mathbf{X})$ can be calculated as follows:
    $$
        \Delta S = s(v, \pa_G(v) \cup C \cup \{u\}; \mathcal{T}, \mathbf{X}) - s(v, \pa_G(v) \cup C; \mathcal{T}, \mathbf{X}).
    $$
\end{corollary}
In the observational case, this corollary corresponds to Corollary 16 of \citet{Chickering2002Optimal}.

The most straightforward way to construct an \mcI-essential graph $G' \in \boldsymbol{\mathcal{E}}_\mcI^+(G)$ characterized by the triple $(u, v, C)$ as defined in Proposition \ref{prop:forward-characterization} would be to create a representative $D \in \mathbf{D}(G)$ by orienting the edges of $T_G(v)$ as indicated by the set $C$, add the arrow $u \grarright v$ to get $D'$, and finally construct $\mathcal{E_I}(D')$ with Algorithm \ref{alg:iterative-construction-essential-graph}.  The next lemma suggests a novel shortcut to this procedure: it is sufficient to orient the edges of the chain component $T_G(v)$ \emph{only} to get a partial \mcI-essential graph of $D'$ after adding the arrow $u \grarright v$.

\begin{lemma}
    \label{lem:forward-partial-essential-graph}
    Let $G$, $u$, $v$, $C$, $D$ and $D'$ be as in Proposition \ref{prop:forward-characterization}.  Let $H$ be the graph that we get by orienting all edges of $T_G(v)$ as in $D$ (leaving other chain components unchanged) and inserting the arrow $(u, v)$.  Then $H$ is a partial \mcI-essential graph of $D'$.
\end{lemma}

\begin{figure}[t]
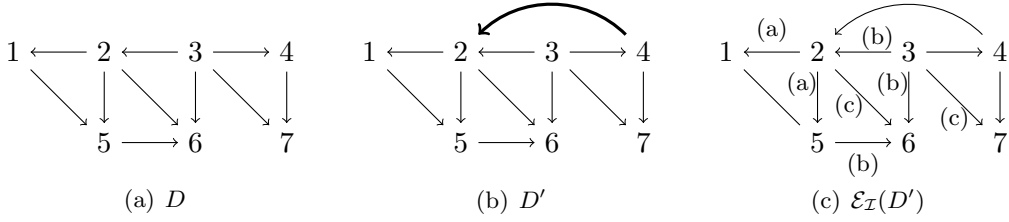

    \centering
    \subfigure[$D$]{%
        \begin{exgraphpicture}
            \draw[->] (v2) -- (v1);
            \draw[->] (v3) -- (v2);
            \draw[->] (v3) -- (v4);
            \draw[->] (v1) -- (v5);
            \draw[->] (v2) -- (v5);
            \draw[->] (v2) -- (v6);
            \draw[->] (v3) -- (v6);
            \draw[->] (v5) -- node[below]{\footnotesize \phantom{(b)}} (v6);
            \draw[->] (v3) -- (v7);
            \draw[->] (v4) -- (v7);
        \end{exgraphpicture}
    } \quad
    \subfigure[$D'$]{%
        \begin{exgraphpicture}
            \draw[->] (v2) -- (v1);
            \draw[->] (v3) -- (v2);
            \draw[->] (v3) -- (v4);
            \draw[->] (v1) -- (v5);
            \draw[->] (v2) -- (v5);
            \draw[->] (v2) -- (v6);
            \draw[->] (v3) -- (v6);
            \draw[->] (v5) -- node[below]{\footnotesize \phantom{(b)}} (v6);
            \draw[->] (v3) -- (v7);
            \draw[->] (v4) -- (v7);
            \draw[->, bend right=45, very thick] (v4) to (v2);
        \end{exgraphpicture}
    } \quad
    \subfigure[$\mathcal{E_I}(D')$]{%
        \begin{exgraphpicture}
            \draw[->] (v2) -- node[above]{\footnotesize (a)} (v1);
            \draw[->] (v3) -- node[above=-3pt, near start]{\footnotesize (b)} (v2);
            \draw[->] (v3) -- (v4);
            \draw[-]  (v1) -- (v5);
            \draw[->] (v2) -- node[left=-4pt, near start]{\footnotesize (a)} (v5);
            \draw[->] (v2) -- node[below, near start]{\footnotesize (c)} (v6);
            \draw[->] (v3) -- node[left=-4pt, near start]{\footnotesize (b)} (v6);
            \draw[->] (v5) -- node[below]{\footnotesize (b)} (v6);
            \draw[->] (v3) -- node[below]{\footnotesize (c)} (v7);
            \draw[->] (v4) -- (v7);
            \draw[->, bend right=45] (v4) to (v2);
        \end{exgraphpicture}
    }
    \caption{DAGs $D$, $D'$ and $\mathcal{E_I}(D')$ illustrating a possible forward step of GIES for the family of targets $\mcI = \{\emptyset, \{4\}\}$, applied to the \mcI-essential graph $G$ of Figure \ref{fig:ex-strong-protection} for the parameters $(u, v, C) = (4, 2, \{3\})$ (notation according to Proposition \ref{prop:forward-characterization}).  In parentheses in Figure (c): arrow configurations according to Definition \ref{def:strongly-protected-arrow}; arrows incident to $4$ are strongly \mcI-protected by the intervention target $\{4\}$.}
    \label{fig:ex-forward}
\end{figure}

Algorithm \ref{alg:forward} shows our implementation of the forward phase of GIES, summarizing the results of Proposition \ref{prop:forward-characterization}, Corollary \ref{cor:forward-score-change} and Lemma \ref{lem:forward-partial-essential-graph}.  Figure \ref{fig:ex-forward} illustrates one forward step, applied to the \mcI-essential graph $G$ (for $\mcI = \{\emptyset, \{4\}\}$) of Figure \ref{fig:ex-strong-protection} and characterized by the triple $(u, v, C) = (4, 2, \{3\})$.  Note that this triple is indeed valid in the sense of Proposition \ref{prop:forward-characterization}: $\{3\}$ is clearly a clique (point \ref{itm:forward-C-clique}), $\nb_G(2) \cap \ad_G(4) = \{3\}$ (point \ref{itm:forward-N-subset-C}), and there is no path from $2$ to $4$ in $G[[p] \setminus C]$ (point \ref{itm:forward-paths}).

\subsection{Backward Phase}
\label{sec:gies-backward}

In analogy to the forward phase, one step of the backward phase can be formalized as follows: for an \mcI-essential graph $G_i$, find its successor $G_{i+1} := \mathcal{E_I}(D_{i+1})$, where
\begin{align*}
    D_{i+1} & := \argmax_{D' \in \mathbf{D}^-(G_i)} S(D'; \mathbf{X}), \text{ and} \\
    \mathbf{D}^-(G_i) & := \{D' \text{ a DAG} \spst \exists \ D \in \mathbf{D}(G_i), u \grarright v \in D: D' = D - (u, v)\} \ .
\end{align*}
If no candidate DAG $D' \in \mathbf{D}^+(G_i)$ scores higher than $G_i$, the backward phase is aborted.

Whenever we have some representative $D \in \mathbf{D}(G)$ of an \mcI-essential graph $G$, we get a DAG in $\mathbf{D}^-(G)$ by removing any arrow of $D$.  This is in contrast to the forward phase where we do not necessarily get a DAG in $\mathbf{D}^+(G)$ by adding an arbitrary arrow to $D$.  By adding arrows, new directed cycles could be created, something which is not possible by removing arrows.  This is the reason why the backward phase is generally simpler to implement than the forward phase.

\begin{algorithm}[b!]
    \caption{$\mathsc{BackwardStep}(G; \mathcal{T}, \mathbf{X})$.  One step of the backward phase of GIES.}
    \label{alg:backward}
    \SetKwInOut{Input}{Input}
    \SetKwInOut{Output}{Output}
    \Input{$G = ([p], E)$: \mcI-essential graph; $(\mathcal{T}, \mathbf{X})$: interventional data for \mcI}
    \Output{$G' \in \boldsymbol{\mathcal{E}}^-_\mcI(G)$, or $G$}
    $\Delta S\subscr{max} \leftarrow 0$\;
    \ForEach{$v \in [p]$}{%
        \ForEach{$u \in \nb_G(v) \cup \pa_G(v)$}{%
            $N \leftarrow \nb_G(v) \cap \ad_G(u)$\;
            \ForEach{clique $C \subset N$} {
                $\Delta S \leftarrow s(v, (\pa_G(v) \cup C) \setminus \{u\}; \mathcal{T}, \mathbf{X}) - s(v, \pa_G(v) \cup C \cup \{u\}; \mathcal{T}, \mathbf{X})$\;
                \If{$\Delta S > \Delta S\subscr{max}$}{%
                    $\Delta S\subscr{max} \leftarrow \Delta S$\;
                    $(u\subscr{max}, v\subscr{max}, C\subscr{max}) \leftarrow (u, v, C)$\;
                }
            }
        }
    }
    \If{$\Delta S\subscr{max} > 0$}{%
        \lIf{$u\subscr{max} \in \nb_G(v\subscr{max})$}{$\sigma \leftarrow \LexBFS((C\subscr{max}, u\subscr{max}, v\subscr{max}, \ldots), E[T_G(v\subscr{max})])$\;}
        \lElse{$\sigma \leftarrow \LexBFS((C\subscr{max}, v\subscr{max}, \ldots), E[T_G(v\subscr{max})])$\;}
        Orient edges of $G[T_G(v\subscr{max})]$ according to $\sigma$\;
        Remove edge $(u\subscr{max}, v\subscr{max})$ from $G$\;
        \Return{$\mathsc{ReplaceUnprotected}(\mcI, G)$} \tcp*{See Algorithm \ref{alg:iterative-construction-essential-graph}}
    }
    \lElse{\Return{$G$\;}}
\end{algorithm}

In Proposition \ref{prop:backward-characterization} (corresponding to Theorem 17 of \citet{Chickering2002Optimal} for the observational case), we show that we can, similarly to the forward phase, characterize an \mcI-essential graph of $\boldsymbol{\mathcal{E}}^-_\mcI(G) := \{\mathcal{E_I}(D') \spst D' \in \mathbf{D}^-(G)\}$ by a triple $(u, v, C)$, where $C$ is a clique in $\nb_G(v)$.  As in the forward phase, we see that the score difference of $D$ and $D'$ is determined by the local score change at the vertex $v$ (Corollary \ref{cor:backward-score-change}), and that lines in chain components other than $T_G(v)$ remain lines in $G' = \mathcal{E_I}(D')$ (Lemma \ref{lem:backward-partial-essential-graph}).  Algorithm \ref{alg:backward} summarizes the results of the propositions in this section.

\begin{proposition}
    \label{prop:backward-characterization}
    Let $G = ([p], E)$ be an \mcI-essential graph with $(u, v) \in E$ (that is, $u \grline v \in G$ or $u \grarright v \in G$), and let $C \subset \nb_G(v)$.  There is a DAG $D \in \mathbf{D}(G)$ with $u \grarright v \in D$ and $\{a \in \nb_G(v) \setminus \{u\} \spst a \grarright v \in D\} = C$ such that $D' := D - (u, v) \in \mathbf{D}^-(G)$ if and only if
    \begin{subprop}
        \item \label{itm:backward-C-clique} $C$ is a clique in $G[T_G(v)]$;
        
        \item \label{itm:backward-C-subset-N} $C \subset N := \nb_G(v) \cap \ad_G(u)$.
    \end{subprop}
    Moreover, $u$, $v$ and $C$ determine $D'$ uniquely up to \mcI-equivalence for a given $G$.
\end{proposition}

\begin{corollary}
    \label{cor:backward-score-change}
    Let $G$, $u$, $v$, $C$, $D$ and $D'$ be as in Proposition \ref{prop:backward-characterization}.  The score difference $\Delta S := S(D'; \mathcal{T}, \mathbf{X}) - S(D; \mathcal{T}, \mathbf{X})$ is:
    $$
        \Delta S = s(v, (\pa_G(v) \cup C) \setminus \{u\}; \mathcal{T}, \mathbf{X}) - s(v, \pa_G(v) \cup C \cup \{u\}; \mathcal{T}, \mathbf{X}).
    $$
\end{corollary}
In the observational case, this corresponds to Corollary 18 in \citet{Chickering2002Optimal}.  The analogue to Lemma \ref{lem:forward-partial-essential-graph} for a computational shortcut in the forward phase reads as follows:

\begin{lemma}
    \label{lem:backward-partial-essential-graph}
    Let $G$, $u$, $v$, $C$, $D$ and $D'$ be as in Proposition \ref{prop:backward-characterization}.  Let $H$ be the graph that we get by orienting all edges of $T_G(v)$ as in $D$ and removing the arrow $(u, v)$.  Then $H$ is a partial \mcI-essential graph of $D'$.
\end{lemma}

A backward step of GIES is summarized in Algorithm \ref{alg:backward} and illustrated in Figure \ref{fig:ex-backward}.  The triple $(u, v, C) = (2, 5, \emptyset)$ used there to characterize the backward step obviously fulfills the requirements of Proposition \ref{prop:backward-characterization}.

\begin{figure}
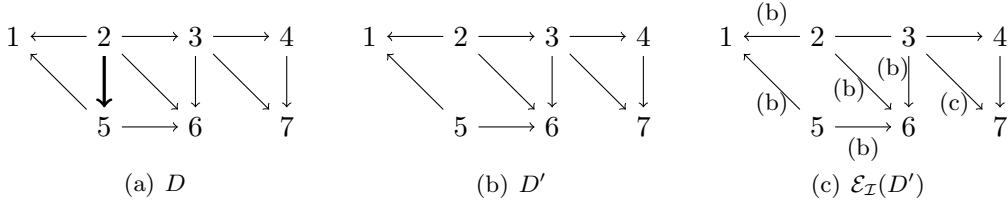

    \centering
    \subfigure[$D$]{%
        \begin{exgraphpicture}
            \draw[->] (v2) -- (v1);
            \draw[->] (v5) -- (v1);
            \draw[->] (v2) -- (v3);
            \draw[->] (v3) -- (v4);
            \draw[->, very thick] (v2) -- (v5);
            \draw[->] (v2) -- (v6);
            \draw[->] (v3) -- (v6);
            \draw[->] (v5) -- node[below]{\footnotesize \phantom{(b)}} (v6);
            \draw[->] (v3) -- (v7);
            \draw[->] (v4) -- (v7);
        \end{exgraphpicture}
    } \quad
    \subfigure[$D'$]{%
        \begin{exgraphpicture}
            \draw[->] (v2) -- (v1);
            \draw[->] (v5) -- (v1);
            \draw[->] (v2) -- (v3);
            \draw[->] (v3) -- (v4);
            \draw[->] (v2) -- (v6);
            \draw[->] (v3) -- (v6);
            \draw[->] (v5) -- node[below]{\footnotesize \phantom{(b)}} (v6);
            \draw[->] (v3) -- (v7);
            \draw[->] (v4) -- (v7);
        \end{exgraphpicture}
    } \quad
    \subfigure[$\mathcal{E_I}(D')$]{%
        \begin{exgraphpicture}
            \draw[->] (v2) -- node[above]{\footnotesize (b)} (v1);
            \draw[->] (v5) -- node[below]{\footnotesize (b)} (v1);
            \draw[-]  (v2) -- (v3);
            \draw[->] (v3) -- (v4);
            \draw[->] (v2) -- node[below, near start]{\footnotesize (b)} (v6);
            \draw[->] (v3) -- node[left=-4pt, near start]{\footnotesize (b)} (v6);
            \draw[->] (v5) -- node[below]{\footnotesize (b)} (v6);
            \draw[->] (v3) -- node[below]{\footnotesize (c)} (v7);
            \draw[->] (v4) -- (v7);
        \end{exgraphpicture}
    }
    \caption{DAGs $D$, $D'$ and $\mathcal{E_I}(D')$ illustrating a possible backward step of GIES for the family of targets $\mcI = \{\emptyset, \{4\}\}$, applied to the \mcI-essential graph $G$ of Figure \ref{fig:ex-strong-protection} for the parameters $(u, v, C) = (2, 5, \emptyset)$ (notation according to Proposition \ref{prop:backward-characterization}).  Figure (c), in parentheses: arrow configurations according to Definition \ref{def:strongly-protected-arrow}.}
    \label{fig:ex-backward}
\end{figure}

\subsection{Turning Phase}
\label{sec:gies-turning}

Finally, we characterize a step of the turning phase of GIES, in which we want to find the successor $G_{i+1} := \mathcal{E_I}(D_{i+1})$ for an \mcI-essential graph $G_i$ by the rule
\begin{align*}
    D_{i+1} := & \argmax_{D' \in \mathbf{D}^\circlearrowleft(G_i)} S(D'; \mathcal{T}, \mathbf{X}), \text{ where} \\
    \mathbf{D}^\circlearrowleft(G_i) := & \{D' \text{ a DAG } | D' \notin \mathbf{D}(G_i), \text{ and } \exists \text{ an arrow } u \grarright v \in D': \\
    & D' - (u, v) + (v, u) \in \mathbf{D}(G_i)\} \ .
\end{align*}
When the score cannot be augmented anymore, the turning phase is aborted.  The additional condition ``$D' \notin \mathbf{D}(G_i)$'' is not necessary in the definitions of $\mathbf{D}^+(G_i)$ and $\mathbf{D}^-(G_i)$; when adding or removing an arrow from a DAG, the skeleton changes, hence the new DAG is certainly not \mcI-equivalent to the previous one.  However, when \emph{turning} an arrow, the skeleton remains the same, and the danger of staying in the same equivalence class exists.

Again, we are looking for an efficient method to find a representative $D'$ for each $G' \in \boldsymbol{\mathcal{E}}_\mcI^\circlearrowleft(G_i) := \{\mathcal{E_I}(D') \spst D' \in \mathbf{D}^\circlearrowleft(G_i)\}$.  It makes sense to distinguish whether the arrow that should be turned in a representative $D \in \mathbf{D}(G_i)$ is \mcI-essential or not.  We start with the case where we want to turn an arrow which is \emph{not} \mcI-essential.

\begin{proposition}
    \label{prop:turning-undirected-characterization}
    Let $G$ be an \mcI-essential graph with $u \grline v \in G$, and let $C \subset \nb_G(v) \setminus \{u\}$.  Define $N := \nb_G(v) \cap \ad_G(u)$.  Then there is a DAG $D \in \mathbf{D}(G)$ with $u \grarleft v \in D$ and $\{a \in \nb_G(v) \spst a \grarright v \in D\} = C$ such that $D' := D - (v, u) + (u, v) \in \mathbf{D}^\circlearrowleft(G)$ if and only if
    \begin{subprop}
        \item \label{itm:turning-undirected-C-clique} $C$ is a clique in $G[T_G(v)]$;
        
        \item \label{itm:turning-undirected-diff-C-N} $C \setminus N \neq \emptyset$;
        
        \item \label{itm:turning-undirected-separation} $C \cap N$ separates $C \setminus N$ and $N \setminus C$ in $G[\nb_G(v)]$.
    \end{subprop}
    For a given $G$, $u$, $v$ and $C$ determine $D'$ up to \mcI-equivalence.
\end{proposition}
There are now \emph{two} vertices that have different parents in the DAGs $D$ and $D'$, namely $u$ and $v$; thus the calculation of the score difference between $D$ and $D'$ involves two local scores instead of one.
\begin{corollary}
    \label{cor:turning-undirected-score-change}
    Let $G$, $u$, $v$, $C$, $D$ and $D'$ be as in Proposition \ref{prop:turning-undirected-characterization}.  Then the score difference $\Delta S := S(D'; \mathcal{T}, \mathbf{X}) - S(D; \mathcal{T}, \mathbf{X})$ can be calculated as follows:
    \begin{align*}
        \Delta S = \ & s(v, \pa_G(v) \cup C \cup \{u\}; \mathcal{T}, \mathbf{X}) + s(u, \pa_G(u) \cup (C \cap N); \mathcal{T}, \mathbf{X}) \\
        & - s(v, \pa_G(v) \cup C; \mathcal{T}, \mathbf{X}) - s(u, \pa_G(u) \cup (C \cap N) \cup \{v\}; \mathcal{T}, \mathbf{X}).
    \end{align*}
\end{corollary}

\begin{lemma}
    \label{lem:turning-undirected-partial-essential-graph}
    Let $G$, $u$, $v$, $C$, $D$ and $D'$ be as in Proposition \ref{prop:turning-undirected-characterization}.  Let $H$ be the graph that we get by orienting all edges of $T_G(v)$ as in $D$ and turning the arrow $(v, u)$.  Then $H$ is a partial \mcI-essential graph of $D'$.
\end{lemma}

\begin{figure}
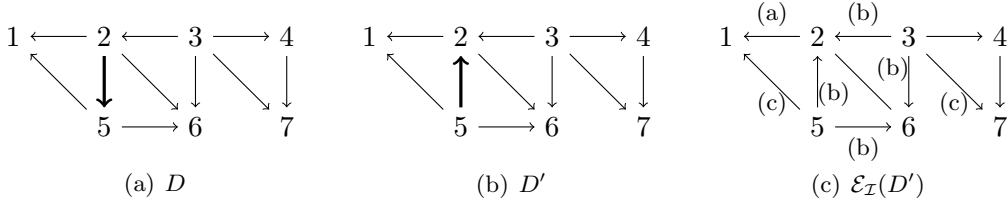

    \centering
    \subfigure[$D$]{%
        \begin{exgraphpicture}
            \draw[->] (v2) -- (v1);
            \draw[->] (v5) -- (v1);
            \draw[->] (v3) -- (v2);
            \draw[->] (v3) -- (v4);
            \draw[->, very thick] (v2) -- (v5);
            \draw[->] (v2) -- (v6);
            \draw[->] (v3) -- (v6);
            \draw[->] (v5) -- node[below]{\footnotesize \phantom{(b)}} (v6);
            \draw[->] (v3) -- (v7);
            \draw[->] (v4) -- (v7);
        \end{exgraphpicture}
    } \quad
    \subfigure[$D'$]{%
        \begin{exgraphpicture}
            \draw[->] (v2) -- (v1);
            \draw[->] (v5) -- (v1);
            \draw[->] (v3) -- (v2);
            \draw[->, very thick] (v5) -- (v2);
            \draw[->] (v3) -- (v4);
            \draw[->] (v2) -- (v6);
            \draw[->] (v3) -- (v6);
            \draw[->] (v5) -- node[below]{\footnotesize \phantom{(b)}} (v6);
            \draw[->] (v3) -- (v7);
            \draw[->] (v4) -- (v7);
        \end{exgraphpicture}
    } \quad
    \subfigure[$\mathcal{E_I}(D')$]{%
        \begin{exgraphpicture}
            \draw[->] (v2) -- node[above]{\footnotesize (a)} (v1);
            \draw[->] (v5) -- node[below]{\footnotesize (c)} (v1);
            \draw[->] (v3) -- node[above]{\footnotesize (b)} (v2);
            \draw[->] (v5) -- node[right=-4pt, near start]{\footnotesize (b)} (v2);
            \draw[->] (v3) -- (v4);
            \draw[-]  (v2) -- (v6);
            \draw[->] (v3) -- node[left=-4pt, near start]{\footnotesize (b)} (v6);
            \draw[->] (v5) -- node[below]{\footnotesize (b)} (v6);
            \draw[->] (v3) -- node[below]{\footnotesize (c)} (v7);
            \draw[->] (v4) -- (v7);
        \end{exgraphpicture}
    }
    \caption{DAGs $D$, $D'$ and $\mathcal{E_I}(D')$ illustrating a possible turning step of GIES applied to the \mcI-essential graph $G$ ($\mcI = \{\emptyset, \{4\}\}$) of Figure \ref{fig:ex-strong-protection} for the parameters $(u, v, C) = (5, 2, \{3\})$ (notation of Proposition \ref{prop:turning-undirected-characterization}).  The arrow $2 \grarright 5$ is not \mcI-essential in $D$.  Figure (c): arrow configurations in parentheses, see Definition \ref{def:strongly-protected-arrow}.}
    \label{fig:ex-turning-undirected}
\end{figure}

A possible turning step is illustrated in Figure \ref{fig:ex-turning-undirected}, where a non-\mcI-essential arrow (for $\mcI = \{\emptyset, \{4\}\}$) of a representative of the graph $G$ of Figure \ref{fig:ex-strong-protection} is turned.  The step is characterized by the triple $(u, v, C) = (5, 2, \{3\})$ which satisfies the conditions of Proposition \ref{prop:turning-undirected-characterization}: $\{3\}$ is obviously a clique (point \ref{itm:turning-undirected-C-clique}), $C \setminus N = C$ since $N = \{1\}$ (point \ref{itm:turning-undirected-diff-C-N}), and $C \setminus N = \{3\}$ and $N \setminus C = \{1\}$ are separated in $G[\nb_G(2)]$ (point \ref{itm:turning-undirected-separation}).  In contrast, the triple $(u, v, C) = (5, 2, \{1\})$ fulfills points \ref{itm:turning-undirected-C-clique} and \ref{itm:turning-undirected-separation} of Proposition \ref{prop:turning-undirected-characterization}, but not point \ref{itm:turning-undirected-diff-C-N}.  There is a DAG $D \in \mathbf{D}(G)$ with $\{a \in \nb_G(2) \spst a \grarright 2 \in D\} = \{1\}$, and turning the arrow $2 \grarright 5$ in $D$ yields another DAG $D'$ (that is, does not create a new cycle).  This new DAG $D'$, however, is \mcI-equivalent to $D$, and hence not a member of $\mathbf{D}^\circlearrowleft(G)$ (see the discussion above).

We now proceed to the case where an \mcI-essential arrow of a representative of $G$ is turned; here there is no danger to remain in the same Markov equivalence class.  The characterization of this case is similar to the forward phase.

\begin{proposition}
    \label{prop:turning-directed-characterization}
    Let $G$ be an \mcI-essential graph with $u \grarleft v \in G$, and let $C \subset \nb_G(v)$.  Then there is a DAG $D \in \mathbf{D}(G)$ with $\{a \in \nb_G(v) \spst a \grarright v \in D\} = C$ such that $D' := D - (v, u) + (u, v) \in \mathbf{D}^\circlearrowleft(G)$ if and only if
    \begin{subprop}
        \item \label{itm:turning-directed-C-clique} $C$ is a clique;
        
        \item \label{itm:turning-directed-N-subset-C} $N := \nb_G(v) \cap \ad_G(u) \subset C$;
        
        \item \label{itm:turning-directed-paths} every path from $v$ to $u$ in $G$ except $(v, u)$ has a vertex in $C \cup \nb_G(u)$.
    \end{subprop}
    Moreover, $u$, $v$ and $C$ determine $D'$ up to \mcI-equivalence.
\end{proposition}
\citet{Chickering2002Learning} has already proposed a turning step for essential arrows in the observational case; however, he did not provide necessary and sufficient conditions specifying all possible turning steps as Proposition \ref{prop:turning-directed-characterization} does.
\begin{lemma}
    \label{lem:turning-directed-partial-essential-graph}
    Let $G$, $u$, $v$, $C$, $D$ and $D'$ be as in Proposition \ref{prop:turning-directed-characterization}, and let $H$ be the graph that we get by orienting all edges of $T_G(v)$ and $T_G(u)$ as in $D$ and by turning the edge $(v, u)$.  Then $H$ is a partial \mcI-essential graph of $D'$.
\end{lemma}

\begin{figure}[t]
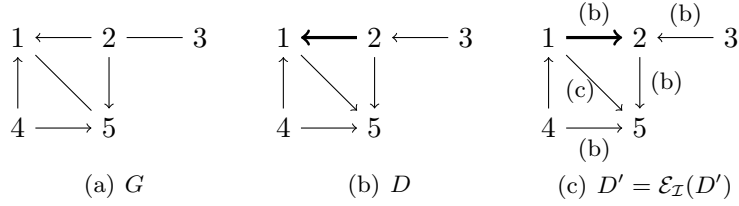

    \centering
    \subfigure[$G$]{%
        \begin{exgraphpicturesmall}
            \draw[->] (v2) -- (v1);
            \draw[->] (v4) -- (v1);
            \draw[-]  (v5) -- (v1);
            \draw[-]  (v3) -- (v2);
            \draw[->] (v2) -- (v5);
            \draw[->] (v4) -- node[below]{\footnotesize \phantom{(b)}} (v5);
        \end{exgraphpicturesmall}
    } \quad
    \subfigure[$D$]{%
        \begin{exgraphpicturesmall}
            \draw[->, very thick] (v2) -- (v1);
            \draw[->] (v4) -- (v1);
            \draw[->] (v3) -- (v2);
            \draw[->] (v1) -- (v5);
            \draw[->] (v2) -- (v5);
            \draw[->] (v4) -- node[below]{\footnotesize \phantom{(b)}} (v5);
        \end{exgraphpicturesmall}
    } \quad
    \subfigure[$D' = \mathcal{E_I}(D')$]{%
        \begin{exgraphpicturesmall}
            \draw[->] (v4) -- (v1);
            \draw[->, very thick] (v1) -- node[above]{\footnotesize (b)} (v2);
            \draw[->] (v3) -- node[above]{\footnotesize (b)} (v2);
            \draw[->] (v1) -- node[below, near start]{\footnotesize (c)} (v5);
            \draw[->] (v2) -- node[right]{\footnotesize (b)} (v5);
            \draw[->] (v4) -- node[below]{\footnotesize (b)} (v5);
        \end{exgraphpicturesmall}
    }
    \caption{Graphs $G$, $D$, $D'$ and $\mathcal{E_I}(D')$ illustrating a possible turning step of GIES for the family of targets $\mcI = \{\emptyset, \{4\}\}$ and the parameters $(u, v, C) = (1, 2, \{3\})$ (notation of Proposition \ref{prop:turning-directed-characterization}).  The arrow $2 \grarright 1$ is \mcI-essential in $D$.  Figure (c): arrow configurations in parentheses, see Definition \ref{def:strongly-protected-arrow}.}
    \label{fig:ex-turning-directed}
\end{figure}

To construct a $G' \in \boldsymbol{\mathcal{E}}_\mcI^\circlearrowleft(G)$ out of $G$, we must possibly orient \emph{two} chain components of $G$ instead of one (Lemma \ref{lem:turning-directed-partial-essential-graph}).  In the example of Figure \ref{fig:ex-turning-directed}, we see that it is indeed not sufficient to orient the edges of $T_G(v)$ alone in order to get a partial \mcI-essential graph of $G'$.  The arrow $1 \grarright 5$ is not \mcI-essential in $D$, hence $5 \in T_G(1)$.  However, the same arrow is \mcI-essential in $D'$ and hence also present in $\mathcal{E_I}(D')$.

\begin{algorithm}[b!]
    \caption{$\mathsc{TurningStep}(G; \mathcal{T}, \mathbf{X})$.  One step of the turning phase of GIES.}
    \label{alg:turning}
    \small
    \SetKwInOut{Input}{Input}
    \SetKwInOut{Output}{Output}
    \Input{$G = ([p], E)$: \mcI-essential graph; $(\mathcal{T}, \mathbf{X})$: interventional data for \mcI}
    \Output{$G' \in \boldsymbol{\mathcal{E}}^\circlearrowleft_\mcI$, or $G$}
    $\Delta S\subscr{max} \leftarrow 0$\;
    \ForEach{$v \in [p]$}{%
        \ForEach(\tcp*[f]{Consider arrows that are not \mcI-essential for turning}){$u \in \nb_G(v)$}{%
            $N \leftarrow \nb_G(u) \cap \ad_G(v)$\;
            \ForEach(\tcp*[f]{Proposition \ref{prop:turning-undirected-characterization}\ref{itm:turning-undirected-C-clique}}){clique $C \subset \nb_G(v) \setminus \{u\}$}{%
                \If{$C \setminus N \neq \emptyset$ and $\{u, v\}$ separates $C$ and $N \setminus C$ in $G[T_G(v)]$}{%
                    \tcp*[f]{Proposition \ref{prop:turning-undirected-characterization}\ref{itm:turning-undirected-diff-C-N} and \ref{itm:turning-undirected-separation}}
                    $\Delta S \leftarrow s(v, \pa_G(v) \cup C \cup \{u\}; \mathcal{T}, \mathbf{X}) + s(u, \pa_G(u) \cup (C \cap N); \mathcal{T}, \mathbf{X})$\;
                    $\Delta S \leftarrow \Delta S - s(v, \pa_G(v) \cup C; \mathcal{T}, \mathbf{X}) - s(u, \pa_G(u) \cup (C \cap N) \cup \{v\}; \mathcal{T}, \mathbf{X})$\;
                    \If{$\Delta S > \Delta S\subscr{max}$}{%
                        $\Delta S\subscr{max} \leftarrow \Delta S$\;
                        $(u\subscr{max}, v\subscr{max}, C\subscr{max}) \leftarrow (u, v, C)$\;
                    }
                }
            }
        }
        
        \ForEach(\tcp*[f]{Consider \mcI-essential arrows for turning}){$u \in \ch_G(v)$}{%
            $N \leftarrow \nb_G(v) \cap \ad_G(u)$\;
            \ForEach(\tcp*[f]{Proposition \ref{prop:turning-directed-characterization}\ref{itm:turning-undirected-C-clique} and \ref{itm:turning-directed-N-subset-C}}){clique $C \subset \nb_G(v)$ with $N \subset C$}{%
                \If(\tcp*[f]{Proposition \ref{prop:turning-directed-characterization}\ref{itm:turning-directed-paths}}){$\not\exists$ path from $v$ to $u$ in $G[[p] \setminus (C \cup \nb_G(u))] - (v, u)$}{%
                    $\Delta S \leftarrow s(v, \pa_G(v) \cup C \cup \{u\}; \mathcal{T}, \mathbf{X}) + s(u, \pa_G(u) \setminus \{v\}; \mathcal{T}, \mathbf{X})$\;
                    $\Delta S \leftarrow \Delta S - s(v, \pa_G(v) \cup C; \mathcal{T}, \mathbf{X}) - s(u, \pa_G(u); \mathcal{T}, \mathbf{X})$\;
                    \If{$\Delta S > \Delta S\subscr{max}$}{%
                        $\Delta S\subscr{max} \leftarrow \Delta S$\;
                        $(u\subscr{max}, v\subscr{max}, C\subscr{max}) \leftarrow (u, v, C)$\;
                    }
                }
            }
        }
    }
    \If{$\Delta S\subscr{max} > 0$}{%
        \If{$v\subscr{max} \grarright u\subscr{max} \in G$}{%
            $\sigma_u := \LexBFS((u\subscr{max}, \ldots), E[T_G(u\subscr{max})])$\;
            Orient edges of $G[T_G(u\subscr{max})]$ according to $\sigma$\;
            $\sigma_v := \LexBFS((C\subscr{max}, v\subscr{max}, \ldots), E[T_G(v\subscr{max})])$\;
        }
        \lElse{$\sigma_v := \LexBFS((C\subscr{max}, v\subscr{max}, u\subscr{max}, \ldots), E[T_G(v\subscr{max})])$\;}
        Orient edges of $G[T_G(v\subscr{max})]$ according to $\sigma_v$\;
        Turn edge $(v\subscr{max}, u\subscr{max})$ in $G$\;
        \Return{$\mathsc{ReplaceUnprotected}(\mcI, G)$} \tcp*{See Algorithm \ref{alg:iterative-construction-essential-graph}}
    }
    \lElse{\Return{$G$\;}}
\end{algorithm}

Despite the fact that we need to orient the edges of $T_G(v)$ and $T_G(u)$ to get a partial \mcI-essential graph of $D'$, $\mathcal{E_I}(D')$ is nevertheless determined by the orientation of edges adjacent to $v$ (determined by the clique $C$) alone.  This comes from the fact that in $D$, defined as in Proposition \ref{prop:turning-directed-characterization}, all arrows of $D[T_G(u)]$ must point away from $u$.

\begin{corollary}
    \label{cor:turning-directed-score-change}
    Let $G$, $u$, $v$, $C$, $D$ and $D'$ be as in Proposition \ref{prop:turning-directed-characterization}.  Then the score difference $\Delta S := S(D'; \mathcal{T}, \mathbf{X}) - S(D; \mathcal{T}, \mathbf{X})$ can be calculated as follows:
    \begin{eqnarray*}
        \Delta S & = & s(v, \pa_G(v) \cup C \cup \{u\}; \mathcal{T}, \mathbf{X}) + s(u, \pa_G(u) \setminus \{v\}; \mathcal{T}, \mathbf{X}) \\
        & & - s(v, \pa_G(v) \cup C; \mathcal{T}, \mathbf{X}) - s(u, \pa_G(u); \mathcal{T}, \mathbf{X}).
    \end{eqnarray*}
\end{corollary}

The entire turning step, for essential and non-essential arrows, is shown in Algorithm \ref{alg:turning}.

\subsection{Discussion}
\label{sec:discussion}

Every step in the forward, backward and turning phase of GIES is characterized by a triple $(u, v, C)$, where $u$ and $v$ are different vertices and $C$ is a clique in the neighborhood of $v$.  To identify the highest scoring movement from one \mcI-essential graph $G$ to a potential successor in $\boldsymbol{\mathcal{E}}_\mcI^+(G)$, $\boldsymbol{\mathcal{E}}_\mcI^-(G)$ or $\boldsymbol{\mathcal{E}}_\mcI^\circlearrowleft(G)$, respectively, one potentially has to examine all cliques in the neighborhood $\nb_G(v)$ of all vertices $v \in [p]$.  The time complexity of any (forward, backward or turning) step applied to an \mcI-essential graph $G$ hence highly depends on the size of the largest clique in the chain components of $G$.  By restricting GIES to \mcI-essential graphs with a bounded vertex degree, the time complexity of a step of GIES is polynomial in $p$; otherwise, it is in the worst case exponential.  We believe, however, that GIES is in practice much more efficient than this worst-case complexity suggests.  Some evidence for this claim is provided by the runtime analysis of our simulation study, see Section \ref{sec:simulations}.

A heuristic approach to guarantee polynomial runtime of a greedy search has been proposed by \citet{Castelo2003Inclusion} for the observational case.  Their Hill Climber Monte Carlo (HCMC) algorithm operates in DAG space, but to account for Markov equivalence, the neighborhood of a number of randomly chosen DAGs equivalent to the current one is scanned in each greedy step.  The equivalence class of the current DAG is explored by randomly turning ``covered arrows'', that is, arrows whose reversal does not change the Markov property.  In our (interventional) notation, an arrow is covered if and only if it is \emph{not} strongly \mcI-protected (Definition \ref{def:strongly-protected-arrow}).  By limiting the number of covered arrow reversals, a polynomial runtime is guaranteed at the cost of potentially lowering the probability of investigating a particular successor in $\boldsymbol{\mathcal{E}}_\mcI^+(G)$, $\boldsymbol{\mathcal{E}}_\mcI^-(G)$ or $\boldsymbol{\mathcal{E}}_\mcI^\circlearrowleft(G)$, respectively.  HCMC hence enables a fine tuning of the trade-off between exploration of the search space and runtime, or between greediness and randomness.

The order of executing the backward and the turning phase seems somewhat arbitrary.  In the analysis of the steps performed by GIES in our simulation study (Section \ref{sec:simulations}), we saw that the turning phase can generally only augment the score when very few backward steps were executed before.  For this reason, we believe that changing the order of the backward and the turning phase would have little effect on the overall performance of GIES.

As already discussed by \citet{Chickering2002Optimal} for the observational case, caching techniques can markedly speed up GES; the same holds for GIES.  The basic idea is the following: in a forward step, the algorithm evaluates a lot of triples $(u, v, C)$ to choose the best one, $(u\subscr{max}, v\subscr{max}, C\subscr{max})$ (lines \ref{ln:foreach-start} to \ref{ln:foreach-end} in Algorithm \ref{alg:forward}).  After performing the forward move corresponding to $(u\subscr{max}, v\subscr{max}, C\subscr{max})$, many of the triples evaluated in the step before are still valid candidates for next step in the sense of Proposition \ref{prop:forward-characterization} and lead to the same score difference as before (see Corollary \ref{cor:forward-score-change}).  Caching those values avoids unnecessary reevaluation of possible forward steps.  The same holds for the backward and the turning phase; since the forward step is most frequently executed, a caching strategy in this phase yields the highest speed-up though.

We emphasize that the characterization of ``neighboring'' \mcI-essential graphs in $\boldsymbol{\mathcal{E}}_\mcI^+(G)$, $\boldsymbol{\mathcal{E}}_\mcI^-(G)$ or $\boldsymbol{\mathcal{E}}_\mcI^\circlearrowleft(G)$, respectively, by triples $(u, v, C)$ is of more general interest for structure learning algorithms, for example for the design of sampling steps of an MCMC algorithm.  Also the beforementioned HCMC algorithm could be extended to interventional data by generalizing the notion of ``covered arcs'' using Definition \ref{def:strongly-protected-arrow}.

The prime example of a score equivalent and decomposable score function is the Bayesian information criterion (BIC) \citep{Schwarz1978Estimating} which we used in our simulations (Section \ref{sec:evaluation}).  It penalizes the complexity of causal models by their number of free parameters ($\ell_0$ penalization); this number is the sum of free parameters of the conditional densities in the Markov factorization (Definition \ref{def:directed-markov-property}), which explains the decomposability of the score.  Using different penalties, for example, $\ell_2$ penalization, can lead to a non-decomposable score function.  GIES can also be adapted to such score functions; the calculation of score differences becomes computationally more expensive in this case since it cannot be done in a local fashion as in Corollaries \ref{cor:forward-score-change}, \ref{cor:backward-score-change}, \ref{cor:turning-undirected-score-change} and \ref{cor:turning-directed-score-change}.

GIES only relies on the notion of interventional Markov equivalence, and on a score function that can be evaluated for a given class of causal models.  As we mentioned in Section \ref{sec:causal-calculus}, we believe that interventional Markov equivalence classes remain unchanged for models that do not have a strictly positive density.  For this reason it should be safe to also apply GIES to such a model class.

\section{Experimental Evaluation}
\label{sec:evaluation}

We evaluated the GIES algorithm on simulated interventional data (Section \ref{sec:simulations}) and on \textit{in silico} gene expression data sets taken from the DREAM4 challenge \citep{Marbach2010Revealing} (Section \ref{sec:dream4}).  In both cases, we restricted our considerations to \emph{Gaussian} causal models as summarized in Section \ref{sec:gaussian-models}.

\subsection{Gaussian Causal Models}
\label{sec:gaussian-models}

Consider a causal model $(D, f)$ with a Gaussian density of the form $\mathcal{N}(0, \Sigma)$.  The observational Markov property of such a model translates to a set of \emph{linear} structural equations
\begin{equation}
    \label{eqn:structural-equations}
    X_i = \sum_{j=1}^p \beta_{ij} X_j + \varepsilon_i, \ \varepsilon_i \stackrel{\text{indep.}}{\sim} \mathcal{N}(0, \sigma_i^2), \quad 1 \leq i \leq p \ ,
\end{equation}
where $\beta_{ij} = 0$ if $j \notin \pa_D(i)$.  When the DAG structure $D$ is known, the covariance matrix $\Sigma$ can be parameterized by the \textbf{weight matrix}
$$
    B := (\beta_{ij})_{i, j = 1}^p \in \mathbf{B}(D) := \{A = (\alpha_{ij}) \in \sR^{p \times p} \spst \alpha_{ij} = 0 \text{ if } j \notin \pa_D(i)\}
$$
that assigns a weight $\beta_{ij}$ to each arrow $j \grarright i \in D$, and the vector of error covariances $\sigma^2 := (\sigma_1^2, \ldots, \sigma_p^2)$:
$$
    \Sigma = \Cov(X) = (\mathbbm{1} - B)^{-1} \diag(\sigma^2) (\mathbbm{1} - B)^{-\transp} \ .
$$
This is a consequence of Equation (\ref{eqn:structural-equations}).

We always assume Gaussian intervention variables $U_I$ (see Section \ref{sec:causal-calculus}).  In this case, not only the observational density $f$ is Gaussian, but also the interventional densities $f(x \spst \doop_D(X_I = U_I))$.  An interventional data set $(\mathcal{T}, \mathbf{X})$ as defined in Equation (\ref{eqn:dataset}) then consists of $n$ independent, but not identically distributed Gaussian samples.

We use the \textbf{Bayesian information criterion} (BIC) as score function for GIES:
$$
    S(D; \mathcal{T}, \mathbf{X}) := \sup \{\ell_D(B, \sigma^2; \mathcal{T}, \mathbf{X}) \spst B \in \mathbf{B}(D), \sigma^2 \in \sR^p_{>0}\} - \frac{k_D}{2} \log(n) \ ,
$$
where $\ell_D$ denotes the log-likelihood of the density in Equation (\ref{eqn:sample-density}):
\begin{eqnarray}
    \ell_D(B, \sigma^2; \mathcal{T}, \mathbf{X}) & := & \sum_{i=1}^n \log f\big(X^{(i)} \spst \doop_D(X_{T^{(i)}}^{(i)} = U_{T^{(i)}})\big) \label{eqn:gaussian-likelihood} \\
    & = & \sum_{i = 1}^n \Big[ \sum_{j \notin T^{(i)}} \log f(X_j^{(i)} \spst X_{\pa_D(j)}^{(i)} ) + \sum_{j \in T^{(i)}} \log \tilde{f} (X_j^{(i)}) \Big] \nonumber \\
    & = & - \frac{1}{2} \sum_{i = 1}^n \sum_{j \notin T^{(i)}} \Big[ \log \sigma_j^2 + \frac{1}{\sigma_j^2} \left( X_j^{(i)} - B_{j \smallbullet} X^{(i)} \right)^2 \Big] + C \nonumber \\
    & = & - \frac{1}{2} \sum_{j = 1}^p \Big[ |\{i \spst j \notin T^{(i)}\}| \log \sigma_j^2 + \frac{1}{\sigma_j^2} \sum_{i: j \notin T^{(i)}} \left( X_j^{(i)} - B_{j \smallbullet} X^{(i)} \right)^2 \Big] + C \ , \nonumber
\end{eqnarray}
where the constant $C$ is independent of the parameters $(B, \sigma^2)$ of the model.  Since Gaussian causal models with structure $D$ are parameterized by $B \in \mathbf{B}(D)$ and $\sigma^2 \in \sR_{>0}^p$, we have $k_D = p + |E|$ free parameters, where $E$ denotes the edge set of $D$.  It can be seen in Equation (\ref{eqn:gaussian-likelihood}) that the \textbf{maximum likelihood estimator} (MLE) $(\hat{B}, \hat{\sigma}^2)$, the maximizer of $\ell_D$, minimizes the residual sum of squares for the different structural equations; for more details we refer to \citet{Hauser2012Consistent}.

The DAG $\hat{D}$ maximizing the BIC yields a \emph{consistent} estimator for the true causal structure $D$ in the sense that $P[\hat{D} \sim_\mcI D] \to 1$ in the limit $n \to \infty$ as long as the true density $f$ is \textbf{faithful} with respect to $D$, that is, \emph{every} conditional independence relation of $f$ is encoded in the Markov property of $D$ \citep{Hauser2012Consistent}.  Note that the BIC score is even defined in the high-dimensional setting $p > n$; however, we only consider low-dimensional settings here.

\subsection{Simulations}
\label{sec:simulations}

We simulated interventional data from 4000 randomly generated Gaussian causal models as described in Section \ref{sec:model-generation}.  In Sections \ref{sec:different-algorithms} and \ref{sec:quality-measures}, we present our methods for evaluating GIES; the results are discussed in Section \ref{sec:results}.  As a rough summary, GIES markedly beat the conceptually simpler greedy search over the space of DAGs as well as the original GES of \citet{Chickering2002Optimal} ignoring the interventional nature of the simulated data sets.  Its learning performance could keep up with a provably consistent exponential time dynamic programming algorithm at much lower computational cost.

\subsubsection{Generation of Gaussian Causal Models}
\label{sec:model-generation}

For some number $p$ of vertices, we randomly generated Gaussian causal models parameterized by a structure $D$, a weight matrix $B \in \mathbf{B}(D)$ and a vector of error covariances $\sigma^2 \in \sR_{>0}^p$ by a procedure slightly adapted from \citet{Kalisch2007Estimating}:
\begin{enumerate}
    \item \label{itm:draw-dag} For a given \textbf{sparseness parameter} $s \in (0, 1)$, draw a DAG $D$ with topological ordering $(1, \ldots, p)$ and binomially distributed vertex degrees with mean $s(p - 1)$.

    \item Shuffle the vertex indices of $D$ to get a random topological ordering.

    \item \label{itm:draw-weights} For each arrow $j \grarright i \in D$, draw $\beta'_{ij} \sim \mathcal{U}([-1, -0.1] \cup [0.1, 1])$ using independent realizations; for other pairs of $(i, j)$, set $\beta'_{ij} = 0$ (see Equation (\ref{eqn:structural-equations})).  This yields a weight matrix $B' = (\beta'_{ij})_{i, j = 1}^p \in \mathbf{B}(D)$ with positive as well as negative entries which are bounded away from $0$.
    
    \item Draw error variances $\sigma'^2_i \stackrel{\mathrm{i.i.d.}}{\sim} \mathcal{U}([0.5, 1])$.
    
    \item Calculate the corresponding covariance matrix $\Sigma' = (\mathbbm{1} - B')^{-1} \diag(\sigma'^2) (\mathbbm{1} - B')^{-\transp}$.
    
    \item Set $H := \diag((\Sigma'_{11})^{-1/2}, \ldots, (\Sigma'_{pp})^{-1/2})$, and normalize the weights and error variances as follows:
    $$
        B := H B' H^{-1}, \quad (\sigma^2_1, \ldots, \sigma^2_p)^\transp := H^2 (\sigma'^2_1, \ldots, \sigma'^2_p)^\transp \ .
    $$
    It can easily be seen that the corresponding covariance matrix fulfills
    $$
        \Sigma = (\mathbbm{1} - B)^{-1} \diag(\sigma^2) (\mathbbm{1} - B)^{-\transp} = H \Sigma' H \ ,
    $$
    ensuring the desired normalization $\Sigma_{ii} = 1$ for all $i$.
\end{enumerate}
Steps \ref{itm:draw-dag} and \ref{itm:draw-weights} are provided by the function \texttt{randomDAG()} of the R-package \texttt{pcalg} \citep{Kalisch2012Causal}.

We considered families of targets of the form $\mcI = \{\emptyset, I_1, \ldots, I_k\}$, where $I_1, \ldots, I_k$ are $k$ different, randomly chosen intervention targets of size $m$; the target size $m$ had values between $1$ and $4$.  For a fixed sample size $n$, we produced approximately the same number of data samples for each target in the family \mcI{} by using a level density $\mathcal{N}((2, \ldots, 2), (0.2)^2 \mathbbm{1}_m)$ in each case (see the model in Equation (\ref{eqn:interventional-density})).  With this choice and the aforementioned normalization of $\Sigma$, the mean values of the intervention levels lay $2$ standard deviations above the mean values of the observational marginal distributions.  In total, we considered $4000$ causal models and simulated $128$ observational or interventional data sets from each of them by combining the following simulation parameters:
\begin{itemize}
    \item $(p, s) \in \{(10, 0.2), (20, 0.1), (30, 0.1), (40, 0.1)\}$ with $1000$ DAGs each.
    
    \item $k = 0, 0.2p, 0.4p, \ldots, p$ for each value of $p$; the first setting is purely observational.
    
    \item $m \in \{1, 2, 4\}$.
    
    \item $n \in \{50, 100, 200, 500, 1000, 2000, 5000, 10000\}$.
\end{itemize}
In addition, we generated causal models with $p \in \{50, 100, 200\}$ ($100$ DAGs each) and $p = 500$ ($20$ DAGs) with an expected vertex degree of $4$ (which corresponds to a sparseness parameter of $s = 4/(p - 1)$) and simulated $6$ data sets for the parameters $k = 0.4$ and $n \in \{1000, 2000, 5000, 10000, 20000, 50000\}$ from each of these models.  We only used these additional data sets for the investigation of the runtime of GIES.

\subsubsection{Alternative Structure Learning Algorithms}
\label{sec:different-algorithms}

We compare GIES with three alternative greedy search algorithms.  The first one is the original GES of \citet{Chickering2002Optimal} which regards the complete interventional data set as observational (that is, ignores the list $\mathcal{T}$ of an interventional data set $(\mathcal{T}, \mathbf{X})$ as defined in Equation (\ref{eqn:dataset})).  The second one, which we call \textsc{GIES-nt} (for ``no turning''), is a variant of GIES that stops after the first forward and backward phase and lacks the turning phase.  The third algorithm, called GDS for ``greedy DAG search'', is a simple greedy algorithm optimizing the same score function as GIES, but working on the space of DAGs instead of the space of \mcI-essential graphs; GDS simply adds, removes or turns arrows of DAGs in the forward, backward and turning phase, respectively.  Furthermore, for $p \leq 20$, we compare with a dynamic programming (DP) approach proposed by \citet{Silander2006Simple}, an algorithm that finds a global optimum of any decomposable score function on the space of 
DAGs.  Because of the exponential growth in time and memory requirements, we could not calculate DP estimates for models with $p \geq 30$ variables.  For GDS and DP, we examine the \mcI-essential graph of the returned DAGs.

\subsubsection{Quality Measures for Estimated Essential Graphs}
\label{sec:quality-measures}

The \textbf{structural Hamming distance} or SHD (\citealp{Tsamardinos2006Maxmin}; we use the slightly adapted version of \citealp{Kalisch2007Estimating}) is used to measure the distance between an estimated \mcI-essential graph $\hat{G}$ and a true \mcI-essential graph or DAG $G$.  If $A$ and $\hat{A}$ denote the adjacency matrices of $G$ and $\hat{G}$, respectively, the SHD between $G$ and $\hat{G}$ reads
$$
    \SHD(\hat{G}, G) := \sum_{1 \leq i < j \leq p} \big(1 - \mathbbm{1}_{\{(A_{ij} = \hat{A}_{ij}) \wedge (A_{ji} = \hat{A}_{ji})\}}\big) \ .
$$
The SHD between $\hat{G}$ and $G$ is the sum of the numbers of false positives of the skeleton, false negatives of the skeleton, and wrongly oriented edges.  Those quantities are defined as follows.  Two vertices which are adjacent in $\hat{G}$ but not in $G$ count as one false positive, two vertices which are adjacent in $G$ but not in $\hat{G}$ as one false negative.  Two vertices which are adjacent in both $G$ and $\hat{G}$, but connected with different edge types (that is, by a directed edge in one graph, by an undirected one in the other; or by directed edges with different orientations in both graphs) constitute a \textbf{wrongly oriented} edge.

\subsubsection{Results and Discussion}
\label{sec:results}

As we mentioned in Section \ref{sec:essential-graphs-motivation}, the undirected edges in the \mcI-essential graph $\mathcal{E_I}(D)$ of some causal structure $D$ are the edges with unidentifiable orientation.  The number of undirected edges in $\mathcal{E_I}(D)$ analyzed in the next paragraph is therefore a good measure for the identifiability of $D$.  Later on, we study the performance of GIES and compare it to the other algorithms mentioned in Section \ref{sec:different-algorithms}.

\noindent{\emph Identifiability under Interventions}

\begin{figure}[b!]
    \centering \makebox{\includegraphics{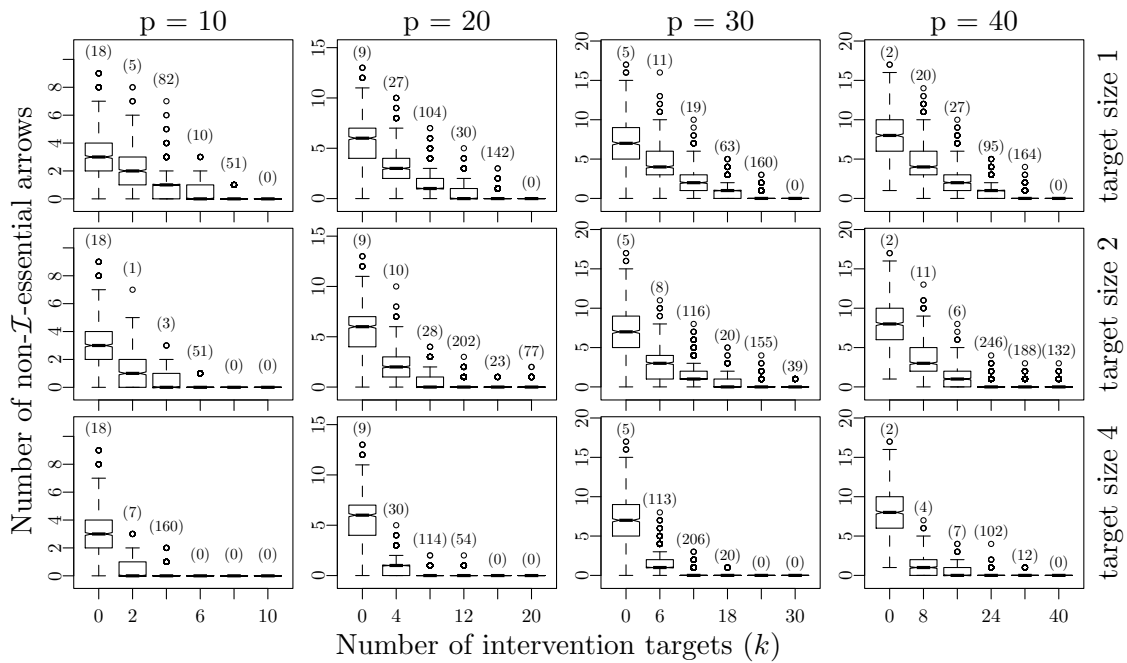}}
    \caption{Number of non-\mcI-essential arrows as a function of the number $k$ of intervention vertices.  In parentheses: number of outliers in the corresponding boxplot.}
    \label{fig:non-orientable-edges}
\end{figure}

In Figure \ref{fig:non-orientable-edges}, the number of non-\mcI-essential arrows is plotted as a function of the number $k$ of non-empty intervention targets ($k = |\mcI| - 1$, see Section \ref{sec:model-generation}).  With single-vertex interventions at $80\%$ of the vertices, the majority of the DAGs used in the simulation are completely identifiable; with target size $m = 2$ or $m = 4$, this is already the case for $k = 0.6p$ or $k = 0.4p$, respectively.  For the small target sizes used, the identifiability under $k$ targets of size $m$ is similar to the identifiability under $k \cdot m$ single-vertex targets.

A certain prudence is advisable when interpreting Figure \ref{fig:non-orientable-edges} since the number of orientable edges also reflects the characteristics of the generated DAGs.  Nevertheless, the plots show that the identifiability of causal models increases quickly even with few intervention targets.  In regard of applications this is an encouraging finding since it illustrates that even a small number of intervention experiments can strongly increase the identifiability of causal structures.

\noindent{\emph Performance of GIES}

\begin{figure}[t]
    \centering \makebox{\includegraphics{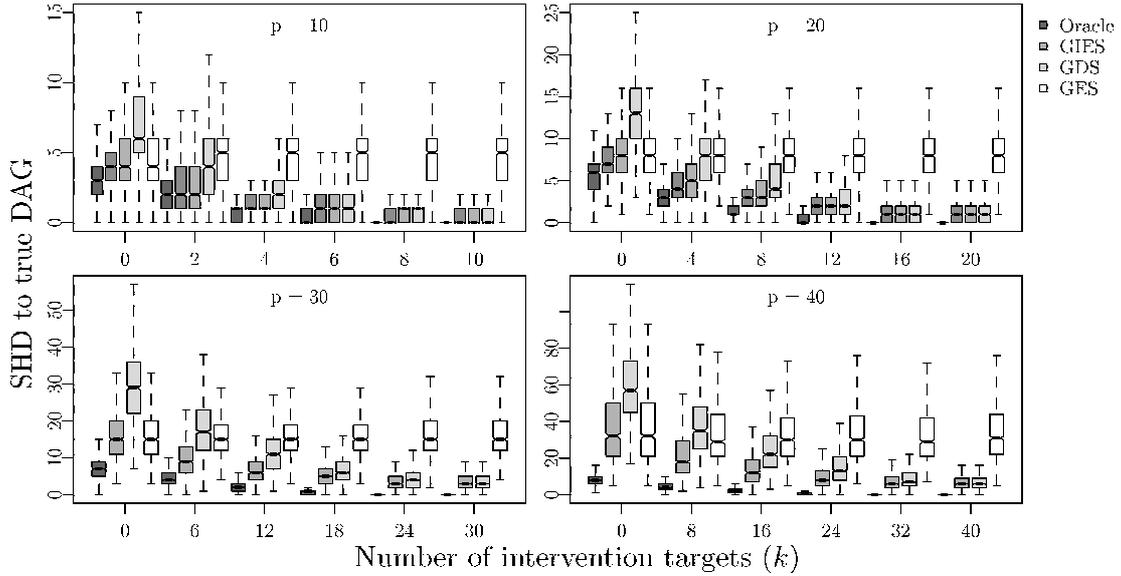}}
    \caption{SHD between \mcI-essential graph $\hat{G}$ estimated from $n = 1000$ data points and true DAG $D$ as a function of the number $k$ of single-vertex intervention targets.  ``Oracle estimates'' denote the respective true \mcI-essential graph $\mathcal{E_I}(D)$, the best possible estimate under some family of targets \mcI{} (see also Figure \ref{fig:non-orientable-edges}).  DP estimates are missing in the two lower plots.}
    \label{fig:shd-fix-n}
\end{figure}

Figure \ref{fig:shd-fix-n} shows the structural Hamming distance between true DAG $D$ and estimated \mcI-essential graph $\hat{G}$ for different algorithms as a function of the number $k$ of intervention targets.  Single-vertex interventions are considered; for larger targets, the overall picture is comparable (data not shown).  In $10$ out of $12$ cases for $p \leq 20$, the median SHD values of GIES and DP estimates are equal; in the remaining cases, too, GIES yields estimates of comparable quality---at much lower computational costs.

In parallel with the identifiability, the estimates produced by the different algorithms improve for growing $k$.  This illustrates that interventional data arising from different intervention targets carry \emph{more} information about the underlying causal model than observational data of the same sample size.

For complete interventions, that is, $k = p$, every DAG is completely identifiable and hence its own \mcI-essential graph.  Therefore, GDS and GIES are exactly the same algorithm in this case.  With shrinking $k$, the performance of GDS compared to that of GIES gets worse.  On the other hand, GES coincides with GIES in the observational case ($k = 0$).  For growing $k$, the estimation performance of GES stays approximately constant; it can, as opposed to GIES, not make use of the additional information coming from interventions.  To sum up, both the price of ignoring interventional Markov equivalence (GDS) and ignoring the interventional nature of the provided data sets (GES) are apparent in Figure \ref{fig:shd-fix-n}.

\begin{figure}[t]
    \centering \makebox{\includegraphics{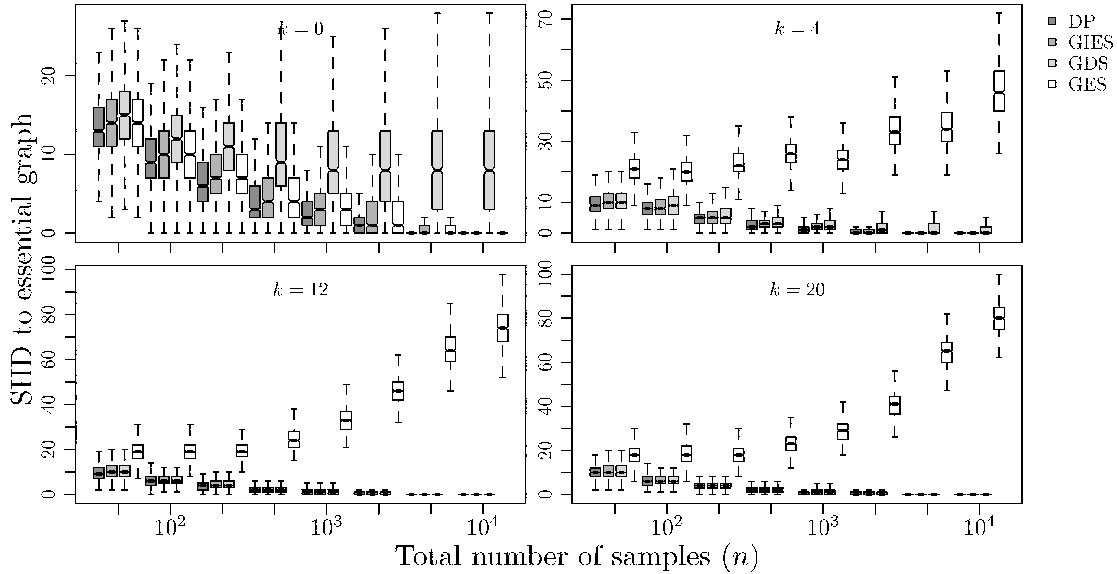}}
    \caption{SHD between estimated and true \mcI-essential graph for different numbers $k$ of intervention targets of size $m = 4$ for the DAGs with $p = 20$ vertices.  The abscissa denotes the \emph{total} sample size $n$.  For example, a data set with $n = 1000$ and $k = 4$ consists of $200$ observational samples and $200$ interventional samples each arising from interventions at four different targets, see Section \ref{sec:model-generation}.}
    \label{fig:shd-essential-varying-n}
\end{figure}

The performance of GIES as a function of the sample size $n$ is plotted in Figure \ref{fig:shd-essential-varying-n} for the DAGs with $p = 20$ vertices and intervention targets of size $m = 4$.  The quality of the GIES estimates is comparable to that of the DP estimates.  The behavior of the SHD values for growing $n$ is a strong hint for the consistency of GIES in the limit $n \to \infty$ \citep[note that the DP algorithm is consistent; ][]{Hauser2012Consistent}.  In contrast, the plots for $k = 0$ and $k = 4$ again reveal the weak performance of GDS for small numbers of intervention vertices; the plots suggest that GDS, in contrast to GIES, does not yield a consistent estimator of the \mcI-essential graph due to being stuck in a bad local optimum.

The most striking result in Figure \ref{fig:shd-essential-varying-n} is certainly the fact that the estimation performance of GES heavily decreases with growing $n$ as long as the data is not observational ($k > 0$).  This is not an artifact of GES, but a problem of model-misspecification: running DP for an \emph{observational} model (that is, considering all data as observational as GES does) yields SHD values maximally $14\%$ below that of GES (data not shown).  For single-vertex interventions, the SHD values of the GES estimates stay approximately constant with growing $n$; for target size $m = 2$, its SHD values also increase, but not to the same extent as for $m = 4$.

\begin{figure}[t]
    \centering \makebox{\includegraphics{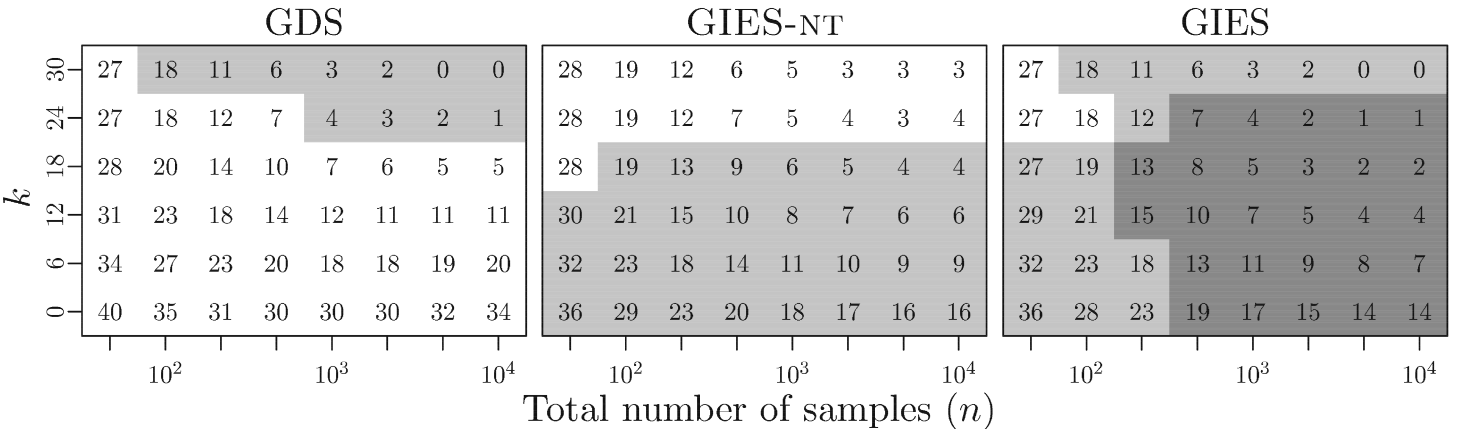}}
    \caption{Mean SHD between estimated and true \mcI-essential graph for different greedy algorithms as a function of $n$ and $k$; data for DAGs with $p = 30$ and single-vertex interventions.  Shading: algorithm yielded significantly better estimates than one (\textcolor{black!25}{$\blacksquare$}) or two (\textcolor{black!50}{$\blacksquare$}) of its competitors, respectively (paired $t$-test on a significance level of $\alpha = 5 \%$).}
    \label{fig:shd-comparison}
\end{figure}

\begin{figure}[t]
    \centering \makebox{\includegraphics{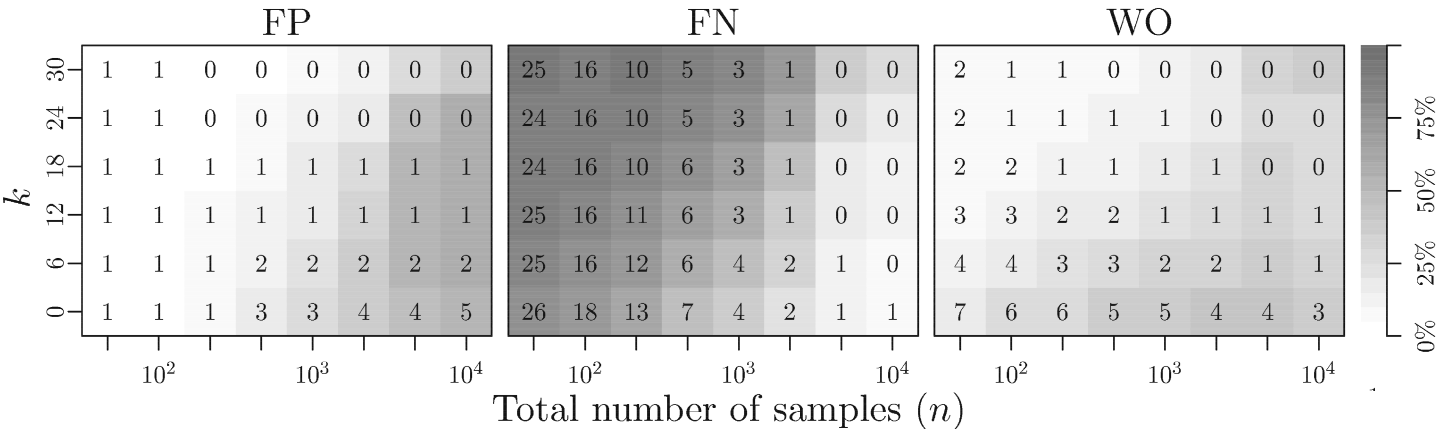}}
    \caption{False positives (FP) and false negatives (FN) of the skeleton and wrongly oriented edges (WO; Section \ref{sec:quality-measures}) of the GIES estimates compared to the true \mcI-essential graphs with $p = 30$ vertices; mean values as a function of $k$ and $n$ for single-vertex interventions.  Shading: ratio of each quantity and the SHD between estimated and true \mcI-essential graph (dark means a large contribution to the SHD).}
    \label{fig:shd-components}
\end{figure}

In Figure \ref{fig:shd-comparison}, we compare the SHD between true and estimated \mcI-essential graphs with $p = 30$ vertices for estimates produced by different greedy algorithms; other vertex numbers give a similar picture.  In most settings, GIES beats both GDS and \textsc{GIES-nt}.  It combines both the advantage of \textsc{GIES-nt}, using the space of interventional Markov equivalence classes as search space, and GDS, the turning phase apparently reducing the risk of getting stuck in local maxima of the score function.

As noted in Section \ref{sec:quality-measures}, the SHD between true and estimated interventional essential graphs can be written as the sum of false positives of the skeleton, false negatives of the skeleton and wrongly oriented edges.  Those numbers are shown in Figure \ref{fig:shd-components}, again for GIES estimates under single-vertex interventions for DAGs with $p = 30$ vertices.  False positives of the skeleton are the main contribution to the SHD values.  In $60 \%$ of the cases, especially for large $n$ and small $k$, wrongly oriented edges represent the second-largest contribution.

\noindent{\emph Runtime Analysis}

\begin{figure}
    \centering \includegraphics{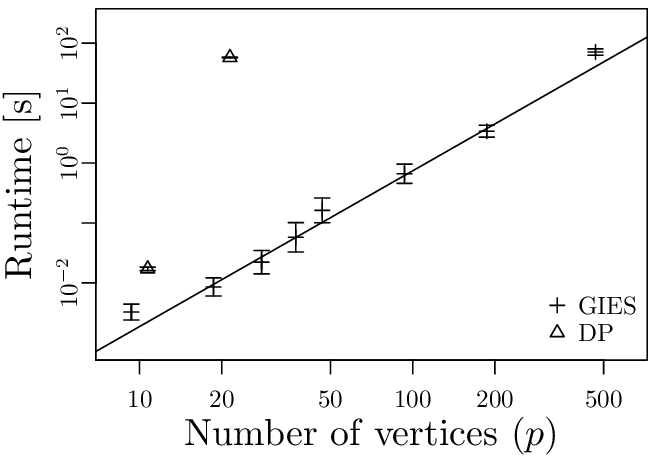}
    \caption{Runtime of GIES and DP as a function of the vertex number.}
    \label{fig:runtime}
\end{figure}

All algorithms evaluated in this section were implemented in C++ and compiled into a library using the GNU compiler g++ 4.6.1.  The simulations---that is, the generation of data and the library calls---were performed using R 2.13.1.  All simulations were run on an AMD Opteron 8380 CPU with \unit[2.5]{GHz} and \unit[2]{GB} RAM.

Figure \ref{fig:runtime} shows the running times of GIES and DP as a function of the number $p$ of vertices. GDS had running times of the same order of magnitude as GIES; they were actually up to $50\%$ higher since we used a basic implementation of GDS compared to an optimized version of GIES (running times of GDS are not plotted for this reason).  The linearity of the GIES values in the log-log plot (see the solid line in Figure \ref{fig:runtime}) indicate a polynomial time complexity of the approximate order $O(p^{2.8})$, in contrast to the exponential complexity of DP; note that GIES also has an exponential \emph{worst case} complexity (see Section \ref{sec:discussion}).  The multiple linear regression $\log(t) = \beta_0 + \beta_1 \log(p) + \beta_2 \log(|E|) + \varepsilon$, where $t$ denotes the runtime and $E$ the edge set of the true DAG, yields coefficients $\hat{\beta}_1 = 1.01$ and $\hat{\beta}_2 = 0.94$.

\subsection{DREAM4 Challenge}
\label{sec:dream4}

We also measured the performance of GIES on synthetic gene expression data sets from the DREAM4 \textit{in silico} challenge \citep{Marbach2010Revealing, Prill2010Towards}.  Our goal here was to evaluate predictions of expression levels of gene knockout or knockdown experiments by cross-validation based on the provided interventional data.

\subsubsection{Data}
\label{sec:dream4-data}

The DREAM4 challenge provides five data sets with an ensemble of interventional and observational data simulated from five biologically plausible, possibly \emph{cyclic} gene regulatory networks with $10$ genes \citep{Marbach2009Generating}.  The data set of each network consists of
\begin{itemize}
    \item $11$ observational measurements, simulated from random fluctuations of the system parameters (resembling observational data measured in different individuals);
    
    \item $10$ measurements from single-gene knockdowns, one knockdown per gene;
    
    \item $10$ measurements from single-gene knockouts, one knockout per gene;
    
    \item five time series with $21$ time points each, simulated from an unknown change of parameters in the first half (corresponding to measurements under a perturbed chemical environment having unknown effects on the gene regulatory network) and from the unperturbed system in the second half.
\end{itemize}
Since our framework can not cope with uncertain interventions (that is, interventions with unknown target), we only used the $50$ observational measurements of the second half of the time series.  Altogether, we have, from each network, a total of $81$ data points, $61$ observational and $20$ interventional ones.  We normalized the data such that the observational samples of each gene have mean $0$ and standard deviation $1$.  In this normalization, $95\%$ of the intervention levels (that is, the expression levels of knocked out or knocked down genes) lie between $-8.37$ and $-0.62$ with a mean of $-3.30$ (Figure \ref{fig:dream4-int-levels}).

\begin{figure}[t]
    \centering
    \includegraphics{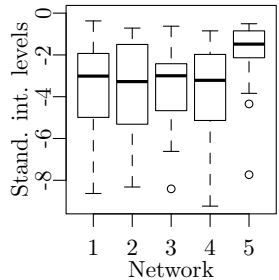}
    \caption{Standardized intervention levels in the different DREAM4 data sets.  Data is scaled such that the observational samples have empirical mean $0$ and standard deviation $1$.}
    \label{fig:dream4-int-levels}
\end{figure}

\subsubsection{Methods}
\label{sec:dream4-methods}

We used each interventional measurement ($20$ per network) as one test data point and predicted its value from a network estimated with training data consisting either of the $80$ remaining data points, or the $61$ observational measurements alone.  We used GIES, GES and PC \citep{Spirtes2000Causation} to estimate the causal models and evaluated the prediction accuracy by the mean squared error (MSE).  We will use abbreviations like ``GES($80$)'' or ``PC($61$)'' to denote GES estimates based on a training set of size $80$ or PC estimates based on an observational training set of size $61$, respectively.

For a given DAG, we predicted interventional gene expression levels based on the estimated structural equation model after replacing the structural equation of the intervened variable by a constant one; see Section \ref{sec:gaussian-models} for connection between Gaussian causal models and structural equations, especially Equation (\ref{eqn:structural-equations}).  GES and PC regard all data as observational and yield an observational essential graph.  For those algorithms, we enumerated all representative DAGs of the estimated equivalence class using the function \texttt{allDags()} of the R package \texttt{pcalg} \citep{Kalisch2012Causal}, calculated an expression level with each of them, and took the mean of those predictions.  GIES($80$) yields a single DAG in each case since the $19$ interventional measurements in the training data ensure complete identifiability.

Furthermore, we used the evaluation script provided by the DREAM4 challenge to assess the quality of our network predictions to those sent in to the challenge by participating teams.  This evaluation is based on the area under the ROC curve (AUROC) of the true and false positive rate of the edge predictions.

\subsubsection{Results}
\label{sec:dream4-results}

\begin{figure}[t]
    \includegraphics{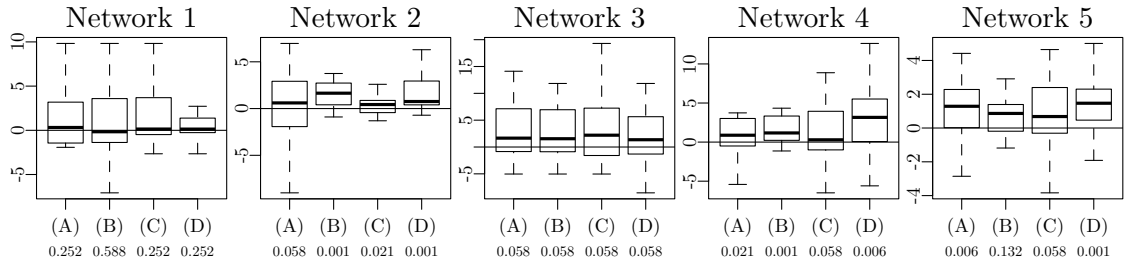}
    \caption{Upper row: MSE values of GIES and competitors; lower row: differences of MSE values as defined in Equation (\ref{eqn:mse-diff}); large values indicate a good performance of GIES.  (A) GIES($80$), (B) PC($80$), (C) PC($61$), (D) GES($80$), (E) GES($61$).  Numbers below the boxplots: p-values of a one-sided sign test.}
    \label{fig:dream4-mse-diff}
\end{figure}

Figure \ref{fig:dream4-mse-diff} shows boxplots of MSE differences between GIES($80$) and its competitors; that is, we consider quantities of the form
\begin{equation}
    \Delta\mathrm{MSE}\subscr{comp} := \mathrm{MSE}\subscr{comp} - \mathrm{MSE}\subscr{GIES(80)}, \label{eqn:mse-diff}
\end{equation}
where comp stands for one of the competitors.  Since the MSE differences showed a skewed distribution in general, we used a sign test for calculating their p-values.

Except for one case (PC($61$) in network 1), GIES($80$) always yielded the best predictions of all competitors.  Although all data sets are dominated by observational data ($61$ observational measurements versus $20$ interventional ones), GIES can make use of the additional information carried by interventional data points to rule out its observational competitors.  On the other hand, the dominance of observational data is probably one of the reasons for the fact that GIES does not outperform the observational methods more clearly but has an overall performance which is comparable with that of its competitors.  Another reason could be the fact that the underlying networks used for data generation are not acyclic as assumed by GIES.  Interestingly, the winning margin of GIES in network 5 was not smaller than in other networks although the corresponding data set has the smallest intervention levels (in absolute values; see Figure \ref{fig:dream4-int-levels}).

29 teams participated in the DREAM4 challenge.  Their AUROC values are available from the DREAM4 website;\footnote{DREAM4 can be found at \url{http://wiki.c2b2.columbia.edu/dream/index.php/D4c2}.} adding our values gives a data set of 30 evaluations.  Among those, our results had overall rank 10, and ranks 8, 4, 21, 10 and 3, respectively, for networks 1 to 5.  Except for network 3, we could keep up with the best third of the participating teams despite the beforementioned model misspecification given by the assumption of acyclicity, and despite the fact that we ignored the time series structure and half of the time series data.

\section{Conclusion}
\label{sec:conclusion}

We gave a definition and a graph theoretic criterion for the Markov equivalence of DAGs under multiple interventions.  We characterized corresponding equivalence classes by their \emph{essential graph}, defined as the union of all DAGs in an equivalence class in analogy to the observational case.  Using those essential graphs as a basis for the algorithmic representation of interventional Markov equivalence classes, we presented a new greedy algorithm (including a new turning phase), GIES, for learning causal structures from data arising from multiple interventions.

In a simulation study, we showed that the number of non-orientable edges in causal structures drops quickly even with a small number of interventions; our description of interventional essential graphs makes it possible to \emph{quantify} the gain in identifiability.  For a fixed sample size $n$, GIES estimates got closer to the true causal structure as the number of intervention vertices grew.  For DAGs with $p \leq 20$ vertices, the GIES algorithm could keep up with a consistent, exponential-time DP approach maximizing the BIC score.  It clearly beat GDS, a simple greedy search on the space of DAGs, as well as GES which cannot cope with interventional data.  Our novel turning phase proved to be an improvement of GES even on observational data, as it was already conjectured by \citet{Chickering2002Optimal}.  Applying GIES to synthetic data sets from the DREAM4 challenge \citep{Marbach2010Revealing}, we got better predictions of gene expression levels of knockout or knockdown experiments than with observational estimation methods.

The accurate structure learning performance of GIES in the limit of large data sets raises the question whether GIES is consistent.  \citet{Chickering2002Optimal} proved the consistency of GES on observational data.  However, the generalization of his proof for GIES operating on interventional data is not obvious since such data are in general not identically distributed.


\section*{Acknowledgments}
We are grateful to Markus Kalisch for very carefully reading the proofs and for his helpful comments on the text.  Furthermore, we would like to thank the anonymous reviewers for constructive comments.

\appendix
\section{Graphs}
\label{sec:graphs}

In this appendix, we shortly summarize our notation \citep[mostly following][]{Andersson1997Characterization} and basic facts concerning graphs.  All statements about perfect elimination orderings that are used in Sections \ref{sec:essential-graphs} and \ref{sec:greedy-search} are listed or proven in Section \ref{sec:perfect-elimination-orderings}.

\subsection{Definitions and Notation}
\label{sec:graphs-notation}

A \textbf{graph} is a pair $G = (V, E)$, where $V$ is a finite set of vertices and $E \subset E^*(V) := (V \times V) \setminus \{(a, a) | a \in V\}$ is a set of edges.  We use graphs to denote causal relationships between random variables $X_1, \ldots, X_p$.  To keep notation simple, we always assume $V = \{1, 2, \ldots, p\} =: [p]$, in order to represent each random variable by its index in the graph.

An edge $(a, b) \in E$ with $(b, a) \in E$ is called \textbf{undirected} (or a \textbf{line}), whereas an edge $(a, b) \in E$ with $(b, a) \notin E$ is called \textbf{directed} (or an \textbf{arrow}).  Consequently, a graph $G$ is called directed (or undirected, resp.) if all its edges are directed (or undirected, resp.); a directed graph is also called \textbf{digraph} for short.  We use the short-hand notation
\begin{eqnarray*}
    a \grarright b \in G & :\Leftrightarrow & (a, b) \in E \wedge (b, a) \notin E, \\ 
    a \grline b \in G & :\Leftrightarrow & (a, b) \in E \wedge (b, a) \in E, \\
    a \grdots b \in G & :\Leftrightarrow & (a, b) \in E \vee (b, a) \in E.
\end{eqnarray*}

A \textbf{subgraph} of some graph $G$ is a graph $G' = (V', E')$ with the property $V' \subset V$, $E' \subset E$, denoted by $G' \subset G$.  For a subset $A \subset V$ of the vertices of $G$, the \textbf{induced subgraph} on $A$ is $G[A] := (A, E[A])$, where $E[A] := E \cap (A \times A)$.  A \textbf{v-structure} \citep[also called \emph{immorality} by, for example,][]{Lauritzen1996Graphical} is an induced subgraph of $G$ of the form $a \grarright b \grarleft c$.  The \textbf{skeleton} of a graph $G$ is the undirected graph $G^u := (V, E^u)$, $E^u := \{(a, b) \in V \times V \spst a \grdots b \in G\}$.  For two graphs $G_1 = (V, E_1)$ and $G_2 = (V, E_2)$ on the same vertex set, we define the union and the intersection as $G_1 \cup G_2 := (V, E_1 \cup E_2)$ and $G_1 \cap G_2 := (V, E_1 \cap E_2)$, respectively.  For a graph $G = (V, E)$ and $(a, b) \in E^*(V)$, we use the shorthand notation $G - (a, b) := (V, E \setminus \{(a, b)\})$ and $G + (a, b) := (V, E \cup \{(a, b)\})$.

The following sets describe the local environment of a vertex $a$ in a graph $G$:
\begin{eqnarray*}
    \pa_G(a) & := & \{b \in V \spst b \grarright a \in G\}, \text{ the \textbf{parents} of } a,\\
    \ch_G(a) & := & \{b \in V \spst a \grarright b \in G\}, \text{ the \textbf{children} of } a,\\
    \nb_G(a) & := & \{b \in V \spst a \grline b \in G\}, \text{ the \textbf{neighbors} of } a,\\
    \ad_G(a) & := & \{b \in V \spst a \grdots b \in G\}, \text{ the vertices \textbf{adjacent} to } a. \\
\end{eqnarray*}
The subscripts ``$G$'' in the above definitions are omitted when it is clear which graph is meant.  For a set $A \subset V$ of vertices, we generalize those definitions as follows:
$$
    \pa_G(A) := \bigcup_{a \in A} \pa_G(a) \setminus A, \quad \nb_G(A) := \bigcup_{a \in A} \nb_G(a) \setminus A, \ \text{ etc.}
$$
The \textbf{degree} of a vertex $a \in V$ is defined as $\deg_G(a) := |\ad_G(a)|$.

For two distinct vertices $a$ and $b \in V$, a \textbf{chain} of length $k$ from $a$ to $b$ is a sequence of distinct vertices $\gamma = (a \equiv a_0, a_1, \ldots, a_k \equiv b)$ such that for each $i = 1, \ldots, k$, either $a_{i-1} \grarright a_i \in G$ or $a_{i-1} \grarleft a_i \in G$; if for all $i$, $(a_{i-1}, a_i) \in E$ (that is, $a_{i-1} \grarright a_i \in G$ or $a_{i-1} \grline a_i \in G$), the sequence $\gamma$ is called a \textbf{path}.  If at least one edge $a_{i-1} \grarright a_i$ is directed in a path, the path is called \emph{directed}, otherwise \emph{undirected}.  A \textbf{(directed) cycle} is defined as a (directed) path with the difference that $a_0 = a_n$.  Paths define a preorder on the vertices of a graph: $a \preceq_G b :\Leftrightarrow \spexists$ a path $\gamma$ from $a$ to $b$ in $G$.  Furthermore, $a \approx_G b :\Leftrightarrow (a \preceq_G b) \wedge (b \preceq_G a)$ is an equivalence relation on the set of vertices.

An undirected graph $G = (V, E)$ is \textbf{complete} if all pairs of vertices are adjacent.  A \textbf{clique} is a subset of vertices $C \subset V$ such that $G[C]$ is complete; a vertex $a \in V$ is called \textbf{simplicial} if $\nb(a)$ is a clique.  An undirected graph $G$ is called \textbf{chordal} if every cycle of length $k \geq 4$ contains a \textbf{chord}, that means two nonconsecutive adjacent vertices.  For pairwise disjoint subsets $A, B, S \subset V$ with $A \neq \emptyset$ and $B \neq \emptyset$, $A$ and $B$ are \textbf{separated} by $S$ in $G$ if every path from a vertex in $A$ to a vertex in $B$ contains a vertex in $S$.

A \textbf{directed acyclic graph}, or \textbf{DAG} for short, is a digraph that contains no cycle.  In the paper, we mostly use the symbol $D$ for DAGs, whereas arbitrary graphs are, as in this appendix, mostly named $G$.  Chain graphs can be viewed as something between undirected graphs and DAGs: a graph $G = (V, E)$ is a \textbf{chain graph} if it contains no directed cycle; undirected graphs and DAGs are special cases of chain graphs.  The equivalence classes in $V$ w.r.t.\ the equivalence relation $\approx_G$ are the connected components of $G$ after removing all directed edges.  We denote the quotient set of $V$ by $\mathbf{T}(G) := V / \approx_G$, and its members $T \in \mathbf{T}(G)$ are called \textbf{chain components} of $G$.  For a vertex $a \in V$, $T_G(a)$ stands for $[a]_{\approx_G}$.  The preorder $\preceq_G$ on $V$ induces in a canonical way a \emph{partial order} on $\mathbf{T}(G)$ which we also denote by $\preceq_G$: $T_G(a) \preceq_G T_G(b) :\Leftrightarrow a \preceq_G b$.  An illustration 
is shown in Figure \ref{fig:ex-chain-graph}.

\begin{figure}[t]
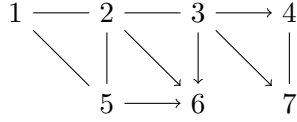

    \centering
    \begin{exgraphpicture}
        \draw[-] (v2) -- (v5) -- (v1) -- (v2) -- (v3);
        \draw[-] (v4) -- (v7);
        \draw[->] (v5) -- (v6);
        \draw[->] (v2) -- (v6);
        \draw[->] (v3) -- (v6);
        \draw[->] (v3) -- (v4);
        \draw[->] (v3) -- (v7);
    \end{exgraphpicture}
    \vspace{10mm}
    \caption{A chain graph $G$ with three chain components $A = T_G(1) = [1]_{\approx_G} = \{1, 2, 3, 5\}$, $B = T_G(6) = \{6\}$ and $C = T_G(4) = \{4, 7\}$.  The arrows induce the partial order $A \preceq_G B$, $A \preceq_G C$.  The graph is no chain graph anymore when we replace the arrow $3 \grarright 4$ by a line since this would create a directed cycle: $(3, 7, 4, 3)$.}
    \label{fig:ex-chain-graph}
\end{figure}

An \textbf{ordering} of a graph is a bijection $[p] \to V$, hence, since we assume $V = [p]$ here, a permutation $\sigma \in S_p$.  An ordering $\sigma$ canonically induces a total order on $V$ by the definition $a \leq_\sigma b :\Leftrightarrow \sigma^{-1}(a) \leq \sigma^{-1}(b)$.  An ordering $\sigma = (v_1, \ldots, v_p)$ is called a \textbf{perfect elimination ordering} if for all $i$, $v_i$ is simplicial in $G^u[\{v_1, \ldots, v_i\}]$.  A graph $G = (V, E)$ is a DAG if and only if the previously defined preorder $\preceq_G$ is a partial order; such a partial order can be extended to a total order \citep{Szpilrajn1930Extension}.  Thus every DAG has at least one \textbf{topological ordering}, that is an ordering $\sigma$ whose total order $\leq_\sigma$ extends $\preceq_G$: $a \preceq_G b \Rightarrow a \leq_\sigma b$.  For $\sigma \in S_p$, a DAG $D = ([p], E)$ is said to be \textbf{oriented according to $\sigma$} if $\sigma$ is a topological ordering of $D$.  In a DAG $D$ with topological ordering $\sigma$,
 the arrows point from vertices with low to vertices with high ordered indices.  The vertex $\sigma(1)$ is a \textbf{source}, that means all arrows point away from it.

\subsection{Perfect Elimination Orderings}
\label{sec:perfect-elimination-orderings}

Perfect elimination orderings play an important role in the characterization of interventional Markov equivalence classes of DAGs as well as in the implementation of the Greedy Interventional Equivalence Search (GIES).  In this section, we provide all results for this topic that are used as auxiliary tools in the proofs of Sections \ref{sec:essential-graphs} and \ref{sec:greedy-search}.

\begin{lemma}
    \label{lem:dag-perfect-ordering}
    Let $D = (V, E)$ be a DAG.  $D$ has no v-structures if and only if any topological ordering of $D$ is a perfect elimination ordering.
\end{lemma}
The proof of this lemma follows easily from the definitions of a v-structure and a perfect elimination ordering.  Moreover, if \emph{any} topological ordering of a DAG is a perfect elimination ordering, this is automatically the case for \emph{every} topological ordering.

\begin{proposition}[\citealp{Rose1970Triangulated}]
    \label{prop:chordal-perfect-ordering}
    Let $G = (V, E)$ be an undirected graph.  Then $G$ is chordal if and only if it has a perfect elimination ordering.
\end{proposition}

\begin{algorithm}[t!]
    \caption{$\LexBFS(V, E)$.  Lexicographic breadth-first search in the so-called ``partitioning paradigm'' \citep{Rose1976Algorithmic, Corneil2005Lexicographic}}
    \label{alg:lex-bfs}
    \fontsize{10}{12}\selectfont
    \SetKwInOut{Input}{Input}
    \SetKwInOut{Output}{Output}
    \Input{An undirected graph $G = (V, E)$}
    \Output{An ordering $\sigma$ of the vertices $V$, called a \textbf{\LexBFS-ordering}}
    $\Sigma \leftarrow (V)$; \tcp{Initialize sequence $\Sigma$ of vertex sets to contain the single set $V$ in the beginning}
    $\sigma \leftarrow ()$; \tcp{Initialize output sequence of vertices}
    \ShowLnLabel{ln:while-start}\While{$\Sigma \neq \emptyset$}{%
        \ShowLnLabel{ln:remove} Remove a vertex $a$ from the first set in the sequence $\Sigma$\; 
        \lIf{first set of $\Sigma$ is empty}{remove first set from $\Sigma$}\;
        Append $a$ to $\sigma$\;
        Mark all sets of $\Sigma$ as not visited\;
        \ForEach{$b \in \nb_G(a)$ s.t. $b \in S$ for some $S \in \Sigma$}{%
            \If{$S$ not visited}{%
                Insert empty set $T$ into $\Sigma$ in front of $S$\;
                Mark $S$ as visited\;
            }
            \lElse{let $T$ be the set preceding $S$ in $\Sigma$}\;
            \ShowLnLabel{ln:move} Move $b$ from $S$ to $T$\;
            \ShowLnLabel{ln:while-end}\lIf{$S = \emptyset$}{remove $S$ from $\Sigma$}\;
        }
    }
\end{algorithm}

Perfect elimination orderings of chordal graphs can be produced by a variant of the breadth-first search algorithm, the so-called lexicographic breadth-first search (\LexBFS; see Algorithm \ref{alg:lex-bfs}).  The term ``lexicographic'' reflects the fact that the algorithm visits edges in lexicographic order w.r.t.\ the produced ordering $\sigma$.

\begin{proposition}[\citealp{Rose1976Algorithmic}]
    \label{prop:lex-bfs-perfect-ordering}
    Let $G = (V, E)$ be an undirected chordal graph with a \LexBFS-ordering $\sigma$.  Then $\sigma$ is also a perfect elimination ordering on $G$.
\end{proposition}

\begin{corollary}
    \label{cor:chordal-dag}
    Let $G$ be an undirected chordal graph with a \LexBFS-ordering $\sigma$.  A DAG $D \subset G$ with $D^u = G$ that is oriented according to $\sigma$ has no v-structures.
\end{corollary}
Corollary \ref{cor:chordal-dag} is a consequence of  Lemma \ref{lem:dag-perfect-ordering} and Proposition \ref{prop:lex-bfs-perfect-ordering}.  According to this corollary, \LexBFS-orderings can be used for constructing representatives of essential graphs (see Proposition \ref{prop:construction-representative}).  Corollary \ref{cor:chordal-dag} as well as Algorithm \ref{alg:lex-bfs} are therefore of great importance for the proofs and algorithms of Sections \ref{sec:essential-graphs} and \ref{sec:greedy-search}.

Figure \ref{fig:ex-lex-bfs} shows an undirected chordal graph $G$ and a DAG $D$ that has the skeleton $G$ and is oriented according to a \LexBFS-ordering $\sigma$ of $G$.  The functioning of Algorithm \ref{alg:lex-bfs} when producing a \LexBFS-ordering on $G$ is illustrated in Table \ref{tab:ex-lex-bfs}.  Note that the ``sets'' in $\Sigma$ are written as tuples.  We use this notation to ensure that we can always remove the first (leftmost) vertex from the first ``set'' of $\Sigma$ (line \ref{ln:remove} in Algorithm \ref{alg:lex-bfs}), and that we keep the relative order of vertices when moving them from one set $S$ to the preceding one, $T$, in $\Sigma$ (line \ref{ln:move} in Algorithm \ref{alg:lex-bfs}).  Throughout the text, we always assume an implementation of Algorithm \ref{alg:lex-bfs} in which the data structure used to represent the ``sets'' in the sequence $\Sigma$ guarantees this ``first in, first out'' (FIFO) behavior.  In particular, the start sequence $(v_1, v_2, \ldots, v_p)$ of the vertices in 
$V$ provided to the algorithm determines the vertex the \LexBFS-ordering $\sigma := \LexBFS((v_1, \ldots, v_p), E)$ starts with: $\sigma(1) = v_1$.  It is often sufficient to specify the start order of \LexBFS{} up to arbitrary orderings of some subsets of vertices.  For a set $A = \{a_1, \ldots, a_k\} \subset V$ and an additional vertex $v \in V \setminus A$, for example, we use the notation
$$
    \LexBFS((A, v, V \setminus (A \cup \{v\})), E), \quad \text{ or even } \quad \LexBFS((A, v, \ldots), E)
$$
to denote a \LexBFS-ordering produced from a start order of the form $(a_1, \ldots, a_k, v, \ldots)$, without specifying the orderings of $A$ and $V \setminus (A \cup \{v\})$.

By using appropriate data structures (for example, doubly linked lists for the representation of $\Sigma$ and its sets, and a pointer at each vertex pointing to the set in $\Sigma$ in which it is contained), Algorithm \ref{alg:lex-bfs} has complexity $O(|E| + |V|)$ \citep{Corneil2005Lexicographic}.

For the rest of this section, we state further consequences of Lemma \ref{lem:dag-perfect-ordering} and Proposition \ref{prop:lex-bfs-perfect-ordering} which are relevant for the proofs of Sections \ref{sec:essential-graphs} and \ref{sec:greedy-search}.

\begin{figure}[t]
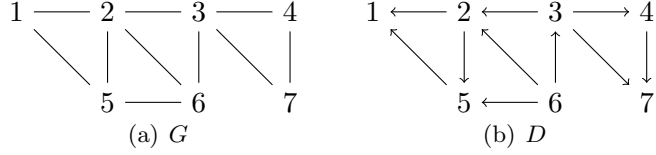

    \centering
    \subfigure[$G$]{%
        \begin{exgraphpicture}
            \draw (v5) -- (v1) -- (v2) -- (v5) -- (v6) -- (v2) -- (v3) -- (v4) -- (v7) -- (v3) -- (v6);
        \end{exgraphpicture}
    } \quad
    \subfigure[$D$]{%
        \begin{exgraphpicture}
            \draw[->] (v2) -- (v1);
            \draw[->] (v5) -- (v1);
            \draw[->] (v3) -- (v2);
            \draw[->] (v6) -- (v2);
            \draw[->] (v6) -- (v3);
            \draw[->] (v3) -- (v4);
            \draw[->] (v2) -- (v5);
            \draw[->] (v6) -- (v5);
            \draw[->] (v3) -- (v7);
            \draw[->] (v4) -- (v7);
        \end{exgraphpicture}
    }
    \caption{An undirected, chordal graph $G = ([7], E)$ and the DAG $D$ we get by orienting all edges of $G$ according to the ordering $\sigma := \LexBFS((6, 3, 1, 2, 4, 5, 7), E)$.}
    \label{fig:ex-lex-bfs}
\end{figure}

\begin{table}[t]
    \centering
    \begin{tabular}{lll}
        $i$ & $\Sigma$ & $\sigma$ \\
        \hline
        $0$ & $((6, 3, 1, 2, 4, 5, 7))$     & $()$ \\
        $1$ & $((3, 2, 5), (1, 4, 7))$      & $(6)$ \\
        $2$ & $((2), (5), (4, 7), (1))$     & $(6, 3)$ \\
        $3$ & $((5), (4, 7), (1))$          & $(6, 3, 2)$ \\
        $4$ & $((4, 7), (1))$               & $(6, 3, 2, 5)$ \\
        $5$ & $((7), (1))$                  & $(6, 3, 2, 5, 4)$ \\
        $6$ & $((1))$                       & $(6, 3, 2, 5, 4, 7)$ \\
        $7$ & $()$                          & $(6, 3, 2, 5, 4, 7, 1)$
    \end{tabular}
    \caption{State of the sequences $\Sigma$ and $\sigma$ after the $i\supscr{th}$ run ($i = 0, \ldots, 7$) of the while loop (lines \ref{ln:while-start} to \ref{ln:while-end}) of Algorithm \ref{alg:lex-bfs} applied to the graph $G$ of Figure \ref{fig:ex-lex-bfs} with start order $(6, 3, 1, 2, 4, 5, 7)$.}
    \label{tab:ex-lex-bfs}
\end{table}

\begin{corollary}
    \label{cor:chordal-graph-different-orientations}
    Let $G = (V, E)$ be an undirected chordal graph, and let $a \grline b \in G$.  There exist DAGs $D_1$ and $D_2$ with $D_1, D_2 \subset G$ and $D_1^u = D_2^u = G$ without v-structures such that $a \grarright b \in D_1$ and $a \grarleft b \in D_2$.
\end{corollary}

\begin{proof}
    Set $\sigma_1 := \LexBFS((a, V \setminus \{a\}), E)$ and $\sigma_2 := \LexBFS((b, V \setminus \{b\}), E)$, and let $D_1$ and $D_2$ be two DAGs with skeleton $G$ and oriented according to $\sigma_1$ and $\sigma_2$, resp.  Then, by Corollary \ref{cor:chordal-dag}, $D_1$ and $D_2$ have the requested properties; in particular, all edges point away from $a$ in $D_1$, whereas all edges point away from $b$ in $D_2$.
\end{proof}
\vspace{-3mm}
\begin{corollary}[\citealp{Andersson1997Characterization}]
    \label{cor:lex-bfs-clique}
    Let $G = (V, E)$ be an undirected chordal graph, $a \in V$ and $C \subset \nb(a)$.  Then there is a DAG $D \subset G$ with $D^u = G$ and $\{b \in \nb(a) \spst b \grarright a \in D\} = C$ that has no v-structures if and only if $C$ is a clique.
\end{corollary}
\vspace{-2mm}
\begin{proof}
    \emph{``$\Rightarrow$'':} Assume that there are non-adjacent vertices $b, c \in C$.  Then, $b \grline a \grline c$ is an induced subgraph of $G$, and by construction, the same vertices occur in configuration $b \grarright a \grarleft c$ in $D$, which means that $D$ has a v-structure, a contradiction.
    
    \proofitem{``$\Leftarrow$''} Let $(c_1, \ldots, c_k)$ be an arbitrary ordering of $C$.  Run \LexBFS{} on a start order of the form $(c_1, \ldots, c_k, a, \ldots)$.  After the first run of the while loop (lines \ref{ln:while-start} to \ref{ln:while-end} of Algorithm \ref{alg:lex-bfs}), $\sigma = (c_1)$, and the first set in the sequence $\Sigma$ contains $(C \cup \{a\}) \setminus \{c_1\}$ as a subset (all vertices in this set are adjacent to $c_1$), in an unchanged order $c_2, \ldots, c_k, a$ due to our FIFO convention.  After the second run of the while loop, $\sigma = (c_1, c_2)$, and the first set in $\Sigma$ contains $(C \cup \{a\}) \setminus \{c_1, c_2\}$, and so on.  In the end, we get a \LexBFS-ordering of the form $\sigma = (c_1, \ldots, c_k, a, \ldots)$.  Orienting the edges of $G$ according to $\sigma$ yields a DAG with the requested properties by Corollary \ref{cor:chordal-dag}.
\end{proof}

\begin{figure}[h!]
    \centering
    \includegraphics{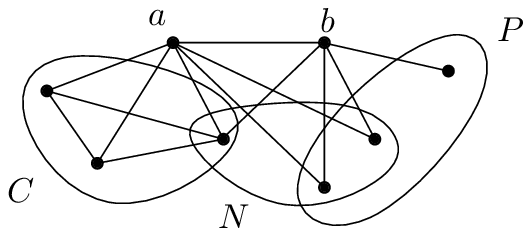}
    \caption{Configuration of vertices in Proposition \ref{prop:lex-bfs-clique-neighborhood}.}
    \label{fig:configuration-neighborhood}
\end{figure}

\begin{proposition}
    \label{prop:lex-bfs-clique-neighborhood}
    Let $G = (V, E)$ be an undirected, chordal graph, $a \grline b \in G$, and $C \subset \nb_G(a) \setminus \{b\}$ a clique.  Let $N := \nb_G(a) \cap \nb_G(b)$, and assume that $C \cap N$ separates $C \setminus N$ and $N \setminus C$ in $G[\nb_G(a)]$ (see Figure \ref{fig:configuration-neighborhood}).  Then there exists a DAG $D \subset G$ with $D^u = G$ such that
    \begin{subprop}
        \item \label{itm:no-v-structures} $D$ has no v-structures;
        \item \label{itm:orientation-inside-C} all edges in $D[C \cup \{a\}]$ point towards $a$;
        \item \label{itm:orientation-outside-C} all other edges of $D$ point away from vertices in $C \cup \{a\}$ (in particular, $a \grarright b \in D$);
        \item \label{itm:orientation-neighbors-b} $b \grarright d \in D$ for all $d \in P := \nb_G(b) \setminus (C \cup \{a\})$.
    \end{subprop}
\end{proposition}

\begin{proof}
    Set $\sigma := \LexBFS((C, a, b, \ldots), E)$, and let $D$ be the DAG that we get by orienting the edges of $G$ according to $\sigma$.  As in Corollary \ref{cor:lex-bfs-clique}, properties \ref{itm:no-v-structures} to \ref{itm:orientation-outside-C} are met.
    
    It remains to show that $b$ occurs before any $d \in P$ in $\sigma$ (that means $b <_\sigma d \spforall d \in P$) in order that $D$ obeys property \ref{itm:orientation-neighbors-b}.   W.l.o.g., we can assume $C = \{1, 2, \ldots, k\}$, $a = k + 1$ and $b = k + 2$.  The start order of the vertices for \LexBFS{} is then $(1, 2, \ldots, p)$.  Due to the FIFO convention for the sets of the sequence $\Sigma$ in Algorithm \ref{alg:lex-bfs}, $b$ always precedes any $d \in P$ whenever they appear in the same set; hence we only must show that the set containing $b$ is never preceded by a set containing some $d \in P$ in $\Sigma$.
    
    Suppose, for the sake of contradiction, that this is the case for some $d \in P$; name $v_1 := d$.  At the beginning, $b$ is in the same set as $v_1$ in the sequence $\Sigma$; there is some vertex $v_2$ that forces \LexBFS{} to move $v_1$ into the set preceding the one containing $b$.  A careful inspection of Algorithm \ref{alg:lex-bfs} shows that $v_2$ is the vertex which is minimal w.r.t.\ $\leq_\sigma$ in
    $$
        S(v_1) := \{v \in V \spst v \in \nb_G(v_1) \setminus \nb_G(b), v <_\sigma b\}.
    $$
    If $v_2 > b$ (that is, if $v_2 \notin C \cup \{a\}$ due to our convention), $v_2$, as $v_1$, always follows $b$ whenever they are in the same set in $\Sigma$.  Therefore, $v_2 <_\sigma b$ implies that there is some vertex $v_3$ that moves $v_2$ in the set preceding the one of $b$ in $\Sigma$ during the execution of \LexBFS; as before, we see that this is the vertex which is minimal w.r.t.\ $\leq_\sigma$ in $S(v_2)$.
    
    We can now continue to construct this sequence $v_{i+1} := \min S(v_i)$ (always taking the minimum w.r.t.\ $\leq_\sigma$) until we find some vertex $v_m < b$; this is a vertex in $C \cup \{a\}$.  Even more, $v_m \in C \setminus N$, since, by definition of $S(v_{m-1})$, we only consider vertices that are not adjacent to $b$.  We now have constructed a path $\gamma = (v_1, \ldots, v_m)$ of length $m \geq 2$ in $G$ such that $v_1 \in P$, $v_i \notin \nb_G(b) \spforall i > 1$, $v_i > b \spforall i < m$ and $v_m \in C \setminus N$; furthermore, we have $v_m <_\sigma \ldots <_\sigma v_1 <_\sigma b$.  The path $\gamma$ can be elongated to a cycle $(a, v_0 := b, v_1, v_2, \ldots, v_m, a)$:
    \begin{center}
        \begin{tikzpicture}
            \tikzstyle{every node}=[anchor = base]
            \node (a) at (120:\exgredge) {$a$};
            \node (b) at (60:\exgredge) {$b$};
            \node (v1) at (0:\exgredge) {$v_1$};
            \node (vm-1) at (240:\exgredge) {$v_{m-1}$};
            \node (vm) at (180:\exgredge) {$v_m$};
            \draw (v1) -- (b);
            \draw (a.mid east) -- (b.mid west);
            \draw (a) -- (vm) -- (vm-1);
            \draw[densely dashed] (vm-1) -- (300:\exgredge) -- (v1);
            \draw[densely dotted] (1.7\exgredge, 0.4\exgredge)  -- (1.4\exgredge, 0.5\exgredge);
            \draw (1.4\exgredge, 0.5\exgredge) .. controls (0.5\exgredge, 0.7\exgredge) and (0.5\exgredge, -0.7\exgredge) .. (1.4\exgredge, -0.5\exgredge);
            \draw[densely dotted] (1.4\exgredge, -0.5\exgredge) -- (1.7\exgredge, -0.4\exgredge);
            \node (P-N) at (1.7\exgredge, 0.6\exgredge) {$P$};
            \draw[densely dotted] (-1.7\exgredge, 0.4\exgredge)  -- (-1.4\exgredge, 0.5\exgredge);
            \draw (-1.4\exgredge, 0.5\exgredge) .. controls (-0.5\exgredge, 0.7\exgredge) and (-0.5\exgredge, -0.7\exgredge) .. (-1.4\exgredge, -0.5\exgredge);
            \draw[densely dotted] (-1.4\exgredge, -0.5\exgredge) -- (-1.7\exgredge, -0.4\exgredge);
            \node (P-N) at (-1.7\exgredge, 0.5\exgredge) {$C$};
        \end{tikzpicture}
    \end{center}
    We now claim that $v_i \grline a \in G$ for all $0 \leq i \leq m$.  This is clearly the case for $i = 0$ and $i = m$ by construction.  Assume, for the sake of contradiction, that there is some $i$, $0 < i < m$, that is not adjacent to $a$.  Let $r$ be the \emph{largest} index smaller than $i$ such that $v_r \grline a \in G$ and $s$ be the \emph{smallest} index larger than $i$ such that $v_s \grline a \in G$.  Then the following is an \emph{induced} subgraph of $G$:
    \begin{center}
        \begin{tikzpicture}
            \tikzstyle{every node}=[anchor = base]
            \node (a) at (90:\exgredge) {$a$};
            \node (vr) at (0:\exgredge) {$v_r$};
            \node (vr+1) at (300:\exgredge) {$v_{r+1}$};
            \node (vs) at (180:\exgredge) {$v_s$};
            \draw (vs) -- (a) -- (vr) -- (vr+1);
            \draw[densely dashed] (vr+1) -- (240:\exgredge) -- (vs);
        \end{tikzpicture}
    \end{center}
    Note that a chord between different $v_l$'s, say, a chord of the form $v_l \grline v_{l + h}$ with $h \geq 2$, would violate the minimality of $v_{l+1}$ in the set $S(v_l)$.  This means that $G$ contains an induced cycle of length $4$ or more, contradicting the chordality.
    
    This proves the claim that $v_i \grline a \in G$ for all $0 \leq i \leq m$, or, in other words, $v_i \in \nb_G(a)$ for all $0 \leq i \leq m$.  Hence $v_1 \in N \setminus C$, and $\gamma$ is a path from $N \setminus C$ to $C \setminus N$ in $G[\nb_G(a)]$ that has no vertex in $C \cap N$, in contradiction with the assumption.
\end{proof}

\begin{proposition}
    \label{prop:chain-graph-orientation}
    Let $G = (V, E)$ be a chain graph with chordal chain components that does not contain $a \grarright b \grline c$ as an induced subgraph, and let $D \subset G$ be a digraph with $D^u = G^u$.  $D$ is acyclic and has the same v-structures as $G$ if and only if $D[T]$ is oriented according to a perfect elimination ordering for each chain component $T \in \mathbf{T}(G)$.
\end{proposition}

\begin{proof}
    \emph{``$\Rightarrow$''}: let $T \in \mathbf{T}(G)$.  $G[T]$ obviously does not have any v-structures, hence $D[T]$ has no v-structures, either.  It follows from Lemma \ref{lem:dag-perfect-ordering} that $D[T]$ must be oriented according to a perfect elimination ordering.
    
    \proofitem{``$\Leftarrow$''} for each $T \in \mathbf{T}(G)$, $D[T]$ is acyclic by construction.  Assume that $D$ has some directed cycle $\gamma$; this cycle must reach different chain components of $G$, so it contains at least one edge $a \grarright b$ that is also present in $G$.  Because of $D \subset G$ and $D^u = G^u$, $\gamma$ is also a cycle in $G$; and since $a \grarright b \in G$, it is even a \emph{directed} cycle in $G$, a contradiction.  So $D$ is acyclic.
    
    By construction, every v-structure in $G$ is also present in $D$.  Suppose that $D$ has some v-structure $a \grarright b \grarleft c$ that $G$ has not.  $a$, $b$ and $c$ cannot belong to the same chain component of $G$ according to Lemma \ref{lem:dag-perfect-ordering}.  So, w.l.o.g., $a \grarright b \grline c$ must be an induced subgraph of $G$, contradicting the assumption.  Hence $D$ and $G$ have the same v-structures.
\end{proof}

\section{Proofs}
\label{sec:proofs}

In this appendix, the technically interested reader finds all proofs that were left out in Sections \ref{sec:model} to \ref{sec:greedy-search} for better readability.

\subsection{Proofs for Section \ref{sec:model}}
\label{sec:proofs-model}

We start with the proof of Lemma \ref{lem:intervention-densities-motivation} which motivates Definition \ref{def:intervention-densities} by showing that, for some DAG $D$ and some (conservative) family of targets \mcI, the elements of $\mathcal{M_I}(D)$ are exactly the density tuples that can be realized as interventional densities of a causal model with structure $D$.  Note that we use the conservativeness of \mcI{} only in the proof of point \ref{itm:interventions-realize-definition}; it can even be proven without assuming conservativeness, although the proof becomes harder.

\begin{proof}[Lemma \ref{lem:intervention-densities-motivation}]
    \begin{subprop}
        \item $f(x | \doop(X_I = U_I))$ obeys the Markov property of $D^{(I)}$ (Section \ref{sec:causal-calculus}).  Furthermore, for $I, J \in \mcI$ and $a \notin I \cup J$, we have
        $$
            f(x_a \spst x_{\pa_D(a)}; \doop(X_I = U_I)) = f(x_a \spst x_{\pa_D(a)}) = f(x_a \spst x_{\pa_D(a)}; \doop(X_J = U_J))
        $$
        by the truncated factorization of Equation (\ref{eqn:interventional-density}).
        
        \item Let $a \in [p]$.  Since \mcI{} is conservative, there is some $I \in \mcI$ such that $a \notin I$.  Define $h_a(x_a, x_{\pa_D(a)}) := f^{(I)}(x_a | x_{\pa_D(a)})$.  Note that, due to Definition \ref{def:intervention-densities}, the function $h_a$ does \emph{not} depend on the choice of $I$.
        
        Let $f(x) := \prod_{a = 1}^p h_a(x_a, x_{\pa_D(a)})$; this is a positive density on $\mathcal{X}$ with $f(x_a | x_{\pa_D(a)}) = h_a(x_a, x_{\pa_D(a)})$, hence $f \in \mathcal{M}(D)$ and $(D, f)$ is a causal model.
        
        By defining level densities $\tilde{f}_I(x_I) := \prod_{i \in I} f^{(I)}(x_i)$, we can construct an intervention setting $\mathcal{S} := \{(I, \tilde{f}_I)\}_{I \in \mcI}$ with the requested properties.
    \end{subprop}
    \proofnegspace
\end{proof}

The proof of the main result of Section \ref{sec:model}, the graph theoretic criterion for two DAGs being interventionally Markov equivalent (Theorem \ref{thm:interventional-markov-equivalence}), requires additional lemmas.

\begin{lemma}
    \label{lem:intervention-densities-projection}
    Let $D$ be a DAG, \mcI{} a family of targets and $I \in \mcI$ a target in this family.  Define
    $$
        \mathcal{M}^{(I)}(D) := \{f^{(I)} \spst (f^{(J)})_{J \in \mcI} \in \mathcal{M_I}(D)\}\ ,
    $$
    the projection of $\mathcal{M_I}(D)$ to the density component associated with the intervention target $I$.  Then, $\mathcal{M}^{(I)}(D) = \mathcal{M}(D^{(I)})$.
\end{lemma}

\begin{proof}
    The inclusion ``$\subset$'' is immediately clear from Definition \ref{def:intervention-densities}.  It remains to show ``$\supset$''.
    
    Let $f \in \mathcal{M}(D^{(I)})$.  Since $D^{(I)} \subset D$, $f$ also obeys the Markov property of $D$; this means $f \in \mathcal{M}(D)$.  Set $\tilde{f}_I(x_I) := f(x_I)$; since $f \in \mathcal{M}(D^{(I)})$, the components of $\tilde{f}_I$ are independent.  For $J \in \mcI$, $J \ne I$, let $\tilde{f}_J$ be an arbitrary level density on $\mathcal{X}_J$.  By Lemma \ref{lem:intervention-densities-motivation}\ref{itm:interventions-meet-definition}, we know that, for intervention variables $U_J \sim \tilde{f}_J$ ($J \in \mcI$),
    $$
        \big( f(\cdot \spst \doop_D(X_J = U_J)) \big)_{J \in \mcI} \in \mathcal{M_I}(D) \ ,
    $$
    hence $f(\cdot \spst \doop_D(X_I = U_I)) \in \mathcal{M}^{(I)}(D)$ by definition of $\mathcal{M}^{(I)}(D)$.  Moreover, by construction of $\tilde{f}_I$, we have $f(x \spst \doop_D(X_I = U_I)) = f(x)$ and hence $f \in \mathcal{M}^{(I)}(D)$.
\end{proof}

\begin{lemma}
    \label{lem:conditional-factorization}
    Let $D$ be a DAG, $f \in \mathcal{M}(D)$, and $A \subset [p]$.  Then,
    $$
        \prod_{a \in A} f(x_a \spst x_{\pa(a)}) = f(x_A \spst x_{\pa(A)}).
    $$
\end{lemma}

\begin{proof}
    Let $\sigma \in S_p$ be a topological ordering of $D$.  Then, for $a \in A$,
    \begin{equation}
        \label{eqn:parents-subset}
        \pa(a) \subset \pa(A) \cup \left[A \cap \sigma^{-1}(\{1, \ldots, a - 1\}) \right]
    \end{equation}
    holds: every $b \in \pa(a)$ either lies in $A^c$ and hence in $\pa(A)$ by the definition given in Appendix \ref{sec:graphs-notation}, or in $A \cap \sigma^{-1}(\{1, \ldots, a - 1\})$ by the definition of a topological ordering.
    
    Hence we conclude
    $$
        f(x_A \spst x_{\pa(A)}) = \prod_{a \in A} f(x_a \spst x_{A \cap \sigma^{-1}(\{1, \ldots, a - 1\})}, x_{\pa(A)}) = \prod_{a \in A} f(x_a \spst x_{\pa(a)});
    $$
    the first equality is the usual factorization of a density, the second equality follows from the Markov properties of $f$ and Equation (\ref{eqn:parents-subset}).
\end{proof}

\begin{lemma}
    \label{lem:conserved-arrow}
    Let \mcI{} be a family of targets.  Assume $D_1$ and $D_2$ are DAGs with the same skeleton and the same v-structures such that $D_1^{(I)}$ and $D_2^{(I)}$ have the same skeleton for all $I \in \mcI$.  Moreover, let $a \grarright b \in D_1$.  If there is some $I \in \mcI$ such that $|I \cap \{a, b\}| = 1$, then the arrow is also present in $D_2$: $a \grarright b \in D_2$.
\end{lemma}

\begin{proof}
    Since $D_1$ and $D_2$ have the same skeleton, we have at least $a \grdots b \in D_2$.  Suppose $a \grarleft b \in D_2$.  If $a \in I$, $b \notin I$, $a$ and $b$ are adjacent in $D_1^{(I)}$, but not in $D_2^{(I)}$, hence $D_1^{(I)}$ and $D_2^{(I)}$ have a different skeleton, a contradiction.  On the other hand, if $a \notin I$ but $b \in I$, $a$ and $b$ are not adjacent in $D_1^{(I)}$, but in $D_2^{(I)}$, a contradiction, too.
\end{proof}
\vspace{-1mm}
\begin{proof}[Theorem \ref{thm:interventional-markov-equivalence}]
    \emph{\ref{itm:interventional-markov-equivalence} $\Rightarrow$ \ref{itm:interventional-densities-equal}:} Let $I \in \mcI$, and let $\mathcal{M}^{(I)}(D_1)$ and  $\mathcal{M}^{(I)}(D_2)$ be defined as in Lemma \ref{lem:intervention-densities-projection}.  By Definition \ref{def:interventional-markov-equivalence} of interventional Markov equivalence, it follows that $\mathcal{M}^{(I)}(D_1) = \mathcal{M}^{(I)}(D_2)$; hence $\mathcal{M}(D_1^{(I)}) = \mathcal{M}(D_2^{(I)})$ by Lemma \ref{lem:intervention-densities-projection}.
    
    \proofitem{\ref{itm:interventional-densities-equal} $\Rightarrow$ \ref{itm:interventional-dags-equivalent}} this implication follows from Theorem \ref{thm:markov-equivalence}.
    
    \proofitem{\ref{itm:interventional-dags-equivalent} $\Rightarrow$ \ref{itm:skeleton-v-structures}} Let $a \grarright b \in D_1$ be an arrow.  Since \mcI{} is conservative, there is some $I \in \mcI$ such that $b \notin I$.  For this $I$, $a \grarright b \in D_1^{(I)}$, so $a \grdots b \in D_2^{(I)}$ by assumption and hence $a \grdots b \in D_2$ because of $D_2^{(I)} \subset D_2$.  Similarly, we can show the implication $a \grarright b \in D_2 \ \Rightarrow \ a \grdots b \in D_1$, what proves that $D_1$ and $D_2$ have the same skeleton.
    
    It remains to show that $D_1$ and $D_2$ also have the same v-structures.  Let $a \grarright b \grarleft c$ be a v-structure of $D_1$.  There is some $I \in \mathcal{I}$ that does not contain $b$; $a \grarright b \grarleft c$ is then an induced subgraph of $D_1^{(I)}$ and hence by assumption also of $D_2^{(I)}$.  By consequence, $a \grarright b \grarleft c$ is also an induced subgraph of $D_2$ since $D_2$ has the same skeleton as $D_1$.  The argument is of course symmetric w.r.t.\ exchanging $D_1$ and $D_2$.
    
    \proofitem{\ref{itm:skeleton-v-structures} $\Rightarrow$ \ref{itm:interventional-markov-equivalence}} Let $(f^{(I)})_{I \in \mcI} \in \mathcal{M_I}(D_1)$.  By Lemma \ref{lem:intervention-densities-motivation}\ref{itm:interventions-realize-definition}, there is some density $f \in \mathcal{M}(D_1)$ and some intervention setting $\mathcal{S} = \{(I, \tilde{f}_I)\}_{I \in \mcI}$ such that $f^{(I)}(\cdot) = f(\cdot | \doop_{D_1}(X_I = U_I))$ for random variables $U_I \sim \tilde{f}_I$, $I \in \mcI$.
    
    The truncated factorization in Equation (\ref{eqn:interventional-density}) tells us
    \begin{eqnarray}
       f(x \spst \doop_{D_1}(X_I = U_I)) & = & \prod_{a \notin I} f(x_a \spst x_{\pa_{D_1}(a)}) \prod_{a \in I} \tilde{f}_I(x_a) = f(x) \prod_{a \in I} \frac{\tilde{f}_I(x_a)}{f(x_a \spst x_{\pa_{D_1}(a)})} \nonumber \\
       & = & f(x) \frac{\tilde{f}_I(x_I)}{f(x_I \spst x_{\pa_{D_1}(I)})}. \label{eqn:interventional-factorization}
    \end{eqnarray}
    The last step uses Lemma \ref{lem:conditional-factorization}.
    
    We now claim that $\pa_{D_1}(I) = \pa_{D_2}(I)$.  Indeed, if $b \in I$ and $a \in \pa_{D_1}(b) \setminus I$, $a \grarright b$ is an arrow in $D_1$ with $|I \cap \{a, b\}| = 1$, hence $a \grarright b \in D_2$ by Lemma \ref{lem:conserved-arrow} and therefore $a \in \pa_{D_2}(I)$; the argument is symmetric w.r.t.\ exchanging $D_1$ and $D_2$.  It follows that $f(x_I | x_{\pa_{D_1}(I)}) = f(x_I | x_{\pa_{D_2}(I)})$, and by repeating the calculation in (\ref{eqn:interventional-factorization}) for $D_2$ instead of $D_1$, we find $f(x | \doop_{D_1}(X_I = U_I)) = f(x | \doop_{D_2}(X_I = U_I))$.
    
    Since this equality is true for all $I \in \mcI$, we have $f^{(I)}(\cdot) = f(\cdot | \doop_{D_2}(X_I = U_I))$ for all $I \in \mcI$, so $(f^{(I)})_{I \in \mcI} \in \mathcal{M_I}(D_2)$ by Lemma \ref{lem:intervention-densities-motivation}\ref{itm:interventions-meet-definition}, which proves $\mathcal{M_I}(D_1) \subset \mathcal{M_I}(D_2)$.  The other direction is completely analogous.
\end{proof}

Points \ref{itm:interventional-markov-equivalence} to \ref{itm:interventional-dags-equivalent} are even equivalent under non-conservative families of targets.  The proof is more difficult in this case though.

\subsection{Proofs for Section \ref{sec:essential-graphs}}
\label{sec:proofs-essential-graphs}

All statements of Section \ref{sec:essential-graphs-characterization} are similar to analogous statements for the observational case developed by \citet{Andersson1997Characterization}.  Some of the proofs given there are even literally valid also for our interventional setting; in such cases, we will not repeat them here, but just refer to the original ones.  However, in most cases, the generalization from the observational to the interventional case is not obvious and requires adapted techniques presented in this section.  Here, \mcI{} always stands for a conservative family of targets.

First, we show that for some DAG $D$, $\mathcal{E_I}(D)$ is a chain graph (Proposition \ref{prop:I-essential-chain-graph}).  For that purpose, we define $\mathcal{E_I}(D)^*$ as the smallest chain graph containing $\mathcal{E_I}(D)$.  $\mathcal{E_I}(D)^*$ is obtained from $\mathcal{E_I}(D)$ by converting all arrows that are part of a directed cycle in $\mathcal{E_I}(D)$ into lines \citep{Andersson1997Characterization}.  We first state a couple of properties of $\mathcal{E_I}(D)$ and $\mathcal{E_I}(D)^*$ (Lemma \ref{lem:I-essential-graph-properties}), and then show that $\mathcal{E_I}(D)^* = \mathcal{E_I}(D)$ (Proposition \ref{prop:I-essential-chain-graph}).

\begin{lemma}[adapted from \citealp{Andersson1997Characterization}]
    \label{lem:I-essential-graph-properties}
    Let $D$ be a DAG.  Then:
    \begin{subprop}
        \item \label{itm:a->b--c-essential} $\mathcal{E_I}(D)$ has no induced subgraph of the form $a \grarright b \grline c$.
        
        \item \label{itm:three-vertex-subgraph} If $\mathcal{E_I}(D)$ has an induced subgraph of the form
        $$
            \threegraph{$a$}{->}{$b$}{-}{$c$}{-},
        $$
        then there exist $D_1, D_2 \in [D]_\mathcal{I}$ such that
        $$
            \threegraph{$a$}{->}{$b$}{->}{$c$}{<-} \subset D_1, \quad \threegraph{$a$}{->}{$b$}{<-}{$c$}{->} \subset D_2.
        $$
        
        \item \label{itm:essential-v-structures} $\mathcal{E_I}(D)^*$ has the same v-structures as $D$ (and hence as $\mathcal{E_I}(D)$).
        
        \item \label{itm:chordless-cycles} $\mathcal{E_I}(D)$ and $\mathcal{E_I}(D)^*$ do not have any undirected chordless $k$-cycle of length $k \geq 4$.
        
        \item \label{itm:a->b--c-chain-graph} $\mathcal{E_I}(D)^*$ has no induced subgraph of the form $a \grarright b \grline c$.
        
        \item \label{itm:intervention-directed} If two vertices $a$ and $b$ are adjacent in $\mathcal{E_I}(D)^*$ and there is some $I \in \mathcal{I}$ such that $|I \cap \{a, b\}| = 1$, then the edge between $a$ and $b$ is directed in $\mathcal{E_I}(D)$ and $\mathcal{E_I}(D)^*$.
    \end{subprop}
\end{lemma}

\begin{proof}
    Points \ref{itm:a->b--c-essential} to \ref{itm:a->b--c-chain-graph} correspond to Facts 1 to 5 of \citet{Andersson1997Characterization} where these properties were proven for observational essential graphs.  A thorough inspection of the proofs given there reveals that they only make use of the fact that two Markov equivalent DAGs have the same skeleton and the same v-structures, which is also true in the interventional case by Theorem \ref{thm:interventional-markov-equivalence}.  Thanks to this, the proofs of \citet{Andersson1997Characterization} can be literally used here.  (Note that the inverse implication also holds in the observational case, but not in the interventional one; see the discussion after Theorem \ref{thm:interventional-markov-equivalence}.)
    
    It remains to prove point \ref{itm:intervention-directed}.  The edge between $a$ and $b$ in $\mathcal{E_I}(D)$ is directed since the arrow between $a$ and $b$ is $\mathcal{I}$-essential in $D$ by Corollary \ref{cor:intervention-essential}.  It remains to show that the edge is also directed in $\mathcal{E_I}(D)^*$, that is, to show that it is \emph{not} part of a directed cycle in $\mathcal{E_I}(D)$.
    
    Let's suppose, for the sake of contradiction, that the edge between $a$ and $b$ is part of a directed cycle $\gamma = (a, b \equiv b_0, b_1, \ldots, b_k \equiv a)$ in $\mathcal{E_I}(D)$.  W.l.o.g., we can assume that $a \grarright b \in \mathcal{E_I}(D)$, and that $\gamma$ is the \emph{shortest} such cycle containing a directed edge with one end point in $I$ and the other one outside $I$.
    
    \proofitem{Case 1} $k = 2$.  Then $\gamma$ is of the form
    $$
        \threegraph{$a$}{->}{$b$}{-}{$b_1$}{-}
    $$
    since two or three directed edges would imply the existence of a digraph with a cycle in the equivalence class of $D$.  By point \ref{itm:three-vertex-subgraph}, there are DAGs $D_1$ and $D_2$ in $[D]_\mathcal{I}$ such that
    $$
         \threegraph{$a$}{->}{$b$}{->}{$b_1$}{<-} \subset D_1 \ , \quad
         \threegraph{$a$}{->}{$b$}{<-}{$b_1$}{->} \subset D_2.
    $$
    The condition $|I \cap \{a, b\}| = 1$ leaves four possibilities:
    \begin{enumerate}[label=\emph{\alph*})]
        \item $a \in I; b, b_1 \notin I$: then, \threegraph{$a$}{->}{$b$}{->}{$b_1$}{<-} $\subset D_1^{(I)}$ \ , \quad \threegraph{$a$}{->}{$b$}{<-}{$b_1$}{} $\subset D_2^{(I)}$ \ ;
        
        \item $a, b_1 \in I; b \notin I$: then, \threegraph{$a$}{->}{$b$}{}{$b_1$}{} $\subset D_1^{(I)}$ \ , \quad \threegraph{$a$}{->}{$b$}{<-}{$b_1$}{} $\subset D_2^{(I)}$ \ ;
        
        \item $b \in I; a, b_1 \notin I$: then, \threegraph{$a$}{}{$b$}{->}{$b_1$}{<-} $\subset D_1^{(I)}$ \ , \quad \threegraph{$a$}{}{$b$}{}{$b_1$}{->} $\subset D_2^{(I)}$ \ ;
        
        \item $b, b_1 \in I; a \notin I$: then, \threegraph{$a$}{}{$b$}{}{$b_1$}{} $\subset D_1^{(I)}$ \ , \quad \threegraph{$a$}{}{$b$}{}{$b_1$}{->} $\subset D_2^{(I)}$ \ .
    \end{enumerate}
    In all four cases, $(D_1^{(I)})^u \neq (D_2^{(I)})^u$, hence $D_1 \not\sim_\mathcal{I} D_2$, a contradiction.
    
    \proofitem{Case 2} $k \geq 3$.  Let $i$ be the smallest index such that $b_i \grline b_{i+1} \in \mathcal{E_I}(D)$ (there must be such an index, otherwise $\gamma$ would be a directed cycle in $D$).
    
    \proofitem{Case 2.1} $i = 0$.  Since $a \grarright b \grline b_1$ cannot be an induced subgraph of $\mathcal{E_I}(D)$ by point \ref{itm:a->b--c-essential}, we must have $a \grdots b_1 \in \mathcal{E_I}(D)$.  More precisely, we must have $a \grarright b_1 \in \mathcal{E_I}(D)$, otherwise $(a, b, b_1, a)$ would form a shorter directed cycle than $\gamma$, in contradiction to the assumption.  This means that there exist DAGs $D_1, D_2 \in [D]_\mathcal{I}$ such that
    $$
        \threegraph{$a$}{->}{$b$}{->}{$b_1$}{<-} \subset D_1, \quad \threegraph{$a$}{->}{$b$}{<-}{$b_1$}{<-} \subset D_2.
    $$
    Again, the condition $|I \cap \{a, b\}| = 1$ leaves four possibilities:
    \begin{enumerate}[label=\emph{\alph*})]
        \item $a \in I; b, b_1 \notin I$: then, \threegraph{$a$}{->}{$b$}{->}{$b_1$}{<-} $\subset D_1^{(I)}$ \ , \quad \threegraph{$a$}{->}{$b$}{<-}{$b_1$}{<-} $\subset D_2^{(I)}$ \ ;
        
        \item $a, b_1 \in I; b \notin I$: then, \threegraph{$a$}{->}{$b$}{}{$b_1$}{} $\subset D_1^{(I)}$ \ , \quad \threegraph{$a$}{->}{$b$}{<-}{$b_1$}{} $\subset D_2^{(I)}$ \ ;
        
        \item $b \in I; a, b_1 \notin I$: then, \threegraph{$a$}{}{$b$}{->}{$b_1$}{<-} $\subset D_1^{(I)}$ \ , \quad \threegraph{$a$}{}{$b$}{}{$b_1$}{<-} $\subset D_2^{(I)}$ \ ;
        
        \item $b, b_1 \in I; a \notin I$: then, \threegraph{$a$}{}{$b$}{}{$b_1$}{} $\subset D_1^{(I)}$ \ , \quad \threegraph{$a$}{}{$b$}{}{$b_1$}{} $\subset D_2^{(I)}$ \ .
    \end{enumerate}
    Cases \emph{b)} and \emph{c)} are not compatible with the condition $(D_1^{(I)})^u = (D_2^{(I)})^u$.  In cases \emph{a)} and \emph{d)}, the arrow $a \grarright b_1$ is part of a directed cycle $(a, b_1, b_2, \ldots, b_k \equiv a)$, furthermore $|I \cap \{a, b_1\}| = 1$; this contradicts the assumption of minimality of the larger cycle $\gamma$.
    
    \proofitem{Case 2.2} $i \geq 1$.  Since $b_{i-1} \grarright b_i \grline b_{i+1}$ cannot be an induced subgraph of $\mathcal{E_I}(D)$, we must have $b_{i-1} \grdots b_{i+1} \in \mathcal{E_I}(D)$.  Either $b_{i-1} \grarleft b_{i+1} \in \mathcal{E_I}(D)$, that is
    $$
        \threegraph{$b_{i-1}$}{->}{$b_i$}{-}{$b_{i+1}$}{->} \subset \mathcal{E_I}(D) \ ,
    $$
    which would imply the existence of a digraph with a directed 3-cycle in the equivalence class of $D$, a contradiction.  The other cases are $b_{i-1} \grarright b_{i+1} \in \mathcal{E_I}(D)$ or $b_{i-1} \grline b_{i+1} \in \mathcal{E_I}(D)$ which would mean that $a \grarright b$ would be part of a shorter directed cycle $(a, b \equiv b_0, \ldots, b_{i-1}, b_{i+1}, \ldots, b_k \equiv a)$, contradicting the assumption of minimality of the cycle $\gamma$.
\end{proof}

\begin{proof}[Proposition \ref{prop:I-essential-chain-graph}]
    We only prove the first point; the second one is an immediate consequence of  Lemma \ref{lem:I-essential-graph-properties}\ref{itm:chordless-cycles}.  We have to show that $\mathcal{E_I}(D) = \mathcal{E_I}(D)^*$, that means that
    $$
        a \grline b \in \mathcal{E_I}(D)^* \ \Rightarrow \ a \grline b \in \mathcal{E_I}(D) \ .
    $$
    By Lemma \ref{lem:I-essential-graph-properties}\ref{itm:chordless-cycles}, all chain components of $\mathcal{E_I}(D)^*$ are chordal.  Let $D_1$ and $D_2$ be two DAGs that are obtained by orienting all chain components of $\mathcal{E_I}(D)^*$ according to some perfect elimination ordering, such that $a \grarright b \in D_1$ and $a \grarleft b \in D_2$; such orientations exist by Proposition \ref{prop:chain-graph-orientation} and Corollary \ref{cor:chordal-graph-different-orientations}.
    
    We now claim that $D_1 \sim_\mathcal{I} D_2$ by verifying the criteria of Theorem \ref{thm:interventional-markov-equivalence}\ref{itm:skeleton-v-structures}; it then follows that $a \grline b \in \mathcal{E_I}(D)$ because of $D_1 \cup D_2 \subset \mathcal{E_I}(D)$:
    \begin{itemize}
        \item By Proposition \ref{prop:chain-graph-orientation}, $D_1$ and $D_2$ have the same skeleton and the same v-structures.
        
        \item $D_1^{(I)}$ and $D_2^{(I)}$ have the same skeleton for all $I \in \mathcal{I}$: suppose, for the sake of contradiction, that $(D_1^{(I)})^u$ has some edge $c \grline d$ that $(D_2^{(I)})^u$ has not.  W.l.o.g., we then have $c \grarright d \in D_1$, $c \grarleft d \in D_2$, $c \in I$, $d \notin I$.  But then $c$ and $d$ are adjacent in $\mathcal{E_I}(D)^*$ with $|I \cap \{c, d\}| = 1$, hence the edge between $c$ and $d$ must be oriented in $\mathcal{E_I}(D)^*$ by point \ref{itm:intervention-directed} of Lemma \ref{lem:I-essential-graph-properties}, and hence it is not possible that this edge has two different orientations in $D_1$ and $D_2$ by their construction.
    \end{itemize}
    \proofnegspace
\end{proof}

\begin{proof}[Proposition \ref{prop:construction-representative}]
    \emph{``$\Leftarrow$'':} By the construction of $\mathcal{E_I}(D)$, we know that $D \subset \mathcal{E_I}(D)$ and $D^u = \mathcal{E_I}(D)^u$.  Furthermore, $D$ has the same v-structures as $\mathcal{E_I}(D)$.  Let $D'$ be another digraph that is obtained by orienting all chain components of $\mathcal{E_I}(D)$ according to a perfect elimination ordering; by Proposition \ref{prop:chain-graph-orientation}, $D'$ is acyclic and has the same v-structures as $\mathcal{E_I}(D)$ and hence as $D$.  It remains to show that $D^{(I)}$ and $D'^{(I)}$ have the same skeleton for all $I \in \mathcal{I}$; this can be done similarly to the proof of Proposition \ref{prop:I-essential-chain-graph}.
        
    \proofitem{``$\Rightarrow$''} let $D'$ be a DAG with $D' \sim_\mathcal{I} D$.  In particular, $D'$ and $D$ have the same skeleton and the same v-structures, so $D'$ also has the same skeleton and the same v-structures as $\mathcal{E_I}(D)$.  It follows, with Proposition \ref{prop:chain-graph-orientation}, that $D'$ is oriented according to a perfect elimination ordering on all chain components of $\mathcal{E_I}(D)$.
\end{proof}

\begin{lemma}
    \label{lem:essential-protected}
    Let $D$ be a DAG and $a \grarright b$ an \mcI-essential arrow in $D$.  Then $a \grarright b$ is strongly \mcI-protected in $\mathcal{E_I}(D)$.
\end{lemma}
This lemma is an auxiliary result needed to prove Theorem \ref{thm:essential-graph-characterization}.  In its proof, we first show the weaker statement that every \mcI-essential arrow of $D$ is \mcI-protected in $\mathcal{E_I}(D)$.

\begin{definition}[Protection]
    \label{def:protected-arrow} 
    Let $G$ be a graph.  An arrow $a \grarright b \in G$ is \textbf{\mcI-protected} in $G$ if there is some intervention target $I \in \mcI$ such that $|I \cap \{a, b\}| = 1$, or $\pa_G(a) \neq \pa_G(b) \setminus \{a\}$.
\end{definition}
This definition is again a generalization of the notion of protection of \citet{Andersson1997Characterization}; for $\mcI = \{\emptyset\}$, we gain back their definition.  A \emph{strongly} $\mathcal{I}$-protected arrow (Definition \ref{def:strongly-protected-arrow}) is also \mcI-protected.  In a chain graph $G$, an arrow $a \grarright b$ is \mcI-protected if and only if there is some $I \in \mcI$ such that $|I \cap \{a, b\}| = 1$, or the arrow $a \grarright b$ occurs in at least one subgraph of the form (a), (b), (c) in the notation of Definition \ref{def:strongly-protected-arrow}, or in a subgraph of the form (d') \citep{Andersson1997Characterization}, where
$$
\text{(d'):} \threegraph{$a$}{->}{$b$}{<-}{$c$}{-} \ .
$$

\begin{proof}[Lemma \ref{lem:essential-protected}]
    As foreshadowed, we prove this lemma in two steps, corresponding to Facts 6 and 7 of \citet{Andersson1997Characterization}: in a first step, we show that $a \grarright b$ must be $\mathcal{I}$-protected, in a second step, we strengthen the result by showing that it must even be \emph{strongly} $\mathcal{I}$-protected.  For notational convenience, we abbreviate $G := \mathcal{E_I}(D)$.  We skip some steps of the proof that can be literally copied from proofs in \citet{Andersson1997Characterization}.
    
    Suppose, for the sake of contradiction, that $a \grarright b$ is not \mcI-protected.  Let $D_1$ be a digraph that is gained by orienting all chain components of $G$ according to a perfect elimination ordering, where the edges of $T_G(a)$ and $T_G(b)$ are oriented such that all edges point away from $a$ or $b$, respectively.  Then $D_1$ is acyclic and \mcI-equivalent to $D$ by Proposition \ref{prop:construction-representative}.
    
    Let $D_2$ be another digraph, differing from $D_1$ only by the orientation of the edge between $a$ and $b$.  It can be shown that $D_2$ is acyclic too \citep[proof of Fact 6]{Andersson1997Characterization}.  We now claim that $D_1 \sim_\mcI D_2$:
    \begin{itemize}
        \item $D_1$ and $D_2$ clearly have the same skeleton.
        
        \item $D_1$ and $D_2$ have the same v-structures.  Otherwise, there would be some v-structure $c \grarright a \grarleft b$ in $D_2$, or some v-structure $a \grarright b \grarleft c$ in $D_1$.  In both cases, this would imply $\pa_G(a) \neq \pa_G(b) \setminus \{a\}$, contradicting the assumption: in the first case, $c \notin T_G(a)$ by construction (all edges of $T_G(a)$ point away from $a$ in $D_2$), so $c \in \pa_G(a)$, but $c \notin \pa_G(b)$; in the second case, $c \in \pa_G(b)$, but $c \notin \pa_G(a)$ by analogous arguments.
        
        \item $(D_1^{(I)})^u = (D_2^{(I)})^u$ for all $I \in \mcI$.  Otherwise, there would be some $I \in \mcI$ such that the skeletons of $D_1^{(I)}$ and $D_2^{(I)}$ differ in the edge between $a$ and $b$.  This could only happen if $|I \cap \{a, b\}| = 1$, in contradiction with the assumption.
    \end{itemize}
    Hence, since $D_1, D_2 \in [D]_\mcI$, we have $D_1 \cup D_2 \subset G$ and thus $a \grline b \in G$, a contradiction.  This proves that $a \grarright b$ is \mcI-protected in $G$.
    
    In the second step, we show that $a \grarright b$ is even \emph{strongly} \mcI-protected.  If this was not the case, $a \grarright b$ would occur in configuration (d') in $G$, but \emph{not} in configuration (a), (b), (c) or (d) (see the comment following Definition \ref{def:protected-arrow}).  Define $P_a := \{d \in T_G(a) \spst d \grarright b \in G\}$.  It can be shown that $P_a$ is a clique $G[T_G(a)]$ \citep[proof of Fact 7]{Andersson1997Characterization}.
    
    Let $D_1$ be the DAG that we get by orienting all chain components of $G$ according to a perfect elimination ordering, such that, additionally,
    \begin{itemize}
        \item all edges of $D_1[T_G(b)]$ point away from $b$,
        \item all edges of $D_1[P_a]$ point towards $a$,
        \item and all other edges of $D_1[T_G(a)]$ point away from $a$.
    \end{itemize}
    Such an orientation exists by Corollary \ref{cor:lex-bfs-clique}.  Let $D_2$ be the digraph that we get by changing the orientation of the edge $a \grarright b$ in $D_1$; as in the first part, it can be shown that $D_2$ is acyclic \citep[proof of Fact 7]{Andersson1997Characterization}.  Again, we claim that $D_1 \sim_\mcI D_2$:
    \begin{itemize}
        \item $D_1$ and $D_2$ clearly have the same skeleton.
        
        \item $D_1$ and $D_2$ have the same v-structures.  Otherwise, there would be some v-structure $d \grarright a \grarleft b$ in $D_2$, or a v-structure $a \grarright b \grarleft d$ in $D_1$.  In the first case, $d \notin P_a$ (otherwise, $d \grarright b \in G$ by definition of $P_a$, and hence $d \grarright b \in D_2$ since $D_2 \subset G$), and $d \notin T_G(a) \setminus P_a$ by construction (edges in $T_G(a) \setminus P_a$ point away from $a$ in $D_2$), hence $d \grarright a \in G$ and $a \grarright b$ is in configuration (a) in $G$; in the second case, $d \notin T_G(b)$ by construction (all edges of $T_G(b)$ point away from $b$ in $D_1$), so $a \grarright b$ is in configuration (b) (notation of Definition \ref{def:strongly-protected-arrow}) in $G$.  Both cases contradict the assumption.
        
        \item Exactly as in the first part, $(D_1^{(I)})^u = (D_2^{(I)})^u$ for all $I \in \mcI$.
    \end{itemize}
    We can conclude that, since $D_1, D_2 \in [D]_\mcI$, $D_1 \cup D_2 \subset G$, so $a \grline b \in G$, a contradiction.
\end{proof}

\begin{proof}[Theorem \ref{thm:essential-graph-characterization}]
    \emph{``$\Rightarrow$'':} \ref{itm:chain-graph} and \ref{itm:chordal-chain-components} follow from Proposition \ref{prop:I-essential-chain-graph}, \ref{itm:forbidden-subgraph} from Lemma \ref{lem:I-essential-graph-properties}\ref{itm:a->b--c-chain-graph}, \ref{itm:forbidden-edge} from Corollary \ref{cor:intervention-essential} and \ref{itm:strongly-protected-arrows} from Lemma \ref{lem:essential-protected}.
    
    \proofitem{``$\Leftarrow$''} Consider the set $\mathbf{D}(G)$ of all DAGs that can be obtained by orienting the chain components of $G$ according to a perfect elimination ordering; we have $\bigcup \mathbf{D}(G) \subset G$.  On the other hand, for each undirected edge $a \grline b \in G$, there are DAGs $D_1$ and $D_2$ in $\mathbf{D}(G)$ such that $a \grarright b \in D_1$, $a \grarleft b \in D_2$ (Corollary \ref{cor:chordal-graph-different-orientations}), hence $G \subset \bigcup \mathbf{D}(G)$.  Together, we find $G = \bigcup \mathbf{D}(G)$.
    
    We claim that $D_1 \sim_\mcI D_2$ for any two DAGs $D_1, D_2 \in \mathbf{D}(G)$:
    \begin{itemize}
        \item $D_1$ and $D_2$ have the same skeleton and the same v-structures by Proposition \ref{prop:chain-graph-orientation}.
        
        \item $(D_1^{(I)})^u = (D_2^{(I)})^u$ for all $I \in \mcI$.  Otherwise, there would be arrows $a \grarright b \in D_1$, $a \grarleft b \in D_2$, and some $I \in \mcI$ such that $|I \cap \{a, b\}| = 1$; this would mean that $a \grline b \in G$ although $|I \cap \{a, b\}| = 1$, contradicting property \ref{itm:forbidden-edge}.
    \end{itemize}
    Let $D \in \mathbf{D}(G)$.  We have shown that $\mathbf{D}(G) \subset [D]_\mcI$, hence $G = \bigcup \mathbf{D}(G) \subset \mathcal{E_I}(D)$.  It remains to show that $G \supset \mathcal{E_I}(D)$.
    
    Assume, for the sake of contradiction, that $G$ has some arrow $a \grarright b$ where $\mathcal{E_I}(D)$ has an undirected edge $a \grline b$.  According to property \ref{itm:strongly-protected-arrows}, $a \grarright b$ is strongly \mcI-protected in $G$.  If there was some $I \in \mcI$ such that  $|I \cap \{a, b\}| = 1$, the edge between $a$ and $b$ was also directed in $\mathcal{E_I}(D)$ by Corollary \ref{cor:intervention-essential}, a contradiction.  Hence $a \grarright b$ occurs in $G$ in one of the configurations depicted in Definition \ref{def:strongly-protected-arrow}.  Exactly as in the proof of Theorem 4.1 of \citet{Andersson1997Characterization}, we can construct a contradiction for each of the four configurations.  Although the proof given there can be used literally, we reproduce it here since since we will use the following steps again in the proof of Lemma \ref{lem:maximal-partial-essential-graph}.

    We assume w.l.o.g.\ that $T_G(a)$ is minimal in
    $$
        A := \{T \in \mathbf{T}(G) | \spexists a \in T, b \in V(G): a \grarright b \in G, a \grline b \in \mathcal{E_I}(D)\}
    $$
    w.r.t.\ $\preceq_G$, and that $T_G(b)$ is minimal in
    $$
        B := \{T \in \mathbf{T}(G) | \spexists a \in T(a), b \in T: a \grarright b \in G, a \grline b \in \mathcal{E_I}(D)\}.
    $$
    Each configuration (a) to (d) of Definition \ref{def:strongly-protected-arrow} leads to a contradiction ($c$, $c_1$ and $c_2$ denote the vertices involved in the respective configuration):
    \begin{enumerate}[label=(\alph*)]
        \item Because of the minimality of $T_G(a)$, $c \grarright a$ must be oriented in $\mathcal{E_I}(D)$, hence $c \grarright a \grline b$ is an induced subgraph of $\mathcal{E_I}(D)$, contradicting Lemma \ref{lem:I-essential-graph-properties}\ref{itm:a->b--c-essential}.
        
        \item $a \grarright b \grarleft c$ is then a v-structure in $D$, hence it is also a v-structure in $\mathcal{E_I}(D)$, that means $a \grarright b \in \mathcal{E_I}(D)$, a contradiction.
        
        \item Because of the minimality of $T_G(b)$, the edge between $a$ and $c$ must be oriented in $\mathcal{E_I}(D)$, so the vertices $a$, $b$ and $c$ are in one of the following configurations in $\mathcal{E_I}(D)$:
        $$
            \threegraph{$a$}{-}{$b$}{<-}{$c$}{<-}, \quad \threegraph{$a$}{-}{$b$}{-}{$c$}{<-}.
        $$
        Both possibilities violate Proposition \ref{prop:I-essential-chain-graph}\ref{itm:essential-graph-chain-graph}.
        
        \item The v-structure $c_1 \grarright b \grarleft c_2$ of $D$ is also a v-structure of $\mathcal{E_I}(D)$, hence $\mathcal{E_I}(D)$ has two directed 3-cycles $(c_1, b, a, c_1)$ and $(c_2, b, a, c_2)$, a contradiction.
    \end{enumerate}
    \proofnegspace
\end{proof}

\begin{proof}[Lemma \ref{lem:partial-essential-graph-properties}]
    \begin{subprop}
        \item This immediately follows from Theorem \ref{thm:essential-graph-characterization}\ref{itm:forbidden-subgraph}.
        
        \item Let $a \grarright b$ be an arrow in $\mathcal{E_I}(D)$; by Theorem \ref{thm:essential-graph-characterization}\ref{itm:strongly-protected-arrows}, it is strongly \mcI-protected in $\mathcal{E_I}(D)$.  If there is some $I \in \mcI$ such that $|I \cap \{a, b\}| = 1$, the arrow is by definition also strongly \mcI-protected in $G$.  Otherwise, $a \grarright b$ occurs in one of the configurations (a) to (d) of Definition \ref{def:strongly-protected-arrow} in $\mathcal{E_I}(D)$.  In configurations (a) to (c), the other arrows involved ($a \grarleft c$; $c \grarright b$; or $a \grarright c$ and $c \grarright b$, resp.) are also present in $G$, hence $a \grarright b$ is strongly \mcI-protected in $G$ by the same configuration as in $\mathcal{E_I}(D)$.
        
        \begin{table}[b]
            \centering
            \begin{tabular}{c|c|c}
                \multicolumn{2}{c|}{Induced subgraph of $\{a, c_1, c_2\}$\dots} & Configuration \\
                \dots in $D$ & \dots in $G$ & of $a \grarright b$ in $G$ \\
                \hline
                $c_1 \grarleft a \grarleft c_2$  & $c_1 \grarleft a \grarleft  c_2$ & (c) \\
                                                    & $c_1 \grline   a \grarleft  c_2$ & --- \\
                                                    & $c_1 \grarleft a \grline    c_2$ & (c) \\
                                                    & $c_1 \grline   a \grline    c_2$ & (d) \\
                \hline
                $c_1 \grarleft a \grarright c_2$ & $c_1 \grarleft a \grarright c_2$ & (c) \\
                                                    & $c_1 \grline   a \grarright c_2$ & (c) \\
                                                    & $c_1 \grarleft a \grline    c_2$ & (c) \\
                                                    & $c_1 \grline   a \grline    c_2$ & (d) \\
            \end{tabular}
            \caption{Possible configurations for the vertices $\{a, c_1, c_2\}$ in the proof of Lemma \ref{lem:partial-essential-graph-properties}\ref{itm:partial-essential-protected}.  The labels in the last column refer to the configurations of Definition \ref{def:strongly-protected-arrow}.}
            \label{tab:D-G-configurations}
        \end{table}
        
        It remains to show that if $a \grarright b$ is in configuration (d) in $\mathcal{E_I}(D)$, it is also strongly \mcI-protected in $G$.  In $D$, the vertices $\{a, c_1, c_2\}$ as defined in Definition \ref{def:strongly-protected-arrow} can occur in one of the following configurations:
        $$
            c_1 \grarleft a \grarleft c_2, \quad c_1 \grarleft a \grarright c_2, \quad c_1 \grarright a \grarright c_2.
        $$
        The first and the third case are symmetric w.r.t.\ exchanging $c_1$ and $c_2$, hence we only consider the first two.  Table \ref{tab:D-G-configurations} lists all possible configurations for the vertices $\{a, c_1, c_2\}$ in the graph $G$ according to the condition $D \subset G \subset \mathcal{E_I}(D)$.  There is only one possibility for the arrow $a \grarright b$ not to occur in one of the configurations (a) to (d) of Definition \ref{def:strongly-protected-arrow}, and hence not being strongly \mcI-protected in $G$; however, the corresponding subgraph of $\{a, c_1, c_2\}$, $c_1 \grline a \grarleft c_2$, is forbidden by Definition \ref{def:partial-essential-graph}.
        
        \item According to Theorem \ref{thm:interventional-markov-equivalence}, we have to check the following properties:
        \begin{itemize}
            \item $D_1$ and $D_2$ have the same skeleton, namely $D_1^u = D_2^u = G^u$.
            
            \item $D_1$ and $D_2$ have the same v-structures: let $a \grarright b \grarleft c$ be a v-structure in $D_1$.  This v-structure is then also present in $\mathcal{E_I}(D_1)$.  Because of $D_2 \subset G \subset \mathcal{E_I}(D_1)$, we find it also in $G$ and in $D_2$.  The argument is completely symmetric w.r.t.\ exchanging $D_1$ and $D_2$.
            
            \item For all $I \in \mcI$, $D_1^{(I)}$ and $D_2^{(I)}$ have the same skeleton: assume, for the sake of contradiction, that there is some $I \in \mcI$ and an edge $a \grline b$ that is present in $(D_1^{(I)})^u$, but not in $(D_2^{(I)})^u$.  W.l.o.g., we can assume that $a \grarright b \in D_1$, $a \grarleft b \in D_2$, $a \in I$, $b \notin I$.  Because of Theorem \ref{thm:essential-graph-characterization}\ref{itm:forbidden-edge}, we then have $a \grarright b \in \mathcal{E_I}(D_1)$ and $a \grarleft b \in \mathcal{E_I}(D_2)$; however, this is not compatible with the requirements $G \subset \mathcal{E_I}(D_1)$ and $G \subset \mathcal{E_I}(D_2)$.
        \end{itemize}
    \end{subprop}
    \proofnegspace
\end{proof}

\begin{proof}[Lemma \ref{lem:iterative-arrow-replacement}]
    If $a \grarright b \in \mathcal{E_I}(D)$, it would be strongly \mcI-protected by Theorem \ref{thm:essential-graph-characterization}\ref{itm:strongly-protected-arrows}, and hence also strongly \mcI-protected in $G$ by Lemma \ref{lem:partial-essential-graph-properties}\ref{itm:partial-essential-protected}, contradicting the assumption.  Therefore, $a \grline b \in \mathcal{E_I}(D)$ and hence $D \subset G' \subset \mathcal{E_I}(D)$.
    
    Suppose that $G'$ contains an induced subgraph of the form $c \grarright d \grline e$.  Since $G$ does not contain such an induced subgraph, it must be of the form $c \grarright a \grline b$ or $c \grarright b \grline a$ in $G'$.  In both cases, $a \grarright b$ is then strongly \mcI-protected in $G$, either by configuration (a) or (b), a contradiction.
\end{proof}

\begin{proof}[Lemma \ref{lem:maximal-partial-essential-graph}]
    Let $D \subset G \subset \mathcal{E_I}(D)$ be a partial \mcI-essential graph that only has strongly \mcI-protected arrows.  We can literally use the second part of the proof of Theorem \ref{thm:essential-graph-characterization} to show $G \supset \mathcal{E_I}(D)$; there, we only used the fact that every arrow in $G$ is strongly \mcI-protected.
\end{proof}

\subsection{Proofs for Section \ref{sec:greedy-search}}
\label{sec:proofs-greedy-search}

\begin{proof}[Proposition \ref{prop:forward-characterization}]
    \emph{``$\Rightarrow$'':}
    \begin{subprop}
        \item This claim follows from Corollary \ref{cor:lex-bfs-clique}.
        
        \item Suppose that there is some vertex $a \in N \setminus C$, that is a vertex $a \in N$ with $a \grarleft v \in D$.  $D'$ would have a directed cycle if $u \grarleft a \in D$, so $u \grarright a \in D$.  But then, $u \grarright a \grarleft v$ is a v-structure in $D$, hence also in $G$, and consequently $a \notin \nb_G(v)$, a contradiction.
        
        \item Assume that $\gamma = (v \equiv a_0, a_1, \ldots, a_k \equiv u)$ is a shortest path from $v$ to $u$ in $G$ that does not intersect with $C$.  We claim that $\gamma$ is a directed path in $D$, which means that $D'$ has a directed cycle, a contradiction.
        
        Suppose that the claim is wrong, and let $a_i \grarleft a_{i+1} \in D$ be the first edge of (the chain) $\gamma$ that points away from $u$ in $D$; $i \geq 1$ holds by the assumption that, in particular, $a_1 \notin C$.  $a_{i-1} \grarright a_i \grarleft a_{i+1}$ cannot be an induced subgraph of $D$, otherwise it would also be present in $G$ and hence $\gamma$ would not be a \emph{path} in $G$.  Hence $a_{i-1} \grdots a_{i+1} \in G$; more precisely, $a_{i-1} \grarleft a_{i+1} \in G$ (and hence also in $D$), otherwise there would be a shorter path from $v$ to $u$ in $G$ than $\gamma$ that does not intersect with $C$. Because $\gamma$ is a path in $G$, $a_{i-1}$, $a_i$ and $a_{i+1}$ can occur in $G$ only in one of the following configurations:
        $$
            \threegraph{$a_{i-1}$}{<-}{$a_{i+1}$}{-}{$a_i$}{<-}, \quad \threegraph{$a_{i-1}$}{<-}{$a_{i+1}$}{-}{$a_i$}{-}.
        $$
        However, both graphs cannot be an induced subgraph of the chain graph $G$.
    \end{subprop}

    \proofitem{``$\Leftarrow$''} Since $C$ is a clique in $G[T_G(v)]$, there is a DAG $D \in \mathbf{D}(G)$ with $\{a \in \nb_G(v) \spst a \grarright v \in D\} = C$ by Proposition \ref{prop:chain-graph-orientation} and Corollary \ref{cor:lex-bfs-clique}.  It remains to show that $D'$ is a DAG.
    
    Assume, for the sake of contradiction, that $D'$ has a directed cycle going through $u \grarright v$.  The return path from $v$ to $u$, $\gamma = (v \equiv a_0, a_1, \ldots, a_k \equiv u)$, must come from a path in $G$ and must therefore, by assumption, contain a vertex $a_i \in C$ ($i \geq 2$).  Since $a_i \grarright v \in D$ by construction, this means that $D$ has a directed cycle $(a_0, a_1, \ldots, a_i, a_0)$, a contradiction.
    
    \proofitem{Uniqueness of $\mathcal{E_I}(D')$} Let $D_1, D_2 \in \mathbf{D}(G)$ with $\{a \in \nb_G(v) \spst a \grarright v \in D_1\} = \{a \in \nb_G(v) \spst a \grarright v \in D_2\} = C$, and set $D'_i := D_i + (u, v)$, $i = 1, 2$; we assume that $D'_1, D'_2 \in \mathbf{D}^+(G)$.  To prove $D'_1 \sim_\mcI D'_2$, we have to check the following three points according to Theorem \ref{thm:interventional-markov-equivalence}\ref{itm:skeleton-v-structures}:
    \begin{itemize}
        \item $D'_1$ and $D'_2$ obviously have the same skeleton.
        
        \item $D'_1$ and $D'_2$ have the same v-structures.  We already know that $D_1$ and $D_2$ have the same v-structures.  Let's assume, for the sake of contradiction, that (w.l.o.g.) $D'_1$ has a v-structure $u \grarright v \grarleft a$ that $D'_2$ has not.  In $G$, we must then have a line $a \grline v$, hence $a \in \nb_G(v)$.  However, the arrow between $a$ and $v$ would then have the same orientation in $D_1$ and $D_2$ by construction, a contradiction.
        
        \item For all $I \in \mcI$, $D'^{(I)}_1$ and $D'^{(I)}_2$ have the same skeleton.  If this was not the case, there would be some vertices $a, b \in [p]$ and some $I \in \mcI$ such that $a \grarright b \in D'_1$, $a \grarleft b \in D'_2$ and $|I \cap \{a, b\}| = 1$.  The arrow $u \grarright v$ is part of $D'_1$ and $D'_2$ by construction, so the arrows between $a$ and $b$ must be present in $D_1$ and $D_2$; however, $D_1^{(I)}$ and $D_2^{(I)}$ would then not have the same skeleton, a contradiction.
    \end{itemize}
    \proofnegspace
\end{proof}

Corollary \ref{cor:forward-score-change} is an immediate consequence of Proposition \ref{prop:forward-characterization} and the fact that we assume the score function to be decomposable, so we skip the proof here.
\vspace{3mm}

\begin{proof}[Lemma \ref{lem:forward-partial-essential-graph}]
    Obviously, we have $D' \subset H$.  To show $H \subset \mathcal{E_I}(D')$, we look at some edge $a \grline b \in G$ with $a, b \notin T_G(v)$ and show that $a \grline b \in \mathcal{E_I}(D')$.  W.l.o.g., we can assume that $a \grarright b \in D$.  By Corollary \ref{cor:chordal-graph-different-orientations}, there exists a $D_2 \in \mathbf{D}(G)$ that has the same orientation of edges in $T_G(v)$, but an orientation of edges in $T_G(a)$ such that $a \grarleft b \in D_2$.  By Proposition \ref{prop:forward-characterization}, we know that $D'_2 := D_2 + (u, v)$ is \mcI-equivalent to $D'$, so in particular $a \grline b \in D' \cup D'_2 \subset \mathcal{E_I}(D')$.
    
    It remains to show that $a \grarright b \grline c$ does not occur as an induced subgraph of $H$.  The inserted arrow $u \grarright v$ cannot be part of such a subgraph, since all other edges incident to $v$ are oriented in $H$ by construction.  Since $G$ has no such subgraph either (Theorem \ref{thm:essential-graph-characterization}), it could only appear in $H$ through one of the newly oriented edges of $T_G(v)$.  This means that if $H$ had an induced subgraph of the form $a \grarright b \grline c$, the corresponding vertices would be in configuration $a \grline b \grline c$ in $G$; however, $c \in T_G(v)$ then, and so the edge between $b$ and $c$ would be oriented in $H$, a contradiction.
\end{proof}

\begin{proof}[Proposition \ref{prop:backward-characterization}]
    \emph{``$\Rightarrow$'':}
    \begin{subprop}
        \item By Corollary \ref{cor:lex-bfs-clique}, $\{a \in \nb_G(v) \spst a \grarright v \in D\}$ is a clique, hence every subset---in particular, $C$---is a clique, too.
        
        \item Assume that there is some $a \in C \setminus \ad_G(u)$; then $u \in \nb_G(v)$, otherwise $u \grarright v \grline a$ would be an induced subgraph of $G$.  Nevertheless, $a \in C$ means that $u \grarright v \grarleft a$ is a v-structure in $D$, which should hence also be present in $G$.
    \end{subprop}

    \proofitem{``$\Leftarrow$''}  We only must prove the existence of the claimed $D \in \mathbf{D}(G)$, see the comment in the beginning of Section \ref{sec:gies-backward}.  We distinguish two cases:
    \begin{itemize}
        \item $u \grarright v \in G$.  The existence of the DAG $D \in \mathbf{D}(G)$ with the requested properties follows from Corollary \ref{cor:lex-bfs-clique}.
        
        \item $u \grline v \in G$, hence $u \grline a \in G$ for all $a \in N$ because $G$ is a chain graph.  Therefore, $C \cup \{u\}$ is a clique in $G[\nb_G(v)]$, and the existence of the claimed $D$ again follows from Corollary \ref{cor:lex-bfs-clique}.
    \end{itemize}
    
    \proofitem{Uniqueness of $\mathcal{E_I}(D')$} Let $D_1, D_2 \in \mathbf{D}(G)$ with $u \grarright v \in D_1, D_2$ and $\{a \in \nb_G(v) \setminus \{u\} \spst a \grarright v \in D_1\} = \{a \in \nb_G(v) \setminus \{u\} \spst a \grarright v \in D_2\} = C$, and set $D'_i := D_i - (u, v)$, $i = 1, 2$.  To prove $D'_1 \sim_\mcI D'_2$, we have to check the following three points according to Theorem \ref{thm:interventional-markov-equivalence}\ref{itm:skeleton-v-structures}:
    \begin{itemize}
        \item $D'_1$ and $D'_2$ have the same skeleton, namely $G^u - (u, v) - (v, u)$.
        
        \item $D'_1$ and $D'_2$ have the same v-structures.  Otherwise, w.l.o.g., $D'_1$ would have a v-structure $a \grarright b \grarleft c$ that $D'_2$ has not.  $D_1$ and $D_2$ have the same v-structures, so $a \grarright b \grarleft c$ is no induced subgraph of $D_1$; this implies $a = u$, $c = v$.  Since $D'_2$ does not have the v-structure $u \grarright b \grarleft v$, the vertices $u$, $b$ and $v$ must occur in configuration $u \grarright b \grarright v$ or $u \grarleft b \grarright v$ in $D'_2$ (the configuration $u \grarleft b \grarleft v$ is not consistent with the acyclicity of $D_2$).  However, all edges incident to $v$ must have the same orientation in $D'_1$ and $D'_2$ by construction, a contradiction.
        
        \item Let $I \in \mcI$.  Because of $(D_1^{(I)})^u = (D_2^{(I)})^u$ and $(D'^{(I)}_i)^u = (D_i^{(I)})^u - (u, v) - (v, u)$ for $i = 1, 2$, we have $(D'^{(I)}_1)^u = (D'^{(I)}_2)^u$.
    \end{itemize}
    \proofnegspace
\end{proof}

Corollary \ref{cor:backward-score-change} follows quickly from Proposition \ref{prop:backward-characterization}, and the proof of Lemma \ref{lem:backward-partial-essential-graph} is very similar to that of Lemma \ref{lem:forward-partial-essential-graph}.  Therefore we skip both proofs here and proceed with the proofs of Section \ref{sec:gies-turning}.
\vspace{3mm}

\begin{proof}[Proposition \ref{prop:turning-undirected-characterization}]
    Note that we can write $N = \nb_G(v) \cap \ad_G(u) = \nb_G(v) \cap \nb_G(u)$ because $u \grline v \in G$ and $G$ is a chain graph.
    
    \proofitem{``$\Rightarrow$''}
    \begin{subprop}
        \item This follows from Corollary \ref{cor:lex-bfs-clique}.

        \item $D$ and $D'$ have the same skeleton; the same is true for $D^{(I)}$ and $D'^{(I)}$ for all $I \in \mcI$.  To see the latter, assume that for some $I \in \mcI$, the intervention graphs $D^{(I)}$ and $D'^{(I)}$ have a different skeleton.  Since $D$ and $D'$ only differ in the orientation of the arrow between $u$ and $v$, the skeletons of $D^{(I)}$ and $D'^{(I)}$ can only differ in that $u$ and $v$ are adjacent in one of them and not adjacent in the other one.  However, this would imply that $|I \cap \{u, v\}| = 1$, and hence the edge between $u$ and $v$ would be directed in $G$ by Theorem \ref{thm:essential-graph-characterization}\ref{itm:forbidden-edge}, contradicting the assumption of the proposition.  Finally, $D'$ has at least all v-structures that $D$ has by construction.
        
        As a consequence $D' \not\sim_\mcI D$ if and only if $D'$ has \emph{more} v-structures than $D$ (Theorem \ref{thm:interventional-markov-equivalence}).  An additional v-structure in $D'$ must be of the form $u \grarright v \grarleft a$.  The edge between $v$ and $a$ cannot be directed in $G$, otherwise $u \grline v \grarleft a$ would be an induced subgraph of $G$, which is forbidden by Theorem \ref{thm:essential-graph-characterization}\ref{itm:forbidden-subgraph}.  Hence $a \in \nb_G(v)$, or, more precisely, $a \in C \setminus N$.
        
        \item If $N \setminus C$ is empty, the statement is trivial.  Otherwise, assume that there is some shortest path $\gamma = (a_0, a_1, \ldots, a_k)$ from $N \setminus C$ to $C \setminus N$ in $G[\nb_G(v)]$ that has no vertex in $C \cap N$.
        
        By definition of $C$, $a_k \grarright v \in D$; furthermore, $u \grarright a_0 \in D$ must hold, otherwise $(v, a_0, u, v)$ would be a directed cycle in $D'$.  Therefore, $\gamma$ must not be a path from $a_0$ to $a_k$ in $D$.  Let $a_i \grarleft a_{i+1}$ be the first arrow in $\gamma$ that points away from $a_k$ in $D$.  If $i = 0$, $u \grarright a_0 \grarleft a_1$ would be a v-structure in $D$ since $a_1 \notin N$: by assumption, $a_1 \notin N \cap C$, and $a_1 \notin N \setminus C$ because of the minimality of $\gamma$.  Hence $i > 0$ (and $k > 1$) must hold, and $a_{i-1} \grdots a_{i+1}$ in $D$ and $G$, otherwise there would be a v-structure in $D$.  However, $\gamma$ is not the \emph{shortest} path with the requested properties then, a contradiction.
    \end{subprop}
    
    \proofitem{``$\Leftarrow$''} From Proposition \ref{prop:lex-bfs-clique-neighborhood}, we see that there exists a DAG $D$ that has the requested properties, and in which, in addition, $\{a \in \nb_G(u) \spst a \grarright u \in D\} = (C \cap N) \cup \{v\}$ (point \ref{itm:orientation-neighbors-b} of Proposition \ref{prop:lex-bfs-clique-neighborhood}).  The fact that $D' := D - (v, u) + (u, v) \not\sim_\mcI D$ can be seen by an argument very similar to the proof of point \ref{itm:turning-undirected-diff-C-N} above; it remains to show that $D$ has no $v$-$u$-path except $(v, u)$.  Suppose that $\gamma = (a_0 \equiv v, a_1, \ldots, a_k \equiv u)$, $k \geq 2$, is such a path.  In particular, $\gamma$ is then also a $v$-$u$-path in $G$, hence $\gamma$ lies completely in $T_G(v)$.
    
    If $k = 2$, then $a_1 \in N$, and so the vertices $u$, $v$ and $a_1$ occur in one of the following configurations in $D$ by Proposition \ref{prop:lex-bfs-clique-neighborhood}:
    $$
        \threegraph{$v$}{->}{$u$}{->}{$a_1$}{<-}, \quad \threegraph{$v$}{->}{$u$}{<-}{$a_1$}{->} \ .
    $$
    Both configurations contradict the assumption that $\gamma = (v, a_1, u)$ forms a path in $D$.  Thus we conclude $k \geq 3$, and we notice $a_{k-1} \in \nb_G(u) \setminus \{v\}$.  If $a_{k-1} \in C$, $a_{k-1} \grarright v \in D$, hence $(a_0, a_1, \ldots, a_{k-1}, a_0)$ would be a cycle in $D$.  On the other hand, if $a_{k-1} \notin C$, we would have $a_{k-1} \grarleft u \in D$, so $\gamma$ would not be a path in $D$.
    
    \proofitem{Uniqueness of $\mathcal{E_I}(D')$} Let $D_1, D_2 \in \mathbf{D}(G)$ with $u \grarleft v \in D_1, D_2$ and $\{a \in \nb_G(v) \spst a \grarright v \in D_1\} = \{a \in \nb_G(v) \spst a \grarright v \in D_2\} = C$, and set $D'_i := D_i - (v, u) + (u, v)$, $i = 1, 2$; we assume that $D'_1, D'_2 \in \mathbf{D}^\circlearrowleft(G)$.  As in the proofs of Proposition \ref{prop:forward-characterization}, we can check that $D'_1 \sim_\mcI D'_2$:
    \begin{itemize}
        \item $D'_1$ and $D'_2$ obviously have the same skeleton.
        
        \item $D_1$ and $D_2$ have the same v-structures.  If this does not hold for $D'_1$ and $D'_2$, (w.l.o.g.) $D'_1$ must have a v-structure $u \grarright v \grarleft a$ that $D'_2$ has not.  Since $u \grline v \grarleft a$ cannot be an induced subgraph of $G$, $a \in \nb_G(v)$; however, the edges between $v$ and its neighbors are oriented in the same way in $D'_1$ and $D'_2$ by construction, a contradiction.
        
        \item For all $I \in \mcI$, $D'^{(I)}_1$ and $D'^{(I)}_2$ have the same skeleton: this can be seen by an argument very similar to that in the proof of Proposition \ref{prop:forward-characterization}.
    \end{itemize}
    \proofnegspace
\end{proof}

\begin{proof}[Corollary \ref{cor:turning-undirected-score-change}]
    We have to show $\pa_D(v) = \pa_G(v) \cup C$ and $\pa_D(u) = \pa_G(u) \cup (C \cap N) \cup \{v\}$.  The first identity is immediately clear.  For the second identity, note that for any vertex $a \in C \cap N$, the arrow between $a$ and $u$ must be oriented as $a \grarright u \in D$ because the other orientation would induce a 3-cycle.  On the other hand, we have $a \grarleft u \in D$ for $a \in N \setminus C$ because a different orientation would induce a 3-cycle in $D'$.  Finally, we also have $a \grarleft u \in D$ for any $a \in \nb_G(u) \setminus (\nb_G(v) \cup \{v\})$ since the other orientation would induce a v-structure $v \grarright u \grarleft a$ in $D$.
\end{proof}

Lemma \ref{lem:turning-undirected-partial-essential-graph} can be proven very similarly as Lemma \ref{lem:forward-partial-essential-graph}.  Finally, we finish this proof section with the proof of Proposition \ref{prop:turning-directed-characterization} characterizing a step of the turning phase of GIES for the case that we turn an \mcI-essential arrow in some representative $D \in \mathbf{D}(G)$.  We will omit the proof of Lemma \ref{lem:turning-directed-partial-essential-graph} since it can be proven similarly to Lemma \ref{lem:forward-partial-essential-graph}.
\vspace{3mm}

\begin{proof}[Proposition \ref{prop:turning-directed-characterization}]
    When $v \grarright u \in G$ (that is, $u$ and $v$ lie in different chain components), $N = \nb_G(v) \cap \ad_G(u) = \nb_G(v) \cap \pa_G(u)$ holds because $G$ is a chain graph.
    
    \proofitem{``$\Rightarrow$''}
    \begin{subprop}
        \item This point follows from Corollary \ref{cor:lex-bfs-clique}.
        
        \item If this was not true, $D'$ would have a cycle of the form $(u, v, a, u)$ for some $a \in N$ since $N \subset \pa_G(u)$.
        
        \item Suppose that the path $\gamma = (a_0 \equiv v, a_1, \ldots, a_k \equiv u)$ is a shortest counterexample of a path without vertex in $C \cup \nb_G(u)$.
        
        Assume that $k = 2$.  Since $u$ and $v$ lie in different chain components, the vertices $u$, $v$ and $a_1$ can occur in one of the following configurations in $G$:
        $$
            \threegraph{$v$}{->}{$u$}{<-}{$a_1$}{<-}, \quad \threegraph{$v$}{->}{$u$}{<-}{$a_1$}{-}, \quad \threegraph{$v$}{->}{$u$}{-}{$a_1$}{<-} \ .
        $$
        The first case implies the existence of a directed cycle in $D'$; in the second case, $a_1 \in N \subset C$, in the third case, $a_1 \in \nb_G(u)$.
        
        Therefore $k \geq 3$. In complete analogy to the proof of Proposition \ref{prop:forward-characterization}, we can show that $\gamma$ is also a $v$-$u$-path in $D$, hence $D'$ has a directed cycle, a contradiction.
    \end{subprop}
    
    \proofitem{``$\Leftarrow$''} Let $D \in \mathbf{D}(G)$ be a DAG with $\{a \in \nb_G(v) \spst a \grarright v \in D\} = C$ and in which all edges of $D[T_G(u)]$ point away from $u$; such a DAG exists by Corollary \ref{cor:lex-bfs-clique} and meets the requirements of Proposition \ref{prop:turning-directed-characterization}.  It remains to show that $D'$ is acyclic, that means that $D$ has no $v$-$u$-path except $(v, u)$.
    
    Suppose, for the sake of contradiction, that $D$ has such a path $\gamma = (a_0 \equiv v, a_1, \ldots, a_k \equiv u)$.  $\gamma$ is then also a $v$-$u$-path in $G$, hence there is, by assumption, some $a_i \in C \cup P$.  If $a_i \in C$, $(a_0, a_1, \ldots, a_i, a_0)$ would be a cycle in $D$; on the other hand, if $a_i \in P$, $(a_i, a_{i+1}, \ldots, a_k, a_i)$ would be a cycle in $D$, a contradiction.

    \proofitem{Uniqueness of $\mathcal{E_I}(D')$} The proof given for Proposition \ref{prop:turning-undirected-characterization} is also valid here.
\end{proof}

\begin{proof}[Corollary \ref{cor:turning-directed-score-change}]
    The fact that $\pa_D(v) = \pa_G(v) \cup C$ is clear from Proposition \ref{prop:turning-directed-characterization}; it remains to show that $\pa_D(u) = \pa_G(u)$.  Any neighbor $a$ of $u$ must also be a child of $v$, otherwise $G$ would have a subgraph of the form $v \grarright u \grline a$, which is forbidden by Theorem \ref{thm:essential-graph-characterization}\ref{itm:forbidden-subgraph}.  Hence $a \grarleft u \in D$ for all $a \in \nb_G(u)$ since the other orientation would imply a directed cycle in $D'$.
\end{proof}

\bibliography{dagestimation}

\end{document}